\pdfoutput=1
\documentclass[12pt]{article}
\usepackage{frankstyle}

\begin{document}

\title{Using Invalid Instruments on Purpose: Focused Moment Selection and Averaging for GMM\footnote{I thank Aislinn Bohren, Xu Cheng, Gerda Claeskens, Bruce Hansen, Byunghoon Kang, Toru Kitagawa, Hannes Leeb, Adam McCloskey, Serena Ng, Alexei Onatski, Hashem Pesaran, Benedikt P\"{o}tscher, Frank Schorfheide, Neil Shephard, Richard J.\ Smith, Stephen Thiele, Melvyn Weeks, and seminar participants at 
Brown, Cambridge, Columbia, George Washington, Oxford, Queen Mary, Rutgers, St Andrews, UPenn, Vienna, and the 2011 Econometric Society European Meetings for their many helpful comments and suggestions. I thank Kai Carstensen for providing data for my empirical example.}}

\author{Francis J.\ DiTraglia}

 \date{\footnotesize Final Version: September 1, 2016} 

\maketitle 
\begin{abstract}
In finite samples, the use of a slightly endogenous but highly relevant instrument can reduce mean-squared error (MSE). 
Building on this observation, I propose a novel moment selection procedure for GMM -- the Focused Moment Selection Criterion (FMSC) -- in which moment conditions are chosen not based on their validity but on the MSE of their associated estimator of a user-specified target parameter.
The FMSC mimics the situation faced by an applied researcher who begins with a set of relatively mild ``baseline'' assumptions and must decide whether to impose any of a collection of stronger but more controversial ``suspect'' assumptions.
When the (correctly specified) baseline moment conditions identify the model, the FMSC provides an asymptotically unbiased estimator of asymptotic MSE, allowing us to select over the suspect moment conditions.
I go on to show how the framework used to derive the FMSC can address the problem of inference post-moment selection.
Treating post-selection estimators as a special case of moment-averaging, in which estimators based on different moment sets are given data-dependent weights, I propose  simulation-based procedures for inference that can be applied to a variety of formal and informal moment-selection and averaging procedures.
Both the FMSC and confidence interval procedures perform well in simulations.
I conclude with an empirical example examining the effect of instrument selection on the estimated relationship between malaria and income per capita.


\end{abstract}

\section{Introduction}
In finite samples, the addition of a slightly endogenous but highly relevant instrument can reduce estimator variance by far more than bias is increased. 
Building on this observation, I propose a novel moment selection criterion for generalized method of moments (GMM) estimation: the focused moment selection criterion (FMSC). 
Rather than selecting only valid moment conditions, the FMSC chooses from a set of potentially mis-specified moment conditions based on the asymptotic mean squared error (AMSE) of their associated GMM estimators of a user-specified scalar target parameter $\mu$.
To ensure a meaningful bias-variance tradeoff in the limit, I employ a drifting asymptotic framework in which mis-specification, while present for any fixed sample size, vanishes asymptotically.
In the presence of such \emph{locally mis-specified} moment conditions, GMM remains consistent although, centered and rescaled, its limiting distribution displays an asymptotic bias. Adding an additional mis-specified moment condition introduces a further source of bias while reducing asymptotic variance. 
The idea behind the FMSC is to trade off these two effects in the limit as an approximation to finite sample behavior.\footnote{When finite-sample MSE is undefined, AMSE comparisons remain meaningful: see Online Appendix \ref{append:trim}.}
I suppose that two blocks of moment conditions are available: one that is assumed correctly specified, and another that may not be.
This mimics the situation faced by an applied researcher who begins with a ``baseline'' set of relatively mild maintained assumptions and must decide whether to impose any of a collection of stronger but also more controversial ``suspect'' assumptions.
When the (correctly specified) baseline moment conditions identify the model, the FMSC provides an asymptotically unbiased estimator of AMSE, allowing us select over the suspect moment conditions.\footnote{When this is not the case, it remains possible to use the AMSE framework to carry out a sensitivity analysis: see Online Appendix \ref{sec:digress}.}

The primary goal of the FMSC is to select estimators with low AMSE, but researchers typically wish to report confidence intervals along with parameter estimates.
Unfortunately the usual procedures for constructing asymptotic confidence intervals for GMM fail when applied to estimators chosen using a moment selection procedure.
A ``na\"{i}ve'' 95\% confidence interval constructed from the familiar textbook formula will generally under-cover: it will contain the true parameter value far less than 95\% of the time because it fails to account for the additional sampling uncertainty that comes from choosing an estimator based on the data. 
To address the challenging problem of inference post-moment selection, I continue under the local mis-specification framework to derive the limit distribution of ``moment average estimators,'' data-dependent weighted averages of estimators based on different moment conditions.
These estimators are interesting in their own right and include post-moment selection estimators as a special case.
I propose two simulation-based procedures for constructing confidence intervals for moment average and post-selection estimators, including the FMSC.
First is a ``2-Step'' confidence interval.
I prove that this interval guarantees asymptotically valid inference: the asymptotic coverage of a nominal $100 \times (1 - \alpha)\%$ interval cannot fall below this level.
The price of valid inference, however, is conservatism: the actual coverage of the 2-Step interval typically exceeds its nominal level.\footnote{This is unavoidable given certain impossibility results concerning post-selection inference. See, e.g.\ \cite{LeebPoetscher2005}.}
As a compromise between the conservatism of the 2-Step interval and the severe under-coverage of the na\"{i}ve interval I go on to propose a ``1-Step'' confidence interval. 
This interval is easier to compute than its 2-Step counterpart and performs well in empirically relevant examples, as I show both theoretically and in simulations below.
The 1-Step interval is far shorter than the corresponding 2-Step interval and, while it can under-cover, the magnitude of the size distortion is modest compared to that of the na\"{i}ve intervals typically reported in applied work.  

While my methods apply to general GMM models, I focus on two simple but empirically relevant examples: choosing between ordinary least squares (OLS) and two-stage least squares (TSLS) estimators, and selecting instruments in linear instrumental variables (IV) models. 
In the OLS versus TSLS example the FMSC takes a particularly transparent form, providing a risk-based justification for the Durbin-Hausman-Wu test, and leading to a novel ``minimum-AMSE'' averaging estimator that combines OLS and TSLS.
The FMSC, averaging estimator, and related confidence interval procedures work well in practice, as I demonstrate in a series of simulation experiments and an empirical example from development economics.

The FMSC and minimum-AMSE averaging estimator considered here are derived for a scalar parameter interest, as this is the most common situation encountered in applied work.\footnote{For an extension of the FMSC to vector target parameters, see Online Appendix \ref{append:mult}.}
As a consequence, Stein-type results do not apply: it is impossible to construct an estimator with uniformly lower risk than the ``valid'' estimator that uses only the baseline moment conditions.
Nevertheless, as my simulation results show, selection and averaging can substantially outperform the valid estimator over large regions of the parameter space, particularly when the ``suspect'' moment conditions are highly informative and \emph{nearly} correct.
This is precisely the situation for which the FMSC is intended.

My approach to moment selection is inspired by the focused information criterion of \citet{ClaeskensHjort2003}, a model selection criterion for maximum likelihood estimation. 
Like \citet{ClaeskensHjort2003}, I study AMSE-based selection under mis-specification in a drifting asymptotic framework. 
In contradistinction, however, I consider moment rather than model selection, and general GMM rather than maximum likelihood estimation.
\cite{Schorfheide2005} uses a similar approach to select over forecasts constructed from mis-specified vector autoregression models, developed independently of the FIC. 
While the use of locally mis-specified moment conditions dates back at least as far as \cite{Newey1985}, the idea of using this framework for AMSE-based moment selection, however, is novel.

The existing literature on moment selection primarily aims to consistently select all correctly specified moment conditions while eliminating all invalid ones\footnote{Under the local mis-specification asymptotics considered below, consistent moment selection criteria simply choose \emph{all} available moment conditions. For details, see Theorem \ref{pro:andrews}.}
This idea begins with \cite{Andrews1999} and is extended by  \cite{AndrewsLu} and \cite{HongPrestonShum}.
More recently, \cite{Liao} proposes a shrinkage procedure for consistent GMM moment selection and estimation. 
In a similar vein, \cite{CanerHanLee} extend and generalize earlier work by \cite{Caner2009} on LASSO-type model selection for GMM to carry out simultaneous model and moment selection via an adaptive elastic net penalty. 
Whereas these proposals examine only the validity of the moment conditions under consideration, the FMSC balances validity against relevance to minimize AMSE.
Although \cite{HallPeixe2003} and \cite{ChengLiao} do consider relevance, their aim is to avoid including redundant moment conditions after consistently eliminating invalid ones.
Some other papers that propose choosing, or combining, instruments to minimize MSE include \cite{DonaldNewey2001}, \cite{DonaldImbensNewey2009}, and \cite{KuersteinerOkui2010}.
Unlike the FMSC, however, these papers consider the \emph{higher-order} bias that arises from including many valid instruments rather than the first-order bias that arises from the use of invalid instruments.

Another distinguishing feature of the FMSC is focus: rather than a one-size-fits-all criterion, the FMSC is really a method of constructing application-specific moment selection criteria.
Consider, for example, a dynamic panel model.
If your target parameter is a long-run effect while mine is a contemporaneous effect, there is no reason to suppose \emph{a priori} that we should use the same moment conditions in estimation, even if we share the same model and dataset.
The FMSC explicitly takes this difference of research goals into account.

Like Akaike's Information Criterion (AIC), the FMSC is a \emph{conservative} rather than consistent selection procedure, as it remains random even in the limit.	
While consistency is a desirable property in many settings, the situation is more complex for model and moment selection: consistent and conservative selection procedures have different strengths, but these strengths cannot be combined \citep{Yang2005}.
The goal of this paper is estimators with low risk.
Viewed from this perspective consistent selection criteria suffer from a serious defect: they exhibit unbounded minimax risk \citep{LeebPoetscher2008}.  
Conservative criteria such as the FMSC do not suffer from this shortcoming.
Moreover, as discussed in more detail below, the asymptotics of consistent selection paint a misleading picture of the effects of moment selection on inference.
For these reasons, the fact that the FMSC is conservative rather than consistent is an asset in the present context.

Because it studies inference post-moment selection, this paper relates to a vast literature on ``pre-test'' estimators.
For an overview, see \citet{LeebPoetscher2005, LeebPoetscher2009}.
There are several proposals to construct valid confidence intervals post-model selection, including \cite{Kabaila1998}, \cite{HjortClaeskens} and \cite{KabailaLeeb2006}. 
To my knowledge, however, this is the first paper to treat the problem in general for post-moment selection and moment average estimators in the presence of mis-specification.
Some related results appear in \cite{Berkowitz2008}, \cite{Berkowitz2012}, \cite{Guggenberger2010}, \cite{Guggenberger2012}, \cite{GuggenbergerKumar}, and \cite{Caner2014}.
While I developed the simulation-based, two-stage confidence interval procedure described below by analogy to a suggestion in \cite{ClaeskensHjortbook}, \cite{Leeb} kindly pointed out that similar constructions have appeared in \cite{Loh1985}, \cite{Berger1994}, and \cite{Silvapulle1996}. More recently, \cite{McCloskey} takes a similar approach to study a class of non-standard testing problems.

The framework within which I study moment averaging is related to the frequentist model average estimators of \cite{HjortClaeskens}.
Two other papers that consider weighting estimators based on different moment conditions are \cite{Xiao} and \cite{ChenChavezLinton}.
Whereas these papers combine estimators computed using valid moment conditions to achieve a minimum variance estimator, I combine estimators computed using potentially invalid conditions with the aim of reducing estimator AMSE.
A related idea underlies the combined moments (CM) estimator of \cite{Judge2007}.
For a different approach to combining OLS and TSLS estimators, similar in spirit to the Stein-estimator and developed independently of the work presented here, see \cite{HansenStein}. 
\cite{ChengLiaoShi} provide related results for Stein-type moment averaging in a GMM context with potentially mis-specified moment conditions.

The results presented here are derived under strong identification and abstract from the many instruments problem. 
Supplementary simulation results presented in Online Appendix \ref{sec:appendWeak}, however, suggest that the FMSC can nevertheless perform well when the ``baseline'' assumptions only weakly identify the target parameter. 
Extending the idea behind the FMSC to allow for weak identification and possibly a large number of moment conditions is a challenging topic that I leave for future research.

The remainder of the paper is organized as follows.
Section \ref{sec:asymp} describes the asymptotic framework and Section \ref{sec:FMSC} derives the FMSC, both in general and for two specific examples: OLS versus TSLS and choosing instrumental variables.
Section \ref{sec:avg} studies moment average estimators and shows how they can be used to construct valid confidence intervals post-moment selection.
Section \ref{sec:simulations} presents simulation results and Section \ref{sec:application} considers an empirical example from development economics.
Proofs appear at the end of the document; computational details and additional material appear in an Online Appendix. 

\section{Assumptions and Asymptotic Framework}
\label{sec:asymp}

\subsection{Local Mis-Specification}
Let $f(\cdot,\cdot)$ be a $(p+q)$-vector of moment functions of a random vector $Z$ and an $r$-dimensional parameter vector $\theta$, partitioned according to $f(\cdot,\cdot) = \left(g(\cdot,\cdot)', h(\cdot,\cdot)'  \right)'$ where $g(\cdot,\cdot)$ and $h(\cdot,\cdot)$ are $p$- and $q$-vectors of moment functions. 
The moment condition associated with $g$ is assumed to be correct whereas that associated with $h$ is locally mis-specified.
More precisely, 
\begin{assump}[Local Mis-Specification]
\label{assump:drift}
Let $\{Z_{ni}\colon 1\leq i \leq n, n =1, 2, \hdots\}$ be an iid triangular array of random vectors defined on a probability space $(\Upsilon, \mathcal{F}, \mathbb{P})$ satisfying
	\begin{enumerate}[(a)]
		\item $E[g(Z_{ni},\theta_0)] = 0$,
		\item $E[h(Z_{ni},\theta_0)] = n^{-1/2}\tau$, where $\tau$ is an unknown constant vector, 
		\item $\{f(Z_{ni},\theta_0)\colon 1\leq i \leq n, n = 1, 2, \hdots\}$ is uniformly integrable, and
		\item $Z_{ni} \rightarrow_d Z_i$.
	\end{enumerate}
\end{assump}
For any fixed sample size $n$, the expectation of $h$ evaluated at the true parameter value $\theta_0$ depends on the unknown constant vector $\tau$. 
Unless all components of $\tau$ are zero, some of the moment conditions contained in $h$ are mis-specified. 
In the limit however, this mis-specification vanishes, as $\tau/\sqrt{n}$ converges to zero. 
Uniform integrability combined with weak convergence implies convergence of expectations, so that $E[g(Z_i, \theta_0)]=0$ and $E[h(Z_i, \theta_0)]=0$. Because the limiting random vectors $Z_i$ are identically distributed, I suppress the $i$ subscript and simply write $Z$ to denote their common marginal law, e.g.\ $E[h(Z,\theta_0)]=0$. 
Local mis-specification is \emph{not} intended as a literal description of real-world datasets: it is merely a device that gives an asymptotic bias-variance trade-off that mimics the finite-sample intuition.
Moreover, while I work with an iid triangular array for simplicity, the results presented here can be adapted to handle dependent random variables.

\subsection{Candidate GMM Estimators}
Define the sample analogue of the expectations in Assumption \ref{assump:drift} as follows:
$$f_n(\theta) = \frac{1}{n}\sum_{i=1}^n f(Z_{ni},\theta) = \left[\begin{array}{c} g_n(\theta)\\ h_n(\theta) \end{array} \right]=\left[\begin{array}{c}n^{-1}\sum_{i=1}^n g(Z_{ni},\theta) \\ n^{-1}\sum_{i=1}^n h(Z_{ni},\theta) \end{array}\right]$$
where $g_n$ is the sample analogue of the correctly specified moment conditions and $h_n$ is that of the (potentially) mis-specified moment conditions.
A candidate GMM estimator $\widehat{\theta}_S$ uses some subset $S$ of the moment conditions contained in $f$ in estimation. 
Let $|S|$ denote the number of moment conditions used and suppose that $|S|>r$ so the GMM estimator is unique.\footnote{Identifying $\tau$ requires futher assumptions, as discussed in Section \ref{sec:ident}.} 
Let $\Xi_S$ be the $|S| \times(p +q)$ \emph{moment selection matrix} corresponding to $S$. That is, $\Xi_S$ is a matrix of ones and zeros arranged such that $\Xi_S f_n(\theta)$ contains only the sample moment conditions used to estimate $\widehat{\theta}_S$. 
Thus, the GMM estimator of $\theta$ based on moment set $S$ is given by 
$$\widehat{\theta}_S = \underset{\theta \in \Theta}{\mbox{arg min}}\; \left[\Xi_S f_n(\theta)\right]' \widetilde{W}_S \; \left[ \Xi_S f_n(\theta)\right].$$
where $\widetilde{W}_S$ is an $|S|\times |S|$, positive definite weight matrix.
There are no restrictions placed on $S$ other than the requirement that $|S| >r$ so the GMM estimate is well-defined. 
In particular, $S$ may \emph{exclude} some or all of the valid moment conditions contained in $g$.
This notation accommodates a wider range of examples, including choosing between OLS and TSLS estimators.

To consider the limit distribution of $\widehat{\theta}_S$, we require some further notation. 
First define the derivative matrices
	$$G = E\left[\nabla_{\theta} \; g(Z,\theta_0)\right], \quad H = E\left[\nabla_{\theta} \; h(Z,\theta_0)\right], \quad F = (G', H')'$$
and let $\Omega = Var\left[ f(Z,\theta_0) \right]$ where $\Omega$ is partitioned into blocks $\Omega_{gg}$, $\Omega_{gh}$, $\Omega_{hg}$, and $\Omega_{hh}$ conformably with the partition of $f$ by $g$ and $h$. 
Notice that each of these expressions involves the \emph{limiting random variable} $Z$ rather than $Z_{ni}$, so that the corresponding expectations are taken with respect to a distribution for which all moment conditions are correctly specified. 
Finally, to avoid repeatedly writing out pre- and post-multiplication by $\Xi_S$, define $F_S = \Xi_S F$ and $\Omega_S = \Xi_S \Omega\Xi_S'$. 
The following high level assumptions are sufficient for the consistency and asymptotic normality of the candidate GMM estimator $\widehat{\theta}_S$. 
\begin{assump}[High Level Sufficient Conditions]
\label{assump:highlevel} 
\mbox{}
	\begin{enumerate}[(a)]
		\item $\theta_0$ lies in the interior of $\Theta$, a compact set
		\item $\widetilde{W}_S \rightarrow_{p} W_S$, a positive definite matrix
		\item $W_S \Xi_S E[f(Z,\theta)]=0$ if and only if $\theta = \theta_0$
		\item $E[f(Z,\theta)]$ is continuous on $\Theta$
		\item $\sup_{\theta \in \Theta}\| f_n(\theta)- E[f(Z,\theta)]\|\rightarrow_{p} 0$
		\item $f$ is Z-almost surely differentiable in an open neighborhood $\mathcal{B}$ of $\theta_0$
		\item $\sup_{\theta \in \Theta} \|\nabla_{\theta}f_n(\theta) - F(\theta)\|\rightarrow_{p} 0$
		\item $\sqrt{n}f_n(\theta_0) \rightarrow_d  M + \left[\begin{array}{c}0\\ \tau \end{array} \right]$ where $M \sim N_{p+q}\left(0, \Omega \right)$
		\item $F_S'W_SF_S$ is invertible
	\end{enumerate}
\end{assump}

Although Assumption \ref{assump:highlevel} closely approximates the standard regularity conditions for GMM estimation, establishing primitive conditions for Assumptions \ref{assump:highlevel} (d), (e), (g) and (h) is slightly more involved under local mis-specification. 
Low-level sufficient conditions for the two running examples considered in this paper appear in Online Appendix \ref{sec:sufficient_conditions}.
For more general results, see \cite{Andrews1988} Theorem 2 and \cite{Andrews1992} Theorem 4.  
Notice that identification, (c), and continuity, (d), are conditions on the distribution of $Z$, the marginal law to which each $Z_{ni}$ converges. 

\begin{thm}[Consistency]
\label{thm:consist}
Under Assumptions \ref{assump:drift} and \ref{assump:highlevel} (a)--(e), $\widehat{\theta}_S \rightarrow_{p} \theta_0$.
\end{thm}

\begin{thm}[Asymptotic Normality]
\label{thm:normality}
Under Assumptions \ref{assump:drift} and \ref{assump:highlevel}
$$\sqrt{n}(\widehat{\theta}_S - \theta_0 ) \rightarrow_d -K_S \Xi_S  \left(\left[\begin{array}
	{c} M_g \\ M_h
\end{array} \right]  + \left[\begin{array}
	{c} 0 \\ \tau
\end{array} \right]\right)$$
where $K_S  = [F_S'W_SF_S]^{-1} F_S'W_S$, $M = (M_g', M_h')'$, and $M \sim N(0,\Omega)$.
\end{thm}

As we see from Theorems \ref{thm:consist} and \ref{thm:normality}, \emph{any} candidate GMM estimator $\widehat{\theta}_S$ is consistent for $\theta_0$ under local mis-specification. 
Unless $S$ excludes \emph{all} of the moment conditions contained in $h$, however, $\widehat{\theta}_S$ inherits an asymptotic bias from the mis-specification parameter $\tau$. 
The local mis-specification framework is useful precisely because it results in a limit distribution for $\widehat{\theta}_S$ with both a bias \emph{and} a variance. 
This captures in asymptotic form the bias-variance tradeoff that we see in finite sample simulations.
In constrast, fixed mis-specification results in a degenerate bias-variance tradeoff in the limit: scaling up by $\sqrt{n}$ to yield an asymptotic variance causes the bias component to diverge.

\subsection{Identification}
\label{sec:ident}
Any form of moment selection requires an identifying assumption: we need to make clear which parameter value $\theta_0$ counts as the ``truth.''
One approach, following \cite{Andrews1999}, is to assume that there exists a unique, maximal set of correctly specified moment conditions that identifies $\theta_0$. 
In the notation of the present paper\footnote{Although \cite{Andrews1999}, \cite{AndrewsLu}, and \cite{HongPrestonShum} consider \emph{fixed} mis-specification, we can view this as a version of local mis-specification in which $\tau \rightarrow \infty$ sufficiently fast.} this is equivalent to the following:

\begin{assump}[\cite{Andrews1999} Identification Condition]
	\label{assump:Andrews}
There exists a subset $S_{max}$ of at least $r$ moment conditions satisfying:
	\begin{enumerate}[(a)]
	 	\item $\Xi_{S_{max}} E[f(Z_{ni},\theta_0)]= 0$
	 	\item For any $S' \neq S_{max}$ such that $\Xi_{S'} E[f(Z_{ni},\theta')]= 0$ for some $\theta' \neq \theta_0$, $|S_{max}| > |S'|$.
	 \end{enumerate}
\end{assump}

\cite{AndrewsLu} and \cite{HongPrestonShum} take the same basic approach to identification, with appropriate modifications to allow for simultaneous model and moment selection. 
An advantage of Assumption \ref{assump:Andrews} is that, under fixed mis-specification, it allows consistent selection of $S_{max}$ without any prior knowledge of \emph{which} moment conditions are correct. 
In the notation of the present paper this corresponds to having no moment conditions in the $g$ block. As \citet[p.\ 254]{Hallbook} points out, however, the second part of Assumption \ref{assump:Andrews} can fail even in very simple settings.
When it does fail, the selected GMM estimator may no longer be consistent for $\theta_0$. 
A different approach to identification is to assume that there is a minimal set of at least $r$ moment conditions \emph{known} to be correctly specified.
This is the approach I follow here, as do \cite{Liao} and \cite{ChengLiao}.\footnote{For a dicussion of why Assumption \ref{assump:Identification} is necessary and how to proceed when it fails, see Online Appendix \ref{sec:digress}.}

\begin{assump}[FMSC Identification Condition] 
\label{assump:Identification}
Let $\widehat{\theta}_v$ denote the GMM estimator based solely on the moment conditions contained in the $g$--block
$$\widehat{\theta}_v = \underset{\theta \in \Theta}{\mbox{arg min}}\; g_n(\theta)' \widetilde{W}_{v} \; g_n(\theta)$$
We call this the ``valid estimator'' and assume that it satisfies all the conditions of Assumption \ref{assump:highlevel}. Note that this implies $p\geq r$.
\end{assump}

Assumption \ref{assump:Identification} and Theorem \ref{thm:normality} immediately imply that the valid estimator shows no asymptotic bias. 
\begin{cor}[Limit Distribution of Valid Estimator]
	\label{cor:valid}
	Let $S_{v}$ include only the moment conditions contained in $g$. 
	Then, under Assumption \ref{assump:Identification} we have
		$$\sqrt{n}\left(\widehat{\theta}_v - \theta_0\right) \rightarrow_d -K_v M_g$$
	by applying Theorem \ref{thm:normality} to $S_{v}$, where $K_v = [G'W_vG]^{-1}G'W_v$ and $M_g \sim N(0,\Omega_{gg})$. 
\end{cor}

Both Assumptions \ref{assump:Andrews} and \ref{assump:Identification} are strong, and neither fully nests the other. 
In the context of the present paper, Assumption \ref{assump:Identification} is meant to represent a situation in which an applied researcher chooses between two groups of assumptions.
The $g$--block contains the ``baseline'' assumptions while the $h$--block contains a set of stronger, more controversial ``suspect'' assumptions.
The FMSC is designed for settings in which the $h$--block is expected to contain a substantial amount of information beyond that already contained in the $g$--block. 
The idea is that, if we knew the $h$--block was correctly specified, we would expect a large gain in efficiency by including it in estimation. 
This motivates the idea of trading off the variance reduction from including $h$ against the potential increase in bias.

\section{The Focused Moment Selection Criterion}
\label{sec:FMSC}

\subsection{The General Case}
The FMSC chooses among the potentially invalid moment conditions contained in $h$ based on the estimator AMSE of a user-specified scalar target parameter.\footnote{Although I focus on the case of a scalar target parameter in the body of the paper, the same idea can be applied to a vector of target parameters. For details see Online Appendix \ref{append:mult}.}
Denote this target parameter by $\mu$, a real-valued, $Z$-almost continuous function of the parameter vector $\theta$ that is differentiable in a neighborhood of $\theta_0$. 
Further, define the GMM estimator of $\mu$ based on $\widehat{\theta}_S$ by $\widehat{\mu}_S = \mu(\widehat{\theta}_S)$ and the true value of $\mu$ by $\mu_0 = \mu(\theta_0)$. 
Applying the Delta Method to Theorem \ref{thm:normality} gives the AMSE of $\widehat{\mu}_S$.

\begin{cor}[AMSE of Target Parameter]
\label{cor:target}
Under the hypotheses of Theorem \ref{thm:normality}, 
$$\sqrt{n}\left(\widehat{\mu}_S - \mu_0\right)\rightarrow_d-\nabla_\theta\mu(\theta_0)'K_S \Xi_S \left(M +  \left[\begin{array}
	{c} 0 \\ \tau
\end{array} \right]\right)$$ 
where $M$ is defined in Theorem \ref{thm:normality}.
Hence,
	$$\mbox{AMSE}\left(\widehat{\mu}_S\right) = \nabla_\theta\mu(\theta_0)'K_S \Xi_S \left\{\left[\begin{array}{cc}0&0\\0&\tau\tau'\end{array}\right] + \Omega\right\}\Xi_S'K_S'\nabla_\theta\mu(\theta_0).$$
\end{cor}

For the valid estimator $\widehat{\theta}_v$ we have $K_v = \left[G'W_{v}G\right]^{-1}G' W_{v}$ and $\Xi_v =\left[\begin{array}{cc} \mathbf{I}_p& \mathbf{0}_{p\times q} \end{array} \right]$. 
Thus, the valid estimator $\widehat{\mu}_v$ of $\mu$ has zero asymptotic bias. 
In contrast, any candidate estimator $\widehat{\mu}_S$ that includes moment conditions from $h$ inherits an asymptotic bias from the corresponding elements of $\tau$, the extent and direction of which depends both on $K_S$ and $\nabla_\theta\mu(\theta_0)$. 
The setting considered here, however, is one in which using moment conditions from $h$ in estimation will reduce the asymptotic variance.
In the nested case, where moment conditions from $h$ are \emph{added} to those of $g$, this follows automatically.
The usual proof that adding moment conditions cannot increase asymptotic variance under efficient GMM \citep[see for example][ch.\ 6]{Hallbook} continues to hold under local mis-specification, because all moment conditions are correctly specified in the limit.
In non-nested examples, for example when $h$ contains OLS moment conditions and $g$ contains IV moment conditions, however, this result does not apply because one would use $h$ \emph{instead of} $g$.
In such examples, one must establish an analogous ordering of asymptotic variances by direct calculation, as I do below for the OLS versus IV example.

Using this framework for moment selection requires estimators of the unknown quantities: $\theta_0$, $K_S$, $\Omega$, and $\tau$. 
Under local mis-specification, the estimator of $\theta$ under \emph{any} moment set is consistent. 
A natural estimator is $\widehat{\theta}_v$, although there are other possibilities. 
Recall that $K_S = [F_S'W_SF_S]^{-1} F_S'W_S \Xi_S$.
Because it is simply the selection matrix defining moment set $S$, $\Xi_S$ is known.  
The remaining quantities $F_S$ and $W_S$ that make up $K_S$ are consistently estimated by their sample analogues under Assumption \ref{assump:highlevel}.
Similarly, consistent estimators of $\Omega$ are readily available under local mis-specification, although the precise form depends on the situation.\footnote{See Sections \ref{sec:OLSvsIVExample} and \ref{sec:chooseIVexample} for discussion of this point for the two running examples.}
The only remaining unknown is $\tau$. Local mis-specification is essential for making meaningful comparisons of AMSE because it prevents the bias term from dominating the comparison. 
Unfortunately, it also prevents consistent estimation of the asymptotic bias parameter.
Under Assumption \ref{assump:Identification}, however, it remains possible to construct an \emph{asymptotically unbiased} estimator $\widehat{\tau}$ of $\tau$ by substituting $\widehat{\theta}_v$, the estimator of $\theta_0$ that uses only correctly specified moment conditions, into $h_n$, the sample analogue of the potentially mis-specified moment conditions. 
In other words,  $\widehat{\tau} = \sqrt{n} h_n(\widehat{\theta}_v)$. 

\begin{thm}[Asymptotic Distribution of $\widehat{\tau}$] 
\label{thm:tau}
Let $\widehat{\tau} = \sqrt{n} h_n(\widehat{\theta}_v)$ where $\widehat{\theta}_v$ is the valid estimator, based only on the moment conditions contained in $g$. 
Then under Assumptions \ref{assump:drift}, \ref{assump:highlevel} and \ref{assump:Identification}
$$\widehat{\tau} \rightarrow_d \Psi\left( M + \left[\begin{array}
	{c} 0 \\ \tau
\end{array} \right]\right), \quad \Psi = \left[\begin{array}{cc} -HK_v & \mathbf{I}_q \end{array}\right]$$ 
where $K_v$ is defined in Corollary \ref{cor:valid} and $\mathbf{I}_q$ denotes the $(q\times q)$ identity matrix.
Thus, $\widehat{\tau}\rightarrow_d (\Psi M + \tau) \sim N_q(\tau, \Psi \Omega \Psi')$.
\end{thm}

Returning to Corollary $\ref{cor:target}$, however, we see that it is $\tau \tau'$ rather than $\tau$ that enters the expression for AMSE. 
Although $\widehat{\tau}$ is an asymptotically unbiased estimator of $\tau$, the limiting expectation of $\widehat{\tau} \widehat{\tau}'$ is not $\tau\tau'$ because $\widehat{\tau}$ has an asymptotic variance.  
Subtracting a consistent estimate of the asymptotic variance removes this asymptotic bias.

\begin{cor}[Asymptotically Unbiased Estimator of $\tau \tau'$]
\label{cor:tautau}
If $\widehat{\Omega}$ and $\widehat{\Psi}$ are consistent for $\Omega$ and $\Psi$, then $ \widehat{\tau}\widehat{\tau}' - \widehat{\Psi}\widehat{\Omega}\widehat{\Psi}$ is an asymptotically unbiased estimator of $\tau\tau'$.
\end{cor}
It follows that
\begin{equation}
\label{eq:fmsc}
	\mbox{FMSC}_n(S) = \nabla_\theta\mu(\widehat{\theta})'\widehat{K}_S\Xi_S \left\{\left[\begin{array}{cc}0&0\\0&\widehat{\tau}\widehat{\tau}' - \widehat{\Psi}\widehat{\Omega}\widehat{\Psi}'\end{array}\right] + \widehat{\Omega}\right\}\Xi_S'\widehat{K}_S' \nabla_\theta\mu(\widehat{\theta})
\end{equation}
provides an asymptotically unbiased estimator of AMSE.
Given a set $\mathscr{S}$ of candidate specifications, the FMSC selects the candidate $S^*$ that \emph{minimizes} the expression given in Equation \ref{eq:fmsc}, that is $S^*_{FMSC} =  \arg \min_{S\in \mathscr{S}} \;\mbox{FMSC}_n(S)$.

In summary, the FMSC aims to choose the moment conditions that provide the lowest risk estimator of a target parameter $\mu$ where risk is defined as MSE.\footnote{One could choose a different risk function and proceed similarly, although I do not consider this idea further below. See, e.g., \cite{Claeskens2006} and \cite{ClaeskensHjort2008}.}
Because finite-sample MSE is unavailable, AMSE in a local-to-zero asymptotic framework serves in its stead.
Since no consistent estimator of AMSE exists in this setting, FMSC uses an asymptotically unbiased estimator.
This is the same idea that underlies the classical AIC and TIC model selection criteria as well as more recent procedures such as those described in \cite{ClaeskensHjort2003} and \cite{Schorfheide2005}.

\subsection{OLS versus TSLS Example}
\label{sec:OLSvsIVExample}
The simplest interesting application of the FMSC is choosing between ordinary least squares (OLS) and two-stage least squares (TSLS) estimators of the effect $\beta$ of a single endogenous regressor $x$ on an outcome of interest $y$.
The intuition is straightforward: because TSLS is a high-variance estimator, OLS will have a lower mean-squared error provided that $x$ isn't \emph{too} endogenous.\footnote{Because the moments of the TSLS estimator only exist up to the order of overidentificiation \citep{Phillips1980, Kinal} mean-squared error should be understood to refer to ``trimmed'' mean-squared error when the number of instruments is two or fewer. For details, see Online Appendix \ref{append:trim}.}
To keep the presentation transparent, I work within an iid, homoskedastic setting for this example and assume, without loss of generality, that there are no exogenous regressors.\footnote{The homoskedasticity assumption concerns the \emph{limit} random variables: under local mis-specification there will be heteroskedasticity for fixed $n$. See Assumption \ref{assump:OLSvsIV} in Online Appendix \ref{sec:sufficient_conditions} for details.}
Equivalently we may suppose that any exogenous regressors, including a constant, have been ``projected out.''
Low-level sufficient conditions for all of the results in this section appear in Assumption \ref{assump:OLSvsIV} of Online Appendix \ref{sec:sufficient_conditions}.
The data generating process is
    \begin{eqnarray}
			y_{ni} &=& \beta x_{ni}  + \epsilon_{ni}\\
	x_{ni} &=& \mathbf{z}_{ni}' \boldsymbol{\pi} + v_{ni}
	\end{eqnarray}
where $\beta$ and $\boldsymbol{\pi}$ are unknown constants, $\mathbf{z}_{ni}$ is a vector of exogenous and relevant instruments, $x_{ni}$ is the endogenous regressor, $y_{ni}$ is the outcome of interest, and $\epsilon_{ni}, v_{ni}$ are unobservable error terms.
All random variables in this system are mean zero, or equivalently all constant terms have been projected out. 
Stacking observations in the usual way, the estimators under consideration are $\widehat{\beta}_{OLS} = \left(\mathbf{x}'\mathbf{x}\right)^{-1}\mathbf{x}'\mathbf{y}$ and
$\widetilde{\beta}_{TSLS} = \left(\mathbf{x}'P_Z\mathbf{x}\right)^{-1}\mathbf{x}'P_Z\mathbf{y}$ where we define $P_Z = Z(Z'Z)^{-1}Z'$. 

\begin{thm}[OLS and TSLS Limit Distributions]
	\label{thm:OLSvsIV} 
  Let $(\mathbf{z}_{ni}, v_{ni}, \epsilon_{ni})$ be a triangular array of random variables such that $E[\mathbf{z}_{ni} \epsilon_{ni}]=\mathbf{0}$, $E[\mathbf{z}_{ni} v_{ni}]=\mathbf{0}$, and $E[\epsilon_{ni}v_{ni}] = \tau/\sqrt{n}$ for all $n$. Then, under standard regularity conditions, e.g. Assumption \ref{assump:OLSvsIV} in Online Appendix \ref{sec:sufficient_conditions}, 
	$$
\left[
\begin{array}{c}
  \sqrt{n}(\widehat{\beta}_{OLS} - \beta) \\
  \sqrt{n}(\widetilde{\beta}_{TSLS} - \beta)
\end{array}
\right] \overset{d}{\rightarrow}
N\left(
\left[
\begin{array}{c}
\tau/\sigma_x^2 \\ 
0
\end{array}
\right],\;
\sigma_\epsilon^2 \left[ \begin{array}{cc}
  1/\sigma_x^2 & 1/\sigma_x^2\\
  1/\sigma_x^2 & 1/\gamma^2 
  \end{array}\right]
  \right)
$$
where $\sigma_x^2 = \gamma^2 + \sigma_v^2$, $\gamma^2 = \boldsymbol{\pi}'Q \boldsymbol{\pi}$, $E[\mathbf{z}_{ni} \mathbf{z}_{ni}'] \rightarrow Q$, $E[v_{ni}^2]\rightarrow \sigma_v^2$, and $E[\epsilon_{ni}^2] \rightarrow \sigma_\epsilon^2$ as $n\rightarrow \infty$.
\end{thm}
We see immediately that, as expected, the variance of the OLS estimator is always strictly lower than that of the TSLS estimator since $\sigma^2_\epsilon/\sigma_x^2 = \sigma^2_\epsilon/(\gamma^2 + \sigma_v^2)$. 
Unless $\tau = 0$, however, OLS shows an asymptotic bias. 
In contrast, the TSLS estimator is asymptotically unbiased regardless of the value of $\tau$.  
Thus,
$$\mbox{AMSE(OLS)} = \frac{\tau^2}{\sigma_x^4} + \frac{\sigma_\epsilon^2}{\sigma_x^2},\quad \quad
  \mbox{AMSE(TSLS)} = \frac{\sigma_\epsilon^2}{\gamma^2}.$$
 and rerranging, we see that the AMSE of the OLS estimator is strictly less than that of the TSLS estimator whenever $\tau^2  < \sigma_x^2 \sigma_\epsilon^2\sigma_v^2/\gamma^2$. 
To estimate the unknowns required to turn this inequality into a moment selection procedure, I set 
  $$\widehat{\sigma}_x^2 = n^{-1}\mathbf{x}'\mathbf{x}, \quad \widehat{\gamma}^2 = n^{-1}\mathbf{x}'Z(Z'Z)^{-1}Z'\mathbf{x}, \quad \widehat{\sigma}_v^2 =  \widehat{\sigma}_x^2 - \widehat{\gamma}^2$$
and define
$$\widehat{\sigma}_\epsilon^2 = n^{-1}\left(\textbf{y} - \textbf{x}\widetilde{\beta}_{TSLS} \right)'\left(\textbf{y} - \textbf{x}\widetilde{\beta}_{TSLS} \right)$$
Under local mis-specification each of these estimators is consistent for its population counterpart.\footnote{While using the OLS residuals to estimate $\sigma_\epsilon^2$ \emph{also} provides a consistent estimate under local mis-specification, the estimator based on the TSLS residuals should be more robust.} 
All that remains is to estimate $\tau^2$. Specializing Theorem \ref{thm:tau} and Corollary \ref{cor:tautau} to the present example gives the following result.
\begin{thm}
	\label{thm:tauOLSvsIV}
	Let $\widehat{\tau} =  n^{-1/2} \mathbf{x}'(\mathbf{y} - \mathbf{x}\widetilde{\beta}_{TSLS})$. Then, under the conditions of Theorem \ref{thm:OLSvsIV},
	$$\widehat{\tau}\rightarrow_d N(\tau,V), \quad V = \sigma_\epsilon^2 \sigma_x^2(\sigma_v^2/\gamma^2).$$ 
\end{thm}
It follows that $\widehat{\tau}^2 -  \widehat{\sigma}_\epsilon^2\widehat{\sigma}_x^2 \left(\widehat{\sigma}_v^2/\widehat{\gamma}^2\right)$ is an asymptotically unbiased estimator of $\tau^2$ and hence, substituting into the AMSE inequality from above and rearranging, the FMSC instructs us to choose OLS whenever $\widehat{T}_{FMSC} = \widehat{\tau}^2/\widehat{V} < 2$
where $\widehat{V} = \widehat{\sigma}_v^2 \widehat{\sigma}_\epsilon^2 \widehat{\sigma}_x^2/\widehat{\gamma}^2$. 
The quantity $\widehat{T}_{FMSC}$ looks very much like a test statistic and indeed it can be viewed as such. 
By Theorem \ref{thm:tauOLSvsIV} and the continuous mapping theorem, $\widehat{T}_{FMSC} \rightarrow_d \chi^2(1)$. 
Thus, the FMSC can be viewed as a test of the null hypothesis $H_0\colon \tau = 0$ against the two-sided alternative with a critical value of $2$. 
This corresponds to a significance level of $\alpha \approx 0.16$. 
But how does this novel ``test'' compare to something more familiar, say the Durbin-Hausman-Wu (DHW) test? 
It turns out that in this particular example, although not in general, the FMSC is \emph{numerically equivalent} to using OLS unless the DHW test rejects at the 16\% level. 
\begin{thm}
    \label{thm:DHW} Under the conditions of Theorem \ref{thm:OLSvsIV}, FMSC selection between the OLS and TSLS estimators is equivalent to a Durbin-Hausman-Wu pre-test with a critical value of $2$.
\end{thm}
The equivalence between FMSC selection and a DHW test in this example is helpful for two reasons. 
First, it provides a novel justification for the use of the DHW test to select between OLS and TSLS. So long as it is carried out with $\alpha \approx 16\%$, the DHW test is equivalent to selecting the estimator that minimizes an asymptotically unbiased estimator of AMSE. 
Note that this significance level differs from the more usual values of 5\% or 10\% in that it leads us to select TSLS \emph{more often}: OLS should indeed be given the benefit of the doubt, but not by so wide a margin as traditional practice suggests. 
Second, this equivalence shows that the FMSC can be viewed as an \emph{extension} of the idea behind the familiar DHW test to more general GMM environments.\footnote{Note that the FMSC in this example, characterized in Theorem \ref{thm:DHW}, chooses between OLS and IV to minimize estimator AMSE. If one wishes to carry out inference post-selection one must contend with the size distortions of the familiar ``textbook'' confidence interval procedure, as pointed out by \cite{Guggenberger2010}. I discuss this point extensively below in Section \ref{sec:avg} ans propose possible remedies.}

\subsection{Choosing Instrumental Variables Example}
\label{sec:chooseIVexample}
The OLS versus TSLS example is really a special case of instrument selection: if $x$ is exogenous, it is clearly ``its own best instrument.'' 
Viewed from this perspective, the FMSC amounts to trading off endogeneity against instrument strength. I now consider instrument selection in general for linear GMM estimators in an iid setting. 
Consider the  model:
\begin{eqnarray}
  y_{ni} &=& \mathbf{x}_{ni}' \beta +  \epsilon_{ni}\\
    \mathbf{x}_{ni} &=&  \Pi_1' \mathbf{z}_{ni}^{(1)} + \Pi_2'\mathbf{z}_{ni}^{(2)} + \mathbf{v}_{ni}
\end{eqnarray}
where $y$ is an outcome of interest, $\mathbf{x}$ is an $r$-vector of regressors, some of which are endogenous, $\mathbf{z}^{(1)}$ is a $p$-vector of instruments known to be exogenous, and $\mathbf{z}^{(2)}$ is a $q$-vector  of \emph{potentially endogenous} instruments. 
The $r$-vector $\beta$, $p\times r$ matrix $\Pi_1$, and $q\times r$ matrix $\Pi_2$ contain unknown constants. Stacking observations in the usual way, we can write the system in matrix form as $\mathbf{y} = X\beta +\boldsymbol{\epsilon}$ and $X =  Z \Pi + V$, where $Z = (Z_1, Z_2)$ and $\Pi = (\Pi_1', \Pi_2')'$. 

In this example, the idea is that the instruments contained in $Z_2$ are expected to be strong.
If we were confident that they were exogenous, we would certainly use them in estimation. 
Yet the very fact that we expect them to be strongly correlated with $\mathbf{x}$ gives us reason to fear that they may be endogenous. 
The exact opposite is true of $Z_1$: these are the instruments that we are prepared to assume are exogenous. 
But when is such an assumption plausible? Precisely when the instruments contained in $Z_1$ are \emph{not especially strong}. 
Accordingly, the FMSC attempts to trade off a small increase in bias from using a \emph{slightly} endogenous instrument against a larger decrease in variance from increased instrument strength.
To this end, consider a general linear GMM estimator of the form
$$\widehat{\beta}_S = (X'Z_S \widetilde{W}_S Z_S' X)^{-1}X'Z_S \widetilde{W}_S  Z_S' \mathbf{y}$$
where $S$ indexes the instruments used in estimation, $Z_S'  = \Xi_S Z'$ is the matrix containing only those instruments included in $S$, $|S|$ is the number of instruments used in estimation and $\widetilde{W}_S$ is an $|S|\times|S|$ positive definite weighting matrix.

\begin{thm}[Choosing IVs Limit Distribution]
\label{thm:chooseIV} 
Let $(\mathbf{z}_{ni}, v_{ni}, \epsilon_{ni})$ be a triangular array of random variables such that $E[\mathbf{z}_{ni} \epsilon_{ni}]=\mathbf{0}$, $E[\mathbf{z}_{ni} v_{ni}]=\mathbf{0}$, and $E[\epsilon_{ni}v_{ni}] = \tau/\sqrt{n}$ for all $n$. Suppose further that $\widetilde{W}_S \rightarrow_p W_S >0$. 
Then, under standard regularity conditions, e.g.\ Assumption \ref{assump:chooseIV} in Online Appendix \ref{sec:sufficient_conditions}, 
$$\sqrt{n}\left(\widehat{\beta}_S - \beta \right) \overset{d}{\rightarrow} -K_S \Xi_S \left(\left[\begin{array}
           {c} \mathbf{0} \\ \boldsymbol{\tau}
         \end{array}\right] + M \right)$$
where
         $$-K_S = \left(\Pi' Q_S W_S Q_S'\Pi\right)^{-1} \Pi'Q_SW_S$$
$M \sim N(\mathbf{0}, \Omega)$, $Q_S = Q \Xi_S'$, $E[\mathbf{z}_{ni} \mathbf{z}_{ni}'] \rightarrow Q$ and $E[\epsilon_{ni}^2 \mathbf{z}_{ni} \mathbf{z}_{ni}'] \rightarrow \Omega$ as $n\rightarrow \infty$
\end{thm}
To implement the FMSC for this example, we simply need to specialize Equation \ref{eq:fmsc}.
To simplify the notation, let
\begin{equation}
\label{eq:xi12}
\Xi_1 = \left[\begin{array}{cc} \mathbf{I}_{p} & 0_{p \times q}  \end{array}\right], \quad
    \Xi_2 = \left[ \begin{array}{cc}
        0_{q \times p}& \mathbf{I}_{q}
            \end{array}\right]	
\end{equation}
where $0_{p\times q}$ denotes a $p\times q$ matrix of zeros and $\mathbf{I}_q$ denotes the $q\times q$ identity matrix.
Using this convention, $Z_1 = Z \Xi_1'$ and $Z_2 = Z \Xi_2'$.
In this example the valid estimator, defined in Assumption \ref{assump:Identification}, is given by
\begin{equation}
\label{eq:betav}
\widehat{\beta}_v = \left(X'Z_1 \widetilde{W}_v Z_1' X\right)^{-1}X'Z_1 \widetilde{W}_v Z_1' \mathbf{y}	
\end{equation}
and we estimate $\nabla_\beta \mu(\beta)$ with $\nabla_\beta \mu(\widehat{\beta}_v)$.  
Similarly, 
$$-\widehat{K}_S = n\left(X'Z \Xi_S' \widetilde{W}_S \Xi_S Z' X\right)^{-1}X' Z \Xi_S' \widetilde{W}_S$$
is the natural consistent estimator of $-K_S$ in this setting.\footnote{The negative sign is squared in the FMSC expression and hence disappears. I write it here only to be consistent with the notation of Theorem \ref{thm:normality}.}
Since $\Xi_S$ is known, the only remaining quantities from Equation \ref{eq:fmsc} are $\widehat{\boldsymbol{\tau}}$, $\widehat{\Psi}$ and $\widehat{\Omega}$. 
The following result specializes Theorem \ref{thm:tau} to the present example.
\begin{thm}
Let $\widehat{\boldsymbol{\tau}} = n^{-1/2} Z_2' ( \mathbf{y} - X\widehat{\beta}_v)$ where $\widehat{\beta}_v$ is as defined in Equation \ref{eq:betav}. Under the conditions of Theorem \ref{thm:chooseIV} we have
$\widehat{\boldsymbol{\tau}} \rightarrow_d \boldsymbol{\tau} + \Psi M$
where $M$ is defined in Theorem \ref{thm:chooseIV},
\begin{eqnarray*}
	\Psi &=&\left[ \begin{array}{cc}-\Xi_2Q \Pi K_v  & I_{q} \end{array}\right] \\
	-K_v &=& \left(\Pi' Q \Xi'_1 W_v \Xi_1 Q'\Pi\right)^{-1} \Pi'Q \Xi_1' W_v
\end{eqnarray*}
$W_v$ is the probability limit of the weighting matrix from Equation \ref{eq:betav}, $I_q$ is the $q\times q$ identity matrix, $\Xi_1$ is defined in Equation \ref{eq:xi12}, and $E[\mathbf{z}_{ni} \mathbf{z}_{ni}'] \rightarrow Q$. 
\end{thm}
Using this result, I construct the asymptotically unbiased estimator $\widehat{\tau}\widehat{\tau}' - \widehat{\Psi}\widehat{\Omega} \widehat{\Psi}'$ of $\tau\tau'$ from
	$$\widehat{\Psi} = \left[ \begin{array}
		{cc}
		-n^{-1}Z_2'X \left(-\widehat{K}_v\right) & I_q
	\end{array}\right], \quad -\widehat{K}_v = n\left(X'Z_1 \widetilde{W}_v Z_1' X\right)^{-1}X'Z_1 \widetilde{W}_v$$

All that remains before substituting values into Equation \ref{eq:fmsc} is to estimate $\Omega$. 
In the simulation and empirical examples discussed below I examine the TSLS estimator, that is $\widetilde{W}_S = (\Xi_S Z'Z\Xi_S)^{-1}$, and estimate $\Omega$ as follows. 
For all specifications \emph{except} the valid estimator $\widehat{\beta}_v$, I employ the centered, heteroskedasticity-consistent estimator
\begin{equation}
	\widehat{\Omega}_S = \frac{1}{n}\sum_{i=1}^n u_i(\widehat{\beta}_S)^2\mathbf{z}_{iS} \mathbf{z}_{iS}'  - \left(\frac{1}{n}\sum_{i=1}^n u_i(\widehat{\beta}_S)\mathbf{z}_{iS}   \right)\left(\frac{1}{n}\sum_{i=1}^n  u_i(\widehat{\beta}_S)\mathbf{z}_{iS}'  \right)
\end{equation}
where $u_i(\beta) = y_i - \mathbf{x}_i'\beta$, $\widehat{\beta}_S = (X'Z_S(Z_S'Z_S)^{-1}Z_S'X)^{-1}X'Z_S(Z_S'Z_S)^{-1}Z_S'\mathbf{y}$, $\mathbf{z}_{iS} = \Xi_S \mathbf{z}_i$ and $Z_S' = \Xi_S Z'$.
Centering allows moment functions to have non-zero means. 
While the local mis-specification framework implies that these means tend to zero in the limit, they are non-zero for any fixed sample size. 
Centering accounts for this fact, and thus provides added robustness. 
Since the valid estimator $\widehat{\beta}_v$ has no asymptotic bias, the AMSE of any target parameter based on this estimator equals its asymptotic variance. 
Accordingly, I use 
\begin{equation}
	\widetilde{\Omega}_{11}= n^{-1}\sum_{i=1}^n u_i(\widehat{\beta}_v)^2\mathbf{z}_{1i}\mathbf{z}_{1i}'
\end{equation}
rather than the $(p\times p)$ upper left sub-matrix of $\widehat{\Omega}$ to estimate this quantity. 
This imposes the assumption that all instruments in $Z_1$ are valid so that no centering is needed, providing greater precision.

\section{Moment Averaging and Post-Selection Estimators}
\label{sec:avg}
Because it is constructed from $\widehat{\tau}$, the FMSC is a random variable, even in the limit.
Combining Corollary \ref{cor:tautau} with Equation \ref{eq:fmsc} gives the following.
\begin{cor}[Limit Distribution of FMSC]
\label{cor:FMSClimit}
	Under Assumptions \ref{assump:drift}, \ref{assump:highlevel} and \ref{assump:Identification}, we have $FMSC_n(S) \rightarrow_d FMSC_S(\tau, M)$, where
		$B(\tau,M) = (\Psi M + \tau)(\Psi M + \tau)' - \Psi \Omega \Psi'$ and 
	\begin{equation*}
		\mbox{FMSC}_S(\tau,M) = \nabla_\theta\mu(\theta_0)'K_S\Xi_S \left\{\left[\begin{array}{cc}0&0\\0& B(\tau,M) \end{array}\right] + \Omega\right\}\Xi_S'K_S'\nabla_\theta\mu(\theta_0).
	\end{equation*}
\end{cor}
This corollary implies that the FMSC is a ``conservative'' rather than ``consistent'' selection procedure.
This lack of consistency is a desirable feature of the FMSC for two reasons.
First, as discussed above, the goal of the FMSC is not to select only correctly specified moment conditions: it is to choose an estimator with a low finite-sample MSE as approximated by AMSE.
The goal of consistent selection is very much at odds with that of controlling estimator risk.
As explained by \cite{Yang2005} and \cite{LeebPoetscher2008}, the worst-case risk of a consistent selection procedure \emph{diverges} with sample size.\footnote{This fact is readily apparent from the results of the simulation study from Section \ref{sec:chooseIVsim}: the consistent criteria, GMM-BIC and HQ, have the highest worst-case RMSE, while the conservative criteria, FMSC and GMM-AIC, have the lowest.}
Second, while we know from both simulation studies \citep{Demetrescu} and analytical examples \citep{LeebPoetscher2005} that selection can dramatically change the sampling distribution of our estimators, invalidating traditional confidence intervals, the asymptotics of consistent selection give the misleading impression that this problem can be ignored.

There are two main problems with applying ``textbook'' confidence intervals post-moment selection.
First is model selection uncertainty: if the data had been slightly different, we would have chosen a different set of moment conditions.
Accordingly, any confidence interval that \emph{conditions} on the selected model must be too short.
Second, textbook confidence intervals ignore the fact that selection is carried out over potentially invalid moment conditions.
Even if our goal were to consistently eliminate such moment conditions, for example by using a consistent criterion such as the GMM-BIC of \cite{Andrews1999}, in finite-samples we would not always be successful.
Because of this, our intervals will be incorrectly centered.
Accounting for these two effects requires a limit theory that accommodates \emph{mixture distributions}: post-selection estimators are randomly-weighted averages of the individual candidate estimators.
Because they choose a single candidate with probability approaching one in the limit, consistent selection procedures make it impossible to represent this phenomenon.
In contrast, conservative selection procedures remain random even as the sample size goes to infinity, allowing us to derive a mixture-of-normals limit distribution and, ultimately, to carry out valid inference post-moment selection.
In the remainder of this section, I derive the asymptotic distribution of generic ``moment average'' estimators and use them to propose simulation-based procedures for post-moment selection inference. 
For certain examples it is possible to analytically characterize the limit distribution of a post-FMSC estimator without resorting to simulation-based methods.
I explore this possibility in detail for my two running examples: OLS versus TSLS and choosing instrumental variables.
I also briefly consider a minimum-AMSE averaging estimator that combines OLS and TSLS. 

\subsection{Moment Average Estimators}
A generic moment average estimator takes the form
\begin{equation}
	\label{eq:avg}
	\widehat{\mu}=\sum_{S \in \mathscr{S}} \widehat{\omega}_S\widehat{\mu}_S
\end{equation}
where $\widehat{\mu}_S = \mu(\widehat{\theta}_S)$ is the estimator of the target parameter $\mu$ under moment set $S$, $\mathscr{S}$ is the collection of all moment sets under consideration, and $\widehat{\omega}_S$ is shorthand for the value of a data-dependent weight function $\widehat{\omega}_S=\omega(\cdot, \cdot)$ evaluated at moment set $S$ and the sample observations $Z_{n1}, \hdots, Z_{nn}$.  
As above $\mu(\cdot)$ is a $\mathbb{R}$-valued, $Z$-almost surely continuous function of $\theta$ that is differentiable in an open neighborhood of $\theta_0$. 
When $\widehat{\omega}_S$ is an indicator, taking on the value one at the moment set moment set that minimizes some moment selection criterion, $\widehat{\mu}$ is a post-moment selection estimator. 
To characterize the limit distribution of $\widehat{\mu}$, I impose the following mild conditions on $\widehat{\omega}_S$, requiring that they sum to one and are ``well-behaved'' in the limit so that I may apply the continuous mapping theorem.
\begin{assump}[Conditions on the Weights]\mbox{}
\label{assump:weights}
\begin{enumerate}[(a)]
	\item $\sum_{S \in \mathscr{S}} \widehat{\omega}_S = 1$, almost surely 
	\item For each $S\in \mathscr{S}$, $\widehat{\omega}_S \rightarrow_d\varphi_S(\tau, M)$, a function of $\tau$, $M$ and consistently estimable constants with at most countably many discontinuities.
\end{enumerate}
\end{assump}

\begin{cor}[Asymptotic Distribution of Moment-Average Estimators]
\label{cor:momentavg}
Under Assumption \ref{assump:weights} and the conditions of Theorem \ref{thm:normality},
	$$\sqrt{n}\left(\widehat{\mu} -  \mu_0\right) \rightarrow_{d}\Lambda(\tau) =  -\nabla_\theta\mu(\theta_0)'\left[\sum_{S \in \mathscr{S}} \varphi_S(\tau,M) K_S\Xi_S\right] \left(M + \left[\begin{array}
	{c} 0 \\ \tau
\end{array} \right]\right).$$
\end{cor}
Notice that the limit random variable from Corollary \ref{cor:momentavg}, denoted $\Lambda(\tau)$, is a \emph{randomly weighted average} of the multivariate normal vector $M$. 
Hence, $\Lambda(\tau)$ is non-normal. 
This is precisely the behavior for which we set out to construct an asymptotic representation.
The conditions of Assumption \ref{assump:weights} are fairly mild. 
Requiring that the weights sum to one ensures that $\widehat{\mu}$ is a consistent estimator of $\mu_0$ and leads to a simpler expression for the limit distribution. 
While somewhat less transparent, the second condition is satisfied by weighting schemes based on a number of familiar moment selection criteria.
It follows immediately from Corollary \ref{cor:FMSClimit}, for example, that the FMSC converges in distribution to a function of $\tau$, $M$ and consistently estimable constants only. 
The same is true for weights based on the $J$-test statistic, as seen from the following result.
\begin{thm}[Distribution of $J$-Statistic under Local Mis-Specification] 
\label{pro:jstat}
	Define the J-test statistic as per usual by $J_n(S)  = n \left[\Xi_S f_n(\widehat{\theta}_S)\right]' \widehat{\Omega}^{-1}\left[\Xi_S f_n(\widehat{\theta}_S)\right]$ where $\widehat{\Omega}^{-1}_S$ is a consistent estimator of $\Omega_S^{-1}$. Then, under the conditions of Theorem \ref{thm:normality}, we have $J_n(S) \rightarrow_dJ_S(\tau, M)$ where
		$$J_S(\tau, M)=[\Omega_S^{-1/2}(M_S + \tau_S)]' (I - P_S)[\Omega_S^{-1/2}\Xi_S(M_S + \tau_S)],$$
$M_S = \Xi_S M$, $\tau_S' = (0', \tau')\Xi_S'$, and $P_S$ is the projection matrix formed from the GMM identifying restrictions $\Omega^{-1/2}_S F_S$.
\end{thm}

Post-selection estimators are merely a special case of moment average estimators.
To see why, consider the weight function
$$\widehat{\omega}_S^{MSC} = \mathbf{1}\left\{\mbox{MSC}_n(S) = \min_{S'\in \mathscr{S}} \mbox{MSC}_n(S')\right\}$$where $\mbox{MSC}_n(S)$ is the value of some moment selection criterion evaluated at the sample observations $Z_{n1}\hdots, Z_{nn}$. 
Now suppose $\mbox{MSC}_n(S) \rightarrow_d\mbox{MSC}_S(\tau,M)$, a function of $\tau$, $M$ and consistently estimable constants only. 
Then, so long as the probability of ties, $P\left\{\mbox{MSC}_S(\tau,M) = \mbox{MSC}_{S'}(\tau,M) \right\}$, is zero for all $S\neq S'$, we have 
	$$\widehat{\omega}_S^{MSC} \rightarrow_d \mathbf{1}\left\{\mbox{MSC}_S(\tau,M) = \min_{S'\in \mathscr{S}} \mbox{MSC}_{S'}(\tau,M)\right\}$$ 
satisfying Assumption \ref{assump:weights} (b). 
Thus, post-selection estimators based on the FMSC, a downward $J$-test procedure, or the GMM moment selection criteria of \cite{Andrews1999} all fall within the ambit of \ref{cor:momentavg}. 
The consistent criteria of \cite{Andrews1999}, however, are not particularly interesting under local mis-specification.\footnote{For more discussion of these criteria, see Section \ref{sec:chooseIVsim} below.}
Intuitively, because they aim to select all valid moment conditions w.p.a.1, we would expect that under Assumption \ref{assump:drift} they choose the full moment set in the limit. 
The following result shows that this intuition is correct.\footnote{This result is a special case of a more general phenomenon: consistent selection procedures cannot detect model violations of order $O(n^{-1/2})$.}
\begin{thm}[Consistent Criteria under Local Mis-Specification]
\label{pro:andrews}
Consider a moment selection criterion of the form $MSC(S) = J_n(S) - h(|S|)\kappa_n$, where $h$ is strictly increasing,  $\lim_{n\rightarrow \infty}\kappa_n = \infty$, and $\kappa_n = o(n)$. Under the conditions of Theorem \ref{thm:normality}, $MSC(S)$ selects the full moment set with probability approaching one.
\end{thm}

\subsection{Digression: Minimum-AMSE Averaging for OLS and TSLS}
\label{sec:momentavgexample}
When competing moment sets have similar criterion values in the population, sampling variation can be \emph{magnified} in the selected estimator.
This motivates the idea of averaging estimators based on different moment conditions rather than selecting them.
To illustrate this idea, I now briefly revisit the OLS versus TSLS example from Section \ref{sec:OLSvsIVExample} and derive an AMSE-optimal weighted average of the two estimators.
Let $\widetilde{\beta}(\omega)$ be a convex combination of the OLS and TSLS estimators, namely  
\begin{equation}
	\widetilde{\beta}(\omega) = \omega \widehat{\beta}_{OLS} + (1 - \omega) \widetilde{\beta}_{TSLS}
\end{equation}
where $\omega \in [0,1]$ is the weight given to the OLS estimator.
\begin{thm}
	\label{thm:OLSvsIVavg} 
	Under the conditions of Theorem \ref{thm:OLSvsIV}, the AMSE of the weighted-average estimator $\sqrt{n}\left[\widehat{\beta}(\omega) - \beta \right]$ is minimized over $\omega \in [0,1]$ by taking $\omega = \omega^*$ where
	$$ \omega^* = \left[1 + \frac{\tau^2/\sigma_x^4}{\sigma_\epsilon^2(1/\gamma^2 - 1/\sigma_x^2)}\right]^{-1} = \left[1 + \frac{\mbox{ABIAS(OLS)}^2}{\mbox{AVAR(TSLS)}-\mbox{AVAR(OLS)}} \right]^{-1}.$$
\end{thm}

The preceding result has several important consequences. 
First, since the variance of the TSLS estimator is always strictly greater than that of the OLS estimator, the optimal value of $\omega$ \emph{cannot} be zero. 
No matter how strong the endogeneity of $x$,
as measured by $\tau$, we should always give some weight to the OLS estimator. 
Second, when $\tau = 0$ the optimal value of $\omega$ is one. If $x$ is exogenous, OLS is strictly preferable to TSLS. 
Third, the optimal weights depend on the strength of the instruments $\mathbf{z}$ as measured by $\gamma$.
All else equal, the stronger the instruments, the less weight we should give to OLS.
To operationalize the AMSE-optimal averaging estimator suggested from Theorem \ref{thm:OLSvsIVavg}, I propose the plug-in estimator 
\begin{equation}
	\widehat{\beta}^*_{AVG} = \widehat{\omega}^* \widehat{\beta}_{OLS} + (1 - \widehat{\omega}^*)\widetilde{\beta}_{TSLS}
	\label{eq:OLSvsIV_AVG1}
\end{equation}
where
\begin{equation}
\widehat{\omega }^* = \left[1 + \frac{\max \left\{0, \; \left(\widehat{\tau}^2 - \widehat{\sigma}_\epsilon^2\widehat{\sigma}_x^2  \left(\widehat{\sigma}_x^2/\widehat{\gamma}^2 - 1 \right) \right)/\;\widehat{\sigma}_x^4 \right\}}{\widehat{\sigma}_\epsilon^2 (1/\widehat{\gamma}^2 - 1/\widehat{\sigma}_x^2)}\right]^{-1}
	\label{eq:OLSvsIV_AVG2}
\end{equation}
This expression employs the same consistent estimators of $\sigma_x^2, \gamma$ and $\sigma_{\epsilon}$ as the FMSC expressions from Section \ref{sec:OLSvsIVExample}.
To ensure that $\widehat{\omega}^*$ lies in the interval $[0,1]$, however, I use a \emph{positive part} estimator for $\tau^2$, namely $\max\{0, \; \widehat{\tau}^2 - \widehat{V}\}$ rather than $\widehat{\tau}^2 - \widehat{V}$.\footnote{While $\widehat{\tau}^2 - \widehat{V}$ is an asymptotically unbiased estimator of $\tau^2$ it \emph{can} be negative.}
In the following section I show how one can construct confidence intervals for $\widehat{\beta}^*$ and related estimators.

\subsection{Inference for Moment-Average Estimators}
Suppose that $K_S$, $\varphi_S$, $\theta_0$, $\Omega$ and $\tau$ were all known. 
Then, by simulating from $M$, as defined in Theorem \ref{thm:normality}, the distribution of $\Lambda(\tau)$, defined in Corollary \ref{cor:momentavg}, could be approximated to arbitrary precision. 
This is the basic intuition that I use to devise inference procedures for moment-average and post-selection estimators.

To operationalize this idea, first consider how we would proceed if we knew \emph{only} the value of $\tau$.  
While $K_S$, $\theta_0$, and $\Omega$ are unknown this presents only a minor difficulty: in their place we can simply substitute the consistent estimators that appeared in the expression for the FMSC above.
To estimate $\varphi_S$, we first need to derive the limit distribution of $\widehat{\omega}_S$, the data-based weights specified by the user. 
As an example, consider the case of moment selection based on the FMSC. Here $\widehat{\omega}_S$ is simply the indicator function
\begin{equation}
	\label{eq:FMSCindicate}
	\widehat{\omega}_S = \mathbf{1}\left\{\mbox{FMSC}_n(S) = \min_{S'\in \mathscr{S}} \mbox{FMSC}_n(S')\right\}
\end{equation}
Substituting estimators of $\Omega$, $K_S$ and $\theta_0$ into $\mbox{FMSC}_S(\tau,M)$, defined in Corollary \ref{cor:FMSClimit}, gives
\begin{equation*}
	\widehat{\mbox{FMSC}}_S(\tau,M) = \nabla_\theta\mu(\widehat{\theta})'\widehat{K}_S\Xi_S \left\{\left[\begin{array}{cc}0&0\\0&\widehat{\mathcal{B}}(\tau,M) \end{array}\right] + \widehat{\Omega}\right\}\Xi_S'\widehat{K}_S'\nabla_\theta\mu(\widehat{\theta})
\end{equation*}
where $\widehat{\mathcal{B}}(\tau,M) = (\widehat{\Psi} M + \tau)(\widehat{\Psi} M + \tau)' - \widehat{\Psi} \widehat{\Omega} \widehat{\Psi}$.
Combining this with Equation \ref{eq:FMSCindicate},
\begin{equation*}
	\widehat{\varphi}_S(\tau,M) = \mathbf{1}\left\{\widehat{\mbox{FMSC}}_S(\tau,M) = \min_{S'\in \mathscr{S}} \widehat{\mbox{FMSC}}_{S'}(\tau,M)\right\}.
\end{equation*}
For GMM-AIC moment selection or selection based on a downward $J$-test, $\varphi_S(\cdot,\cdot)$ may be estimated analogously, following  Theorem \ref{pro:jstat}. 
Continuing to assume for the moment that $\tau$ is known, consider the following algorithm:
\begin{alg}[Simulation-based CI for $\widehat{\mu}$ given $\tau$]
\mbox{}
		\begin{enumerate}
\label{alg:conf_tau_known}
			\item Generate $J$ independent draws $M_j \sim N_{p+q}(0, \widehat{\Omega})$.
			\item Set $\Lambda_j(\tau) = -\nabla_\theta\mu(\widehat{\theta})'\left[\sum_{S \in \mathscr{S}} \widehat{\varphi}_S(\tau,M_j) \widehat{K}_S\Xi_S\right] (M_j + \tau)$.
			\item Using $\{\Lambda_j(\tau)\}_{j=1}^J$, calculate $\widehat{a}(\tau)$, $\widehat{b}(\tau)$ such that
		$P\left\{ \widehat{a}(\tau) \leq\Lambda(\tau)\leq \widehat{b}(\tau) \right\} = 1 - \alpha$.
  \item Define the interval 
    $ \mbox{CI}_{sim}=\left[ \widehat{\mu} - \widehat{b}(\tau)/\sqrt{n}, \quad \widehat{\mu} - \widehat{a}(\tau)/\sqrt{n} \right]$.
		\end{enumerate}
\end{alg}

Given knowledge of $\tau$, Algorithm \ref{alg:conf_tau_known} yields valid inference for $\mu$.
The problem, of course, is that $\tau$ is unknown and cannot even be consistently estimated.
One idea would be to substitute the asymptotically unbiased estimator $\widehat{\tau}$ from \ref{thm:tau} in place of the unknown $\tau$.
This gives rise to a procedure that I call the ``1-Step'' confidence interval:

\begin{alg}[1-Step CI] 
  \label{alg:1step}
Carry out of Algorithm \ref{alg:conf_tau_known} with $\tau$ set equal to the estimator $\widehat{\tau}$ from Theorem \ref{thm:tau}, yielding 
$ \widehat{\mbox{CI}}_{1}=\left[ \widehat{\mu} - \widehat{b}(\widehat{\tau})/\sqrt{n}, \quad \widehat{\mu} - \widehat{a}(\widehat{\tau})/\sqrt{n} \right]$.
\end{alg}

The 1-Step interval defined in Algorithm \ref{alg:1step} is conceptually simple, easy to compute, and can perform well in practice, as I explore below.
But as it fails to account for sampling uncertainty in $\widehat{\tau}$,  it does \emph{not} necessarily yield asymptotically valid inference for $\mu$.
Fully valid inference requires the addition of a second step to the algorithm and comes at a cost: conservative rather than exact inference.
In particular, the following procedure is guaranteed to yield an interval with asymptotic coverage probability of \emph{at least} $(1-\alpha-\delta)\times 100\%$.

\begin{alg}[2-Step CI]
\label{alg:conf}
\mbox{}
\begin{enumerate}
  \item Construct a $(1-\delta)\times 100\%$ confidence region $\mathscr{T}$ for $\tau$ using Theorem \ref{thm:tau}. 
  \item For each $\tau^* \in \mathscr{T}$ carry out Algorithm \ref{alg:conf_tau_known}, yielding a $(1 - \alpha)\times 100\%$ confidence interval $\left[\widehat{a}(\tau^*),\widehat{b}(\tau^*)\right]$ for $\Lambda(\tau^*)$.  
	\item Set $\displaystyle \widehat{a}_{min}=\min_{\tau^* \in \mathscr{T}} \widehat{a}(\tau^*)$ and $\displaystyle \widehat{b}_{max}= \max_{\tau^* \in \mathscr{T}} \widehat{b}(\tau^*)$. 
	\item Construct the interval 
    $ \widehat{\mbox{CI}}_{2}=\left[ \widehat{\mu} - \widehat{b}_{max}/\sqrt{n}, \quad \widehat{\mu} - \widehat{a}_{min}/\sqrt{n} \right]$
\end{enumerate}
\end{alg}

\begin{thm}[2-Step CI]
\label{thm:sim}
Let $\widehat{\Psi}$, $\widehat{\Omega}$, $\widehat{\theta}$, $\widehat{K}_S$, $\widehat{\varphi}_S$ be consistent estimators of $\Psi$, $\Omega$, $\theta_0$, $K_S$, $\varphi_S$ and define 
$\Delta_n(\widehat{\tau},\tau^*) = \left(\widehat{\tau} - \tau^*\right)' (\widehat{\Psi}\widehat{\Omega}\widehat{\Psi}')^{-1} \left(\widehat{\tau} - \tau^*\right)$ 
and 
$\mathscr{T}(\widehat{\tau},\delta) = \left\{\tau^* \colon  \Delta_n(\widehat{\tau},\tau^*) \leq \chi^2_q(\delta)\right\}$
where $\chi^2_q(\delta)$ denotes the $1-\delta$ quantile of a $\chi^2$ distribution with $q$ degrees of freedom.
Then, the interval $\mbox{CI}_{2}$ defined in Algorithm \ref{alg:conf} has asymptotic coverage probability no less than $1-(\alpha + \delta)$ as $J,n\rightarrow \infty$.
\end{thm}

\subsection{A Special Case of Post-FMSC Inference}
\label{sec:limitexperiment}
The preceding section presented two confidence interval that account for the effects of moment selection on subsequent inference.
The 1-Step interval is intuitive and computationally straightforward but lacks theoretical guarantees, while the 2-Step interval guarantees asymptotically valid inference at the cost of greater computational complexity and conservatism.
To better understand these methods and the trade-offs involved in deciding between them, I now specialize them to the two examples of FMSC selection that appear in the simulation studies described below.
The structure of these examples allows us to bypass Algorithm \ref{alg:conf} and characterize the asymptotic properties of various proposals for post-FMSC without resorting to Monte Carlo simulations.
Because this section presents asymptotic results, I treat any consistently estimable quantity that appears in a limit distribution as known.

In both the OLS versus IV example from Section \ref{sec:OLSvsIVExample} and the slightly simplified version of the choosing instrument variables example implemented in Section \ref{sec:chooseIVsim}, the post-FMSC estimator $\widehat{\beta}_{FMSC}$ converges to a very convenient limit experiment.\footnote{The simplified version of the choosing instrumental variables example considers a single potentially endogenous instrument and imposes homoskedasticity.
  For more details see Section \ref{sec:chooseIVsim} and Online Appendix \ref{append:limitexperiment}.}
In particular, 
\begin{equation}
  \sqrt{n}(\widehat{\beta}_{FMSC} - \beta) \overset{d}{\rightarrow}  \mathbf{1}\left\{ |T|<\sigma \sqrt{2} \right\} U +  \mathbf{1}\left\{ |T|\geq\sigma \sqrt{2} \right\} V.
  \label{eq:FMSCLimitExperiment}
\end{equation}
with
\begin{equation}
  T = \sigma Z_1 + \tau, \quad
  U = \eta Z_2 + c\tau, \quad
  V = \eta Z_2 - c\sigma Z_1
\end{equation}
where $Z_1, Z_2$ are independent standard normal random variables, $\eta$, $\sigma$ and $c$ are consistently estimable constants, and $\tau$ is the local mis-specification parameter.
This representation allows us to tabulate the asymptotic distribution, $F_{FMSC}$ as follows: 
\begin{eqnarray}
  F_{FMSC}(x) &=& G(x) + H_1(x) + H_2(x) \\
  \label{eq:FFMSC}
  G(x) &=& \Phi\left( \frac{x - c\tau}{\eta} \right)\left[ \Phi( \sqrt{2} - \tau/\sigma) -  \Phi( -\sqrt{2} - \tau/\sigma )\right]\\
  \label{eq:GFMSC}
  H_1(x) &=& \frac{1}{\sigma}\int_{-\infty}^{-\sigma\sqrt{2} - \tau} \Phi\left( \frac{x + ct}{\eta}\right)\varphi(t/\sigma)\; dt\\
  \label{eq:H1FMSC}
  H_2(x) &=& \frac{1}{\sigma}\int^{+\infty}_{\sigma\sqrt{2} - \tau} \Phi\left( \frac{x + ct}{\eta}\right)\varphi(t/\sigma)\; dt
  \label{eq:H2FMSC}
\end{eqnarray}
where $\Phi$ is the CDF and $\varphi$ the pdf of a standard normal random variable.
Note that the limit distribution of the post-FMSC distribution depends on $\tau$ in addition to the consistently estimable quantities $\sigma, \eta, c$ although I suppress this dependence to simplify the notation.
While these expressions lack a closed form $G$, $H_1$ and $H_2$ are easy to compute, allowing us to calculate both $F_{FMSC}$ and the corresponding quantile function $Q_{FMSC}$\footnote{I provide code to evaluate both $F_{FMSC}$ and $Q_{FMSC}$ in my R package \texttt{fmscr}, available at \url{https://github.com/fditraglia/fmscr}.}

The ability to compute $F_{FMSC}$ and $Q_{FMSC}$ allows us to answer a number of important questions about post-FMSC inference.
First, suppose that we were to carry out FMSC selection and then construct a $(1 - \alpha) \times 100\%$ confidence interval \emph{conditional} in the selected estimator, completely ignoring the effects of the moment selection step.
What would be the resulting asymptotic coverage probability and width of such a ``na\"{i}ve'' confidence interval procedure?
Using calculations similar to those used above in the expression for $F_{FMSC}$, we find that the coverage probability of this na\"{i}ve interval is given by
\begin{eqnarray*}
  \mbox{CP}_{Naive}(\alpha) &=& G(u_\alpha) - G(-u_\alpha) +  H_1(\ell_\alpha) - H_2(-\ell_\alpha) +  H_2(\ell_\alpha) - H_2(-\ell_\alpha) \\
  u_\alpha &=& \Phi(1 - \alpha/2)\; \eta\\
  \ell_{\alpha} &=& \Phi(1 - \alpha/2) \sqrt{\eta^2 + c^2\sigma^2}
\end{eqnarray*}
where $G$, $H_1$, $H_2$ are as defined in Equations \ref{eq:GFMSC}--\ref{eq:H2FMSC}.
And since the width of this na\"{i}ve CI equals that of the textbook interval for $\widehat{\beta}$ when $|\widehat{\tau}|<\sigma\sqrt{2}$ and that of the textbook interval for $\widetilde{\beta}$ otherwise, we have
\begin{equation*}
  \frac{E\left[ \mbox{Width}_{Naive}(\alpha) \right]}{\mbox{Width}_{Valid}(\alpha)} = 1 + \left[ \Phi( \sqrt{2} - \tau/\sigma) -  \Phi( -\sqrt{2} - \tau/\sigma )\right]\left( \sqrt{\frac{\eta^2}{\eta^2 + c^2 \sigma^2}} - 1 \right)
\end{equation*}
where $\mbox{Width}_{Valid}(\alpha)$ is the width of a standard, textbook confidence interval for $\widetilde{\beta}$.

To evaluate these expressions we need values for $c, \eta^2, \sigma^2$ and $\tau$.
For the remainder of this section I will consider the parameter values that correspond to the simulation experiments presented below in Section \ref{sec:simulations}.
For the OLS versus TSLS example we have $c=1$, $\eta^2=1$ and $\sigma^2 = (1-\pi^2)/\pi^2$ where $\pi^2$ denotes the population first-stage R-squared for the TSLS estimator. 
For the choosing IVs example we have $c =\gamma/(\gamma^2 +1/9)$, $\eta^2 = 1/(\gamma^2 + 1/9)$ and $\sigma^2 = 1 + 9\gamma^2$ where $\gamma^2$ is the increase in the population first-stage R-squared of the TSLS estimator from \emph{adding} $w$ to the instrument set.\footnote{The population first-stage R-squared with only $\mathbf{z}$ in the instument set is $1/9$.}

Table \ref{tab:LimitNaiveCover} presents the asymptotic coverage probability and Table \ref{tab:LimitNaiveWidth} the expected relative width of the na\"{i}ve confidence interval procedure for a variety of values of $\tau$ and $\alpha$ for each of the two examples.
For the OLS versus TSLS example, I allow $\pi^2$ to vary while for the choosing IVs example I allow $\gamma^2$ to vary.
Note that the relative expected width does not depend on $\alpha$.
In terms of coverage probability, the na\"{i}ve interval performs very poorly: in some regions of the parameter space the actual coverage is very close to the nominal level, while in others it is far lower.
These striking size distortions, which echo the findings of \cite{Guggenberger2010} and \cite{Guggenberger2012}, provide a strong argument against the use of the na\"{i}ve interval.
Its attraction, of course, is width: the na\"{i}ve interval can be dramatically shorter than the corresponding ``textbook'' confidence interval for the valid estimator.

Is there any way to construct a post-FMSC confidence interval that does not suffer from the egregious size distortions of the na\"{i}ve interval but is still shorter than the textbook interval for the valid estimator?
As a first step towards answering this question, Table \ref{tab:WidthInfeasible} presents the relative width of the shortest possible \emph{infeasible} post-FMSC confidence interval constructed directly from $Q_{FMSC}$.
This interval has asymptotic coverage probability \emph{exactly} equal to its nominal level as it correctly accounts for the effect of moment selection on the asymptotic distribution of the estimators.
Unfortunately it cannot be used in practice because it requires knowledge of $\tau$, for which no consistent estimator exists.
As such, this interval serves as a benchmark against which to judge various feasible procedures that do not require knowledge of $\tau$.
For certain parameter values this interval is shorter than the valid interval but the improvement is not uniform and indeed cannot be.
Just as the FMSC itself cannot provide a uniform reduction in AMSE relative to the valid estimator, the infeasible post-FMSC cannot provide a corresponding reduction in width.
In both cases, however, improvements are possible when $\tau$ is expected to be small, the setting in which this paper assumes that an applied researcher finds herself. 
The potential reductions in width can be particularly dramatic for larger values of $\alpha$.
The question remains: is there any way to capture these gains using a \emph{feasible} procedure?

\begin{table}[h]
  \footnotesize
  \centering
  \begin{subtable}{0.48\textwidth}
    \caption{OLS versus TSLS}
    \begin{tabular}{r|rrrrrr}
\hline\hline
 &\multicolumn{6}{c}{$\tau$} \\ 
 $\alpha = 0.05$ & $0$ & $1$ & $2$ & $3$ & $4$ & $5$ \\ 
 \hline$0.1$ & $91$ & $81$ & $57$ & $41$ & $45$ & $58$\\ 
$\pi^2\;\;\;$ $0.2$ & $91$ & $83$ & $63$ & $58$ & $70$ & $84$\\ 
$0.3$ & $92$ & $84$ & $69$ & $73$ & $86$ & $93$\\ 
$0.4$ & $92$ & $85$ & $76$ & $84$ & $93$ & $95$\\ 
 \hline 
 \end{tabular}
 
 \vspace{2em} 
 
\begin{tabular}{r|rrrrrr}
\hline\hline
 &\multicolumn{6}{c}{$\tau$} \\ 
 $\alpha = 0.1$ & $0$ & $1$ & $2$ & $3$ & $4$ & $5$ \\ 
 \hline$0.1$ & $83$ & $70$ & $45$ & $35$ & $42$ & $55$\\ 
$\pi^2\;\;\;$ $0.2$ & $84$ & $72$ & $53$ & $52$ & $67$ & $81$\\ 
$0.3$ & $85$ & $74$ & $60$ & $68$ & $83$ & $89$\\ 
$0.4$ & $86$ & $76$ & $68$ & $80$ & $89$ & $90$\\ 
 \hline 
 \end{tabular}
 
 \vspace{2em} 
 
\begin{tabular}{r|rrrrrr}
\hline\hline
 &\multicolumn{6}{c}{$\tau$} \\ 
 $\alpha = 0.2$ & $0$ & $1$ & $2$ & $3$ & $4$ & $5$ \\ 
 \hline$0.1$ & $70$ & $54$ & $31$ & $27$ & $37$ & $50$\\ 
$\pi^2\;\;\;$ $0.2$ & $71$ & $57$ & $39$ & $45$ & $62$ & $74$\\ 
$0.3$ & $73$ & $59$ & $49$ & $61$ & $75$ & $79$\\ 
$0.4$ & $74$ & $62$ & $58$ & $72$ & $79$ & $80$\\ 
 \hline 
 \end{tabular}
  \end{subtable}
  ~
  \begin{subtable}{0.48\textwidth}
    \caption{Choosing IVs}
    \begin{tabular}{r|rrrrrr}
\hline\hline
 &\multicolumn{6}{c}{$\tau$} \\ 
 $\alpha = 0.05$ & $0$ & $1$ & $2$ & $3$ & $4$ & $5$ \\ 
 \hline$0.1$ & $93$ & $89$ & $84$ & $85$ & $91$ & $94$\\ 
$\gamma^2\;\;\;$ $0.2$ & $92$ & $87$ & $76$ & $74$ & $83$ & $91$\\ 
$0.3$ & $92$ & $85$ & $71$ & $65$ & $74$ & $86$\\ 
$0.4$ & $91$ & $85$ & $68$ & $59$ & $67$ & $80$\\ 
 \hline 
 \end{tabular}
 
 \vspace{2em} 
 
\begin{tabular}{r|rrrrrr}
\hline\hline
 &\multicolumn{6}{c}{$\tau$} \\ 
 $\alpha = 0.1$ & $0$ & $1$ & $2$ & $3$ & $4$ & $5$ \\ 
 \hline$0.1$ & $87$ & $82$ & $76$ & $79$ & $86$ & $89$\\ 
$\gamma^2\;\;\;$ $0.2$ & $85$ & $78$ & $66$ & $67$ & $79$ & $87$\\ 
$0.3$ & $84$ & $76$ & $61$ & $59$ & $71$ & $82$\\ 
$0.4$ & $84$ & $75$ & $57$ & $52$ & $63$ & $77$\\ 
 \hline 
 \end{tabular}
 
 \vspace{2em} 
 
\begin{tabular}{r|rrrrrr}
\hline\hline
 &\multicolumn{6}{c}{$\tau$} \\ 
 $\alpha = 0.2$ & $0$ & $1$ & $2$ & $3$ & $4$ & $5$ \\ 
 \hline$0.1$ & $75$ & $69$ & $64$ & $70$ & $77$ & $80$\\ 
$\gamma^2\;\;\;$ $0.2$ & $73$ & $64$ & $53$ & $59$ & $71$ & $78$\\ 
$0.3$ & $72$ & $62$ & $47$ & $50$ & $64$ & $75$\\ 
$0.4$ & $72$ & $60$ & $43$ & $44$ & $58$ & $71$\\ 
 \hline 
 \end{tabular}
  \end{subtable}
  \caption{Asymptotic coverage probability of Na\"{i}ve $(1-\alpha)\times 100\%$ confidence interval.}
  \label{tab:LimitNaiveCover}
\end{table}

\begin{table}[h]
  \footnotesize
  \centering
  \begin{subtable}{0.48\textwidth}
    \caption{OLS versus TSLS}
    \begin{tabular}{r|rrrrrr}
\hline\hline
 &\multicolumn{6}{c}{$\tau$} \\ 
  & $0$ & $1$ & $2$ & $3$ & $4$ & $5$ \\ 
 \hline$0.1$ & $ 42$ & $ 44$ & $ 48$ & $ 55$ & $ 64$ & $ 73$\\ 
$\pi^2\;\;\;$ $0.2$ & $ 53$ & $ 56$ & $ 64$ & $ 74$ & $ 85$ & $ 92$\\ 
$0.3$ & $ 62$ & $ 66$ & $ 76$ & $ 87$ & $ 95$ & $ 99$\\ 
$0.4$ & $ 69$ & $ 74$ & $ 85$ & $ 94$ & $ 99$ & $100$\\ 
 \hline 
 \end{tabular}
  \end{subtable}
  ~
  \begin{subtable}{0.48\textwidth}
    \caption{Choosing IVs}
    \begin{tabular}{r|rrrrrr}
\hline\hline
 &\multicolumn{6}{c}{$\tau$} \\ 
  & $0$ & $1$ & $2$ & $3$ & $4$ & $5$ \\ 
 \hline$0.1$ & $ 77$ & $ 80$ & $ 87$ & $ 94$ & $ 98$ & $100$\\ 
$\gamma^2\;\;\;$ $0.2$ & $ 66$ & $ 69$ & $ 77$ & $ 86$ & $ 93$ & $ 98$\\ 
$0.3$ & $ 60$ & $ 62$ & $ 69$ & $ 79$ & $ 88$ & $ 94$\\ 
$0.4$ & $ 55$ & $ 57$ & $ 64$ & $ 73$ & $ 83$ & $ 90$\\ 
 \hline 
 \end{tabular}
  \end{subtable}
  \caption{Asymptotic expected width of na\"{i}ve confidence interval relative to that of the valid estimator. Values are given in percentage points.}
  \label{tab:LimitNaiveWidth}
\end{table}

\begin{table}[h]
  \footnotesize
  \centering
  \begin{subtable}{0.48\textwidth}
    \caption{OLS versus TSLS}
    \begin{tabular}{r|rrrrrr}
\hline\hline
 &\multicolumn{6}{c}{$\tau$} \\ 
 $\alpha = 0.05$ & $0$ & $1$ & $2$ & $3$ & $4$ & $5$ \\ 
 \hline$0.1$ & $ 99$ & $ 92$ & $ 85$ & $ 89$ & $ 95$ & $101$\\ 
$\pi^2\;\;\;$ $0.2$ & $ 97$ & $ 91$ & $ 94$ & $102$ & $110$ & $117$\\ 
$0.3$ & $ 94$ & $ 94$ & $102$ & $111$ & $117$ & $109$\\ 
$0.4$ & $ 92$ & $ 97$ & $107$ & $114$ & $107$ & $100$\\ 
 \hline 
 \end{tabular}
 
 \vspace{2em} 
 
\begin{tabular}{r|rrrrrr}
\hline\hline
 &\multicolumn{6}{c}{$\tau$} \\ 
 $\alpha = 0.1$ & $0$ & $1$ & $2$ & $3$ & $4$ & $5$ \\ 
 \hline$0.1$ & $ 88$ & $ 81$ & $ 85$ & $ 91$ & $ 99$ & $107$\\ 
$\pi^2\;\;\;$ $0.2$ & $ 89$ & $ 88$ & $ 97$ & $107$ & $116$ & $123$\\ 
$0.3$ & $ 86$ & $ 93$ & $105$ & $115$ & $119$ & $103$\\ 
$0.4$ & $ 87$ & $ 98$ & $111$ & $116$ & $104$ & $100$\\ 
 \hline 
 \end{tabular}
 
 \vspace{2em} 
 
\begin{tabular}{r|rrrrrr}
\hline\hline
 &\multicolumn{6}{c}{$\tau$} \\ 
 $\alpha = 0.2$ & $0$ & $1$ & $2$ & $3$ & $4$ & $5$ \\ 
 \hline$0.1$ & $ 48$ & $ 55$ & $ 84$ & $ 96$ & $106$ & $116$\\ 
$\pi^2\;\;\;$ $0.2$ & $ 65$ & $ 80$ & $101$ & $114$ & $125$ & $117$\\ 
$0.3$ & $ 74$ & $ 90$ & $111$ & $123$ & $112$ & $101$\\ 
$0.4$ & $ 80$ & $ 97$ & $116$ & $115$ & $102$ & $100$\\ 
 \hline 
 \end{tabular}
  \end{subtable}
  ~
  \begin{subtable}{0.48\textwidth}
    \caption{Choosing IVs}
    \begin{tabular}{r|rrrrrr}
\hline\hline
 &\multicolumn{6}{c}{$\tau$} \\ 
 $\alpha = 0.05$ & $0$ & $1$ & $2$ & $3$ & $4$ & $5$ \\ 
 \hline$0.1$ & $ 92$ & $ 97$ & $106$ & $111$ & $109$ & $102$\\ 
$\gamma^2\;\;\;$ $0.2$ & $ 93$ & $ 94$ & $101$ & $109$ & $115$ & $114$\\ 
$0.3$ & $ 95$ & $ 93$ & $ 97$ & $105$ & $112$ & $117$\\ 
$0.4$ & $ 97$ & $ 92$ & $ 94$ & $101$ & $108$ & $115$\\ 
 \hline 
 \end{tabular}
 
 \vspace{2em} 
 
\begin{tabular}{r|rrrrrr}
\hline\hline
 &\multicolumn{6}{c}{$\tau$} \\ 
 $\alpha = 0.1$ & $0$ & $1$ & $2$ & $3$ & $4$ & $5$ \\ 
 \hline$0.1$ & $ 89$ & $ 97$ & $108$ & $113$ & $108$ & $101$\\ 
$\gamma^2\;\;\;$ $0.2$ & $ 86$ & $ 93$ & $104$ & $113$ & $118$ & $109$\\ 
$0.3$ & $ 86$ & $ 90$ & $100$ & $109$ & $117$ & $121$\\ 
$0.4$ & $ 88$ & $ 88$ & $ 96$ & $105$ & $114$ & $121$\\ 
 \hline 
 \end{tabular}
 
 \vspace{2em} 
 
\begin{tabular}{r|rrrrrr}
\hline\hline
 &\multicolumn{6}{c}{$\tau$} \\ 
 $\alpha = 0.2$ & $0$ & $1$ & $2$ & $3$ & $4$ & $5$ \\ 
 \hline$0.1$ & $ 86$ & $ 96$ & $111$ & $115$ & $105$ & $101$\\ 
$\gamma^2\;\;\;$ $0.2$ & $ 78$ & $ 89$ & $108$ & $119$ & $118$ & $104$\\ 
$0.3$ & $ 72$ & $ 84$ & $103$ & $116$ & $125$ & $112$\\ 
$0.4$ & $ 67$ & $ 79$ & $ 99$ & $112$ & $123$ & $128$\\ 
 \hline 
 \end{tabular}
  \end{subtable}
  \caption{Width of shortest possible $(1-\alpha)\times 100\%$ post-FMSC confidence interval constructed directly from $Q_{FMSC}$ using knowledge of $\tau$. Values are given in percentage points.}
  \label{tab:WidthInfeasible}
\end{table}

Now, consider the 2-Step confidence interval procedure from Algorithm \ref{alg:conf}.
We can implement an equivalent procedure without simulation as follows.
First we construct a $(1-\alpha_1)\times 100\%$ confidence interval for $\widehat{\tau}$ using $T = \sigma Z_1 + \tau$ where $Z_1$ is standard normal.
Next we construct a $(1-\alpha_2)\times 100\%$ based on $Q_{FMSC}$ for each $\tau^*$ in this interval.
Finally we take the upper and lower bounds over all of the resulting intervals.
This interval is guaranteed to have asymptotic coverage probability of at least $1 - (\alpha_1 + \alpha_2)$ by an argument essentially identical to the proof of Theorem \ref{thm:sim}.
Protection against under-coverage, however, comes at the expense of extreme conservatism, particularly for larger values of $\alpha$.
Numerical values for the coverage and median expected with of this interval appear in Online Appendix \ref{append:limitexperiment_2step}. 
From both the numerical calculations and the theoretical result given in Theorem \ref{thm:sim} we see that the 2-Step systematically \emph{over-covers} and hence \emph{cannot} produce an interval shorter than the textbook CI for the valid estimator.

Now consider the 1-Step confidence interval from Algorithm \ref{alg:1step}. 
Rather than first constructing a confidence region for $\tau$ and then taking upper and lower bounds, 1-Step interval simply takes $\widehat{\tau}$ in place of $\tau$ and then constructs a confidence interval from $Q_{FMSC}$ exactly as in the infeasible interval described above.\footnote{As in the construction of the na\"{i}ve interval, I take the shortest possible interval based on $Q_{FMSC}$ rather than an equal-tailed interval. Additional results for an equal-tailed version of this one-step procedure are available upon request. Their performance is similar.}
Unlike its 2-Step counterpart, this interval comes with no generic theoretical guarantees, so I use the characterization from above to directly calculate its asymptotic coverage and expected relative width.
The results appear in Tables \ref{tab:Limit1StepShortOLSvsIV} and \ref{tab:Limit1StepShortChooseIVs}.
The 1-Step interval effectively ``splits the difference'' between the two-step interval and the na\"{i}ve procedure. 
While it can under-cover, the size distortions are quite small, particularly for $\alpha=0.1$ and $0.05$.
At the same time, when $\tau$ is relatively small this procedure can yield shorter intervals.
While a full investigation of this phenomenon is beyond the scope of the present paper, these calculations suggest a plausible way forward for post-FMSC inference that is less conservative than the two-step procedure from Algorithm \ref{alg:conf} by directly calculating the relevant quantities from the limit distribution of interest.
This is possible because $\pi$ and $\gamma^2$ are both consistently estimable.
And for any particular value of these parameters, the worst-case value of $\tau$ is \emph{interior}.
Using this idea, one could imagine specifying a maximum allowable size distortion and then designing a confidence interval to minimize width, possibly incorporating some prior restriction on the likely magnitude of $\tau$.
Just as the FMSC aims to achieve a favorable trade-off between bias and variance, such a confidence interval procedure could aim to achieve a favorable trade-off between width and coverage.
It would also be interesting to pursue analogous calculations for the minimum AMSE averaging estimator from Section \ref{sec:momentavgexample}.

\begin{table}[h]
  \footnotesize
  \centering
  \begin{subtable}{0.48\textwidth}
    \caption{Coverage Probability}
    \begin{tabular}{r|rrrrrr}
\hline\hline
 &\multicolumn{6}{c}{$\tau$} \\ 
 $\alpha = 0.05$ & $0$ & $1$ & $2$ & $3$ & $4$ & $5$ \\ 
 \hline$0.1$ & $93$ & $94$ & $95$ & $94$ & $91$ & $90$\\ 
$\pi^2\;\;\;$ $0.2$ & $95$ & $95$ & $95$ & $93$ & $91$ & $91$\\ 
$0.3$ & $95$ & $96$ & $94$ & $92$ & $92$ & $94$\\ 
$0.4$ & $96$ & $95$ & $94$ & $93$ & $95$ & $95$\\ 
 \hline 
 \end{tabular}
 
 \vspace{2em} 
 
\begin{tabular}{r|rrrrrr}
\hline\hline
 &\multicolumn{6}{c}{$\tau$} \\ 
 $\alpha = 0.1$ & $0$ & $1$ & $2$ & $3$ & $4$ & $5$ \\ 
 \hline$0.1$ & $89$ & $89$ & $88$ & $86$ & $82$ & $80$\\ 
$\pi^2\;\;\;$ $0.2$ & $91$ & $91$ & $88$ & $85$ & $83$ & $85$\\ 
$0.3$ & $92$ & $91$ & $87$ & $85$ & $87$ & $90$\\ 
$0.4$ & $92$ & $90$ & $87$ & $87$ & $90$ & $91$\\ 
 \hline 
 \end{tabular}
 
 \vspace{2em} 
 
\begin{tabular}{r|rrrrrr}
\hline\hline
 &\multicolumn{6}{c}{$\tau$} \\ 
 $\alpha = 0.2$ & $0$ & $1$ & $2$ & $3$ & $4$ & $5$ \\ 
 \hline$0.1$ & $84$ & $80$ & $71$ & $67$ & $65$ & $64$\\ 
$\pi^2\;\;\;$ $0.2$ & $85$ & $80$ & $71$ & $70$ & $70$ & $76$\\ 
$0.3$ & $84$ & $79$ & $73$ & $72$ & $78$ & $81$\\ 
$0.4$ & $84$ & $79$ & $74$ & $77$ & $81$ & $81$\\ 
 \hline 
 \end{tabular}
  \end{subtable}
  ~
  \begin{subtable}{0.48\textwidth}
    \caption{Relative Width}
    \begin{tabular}{r|rrrrrr}
\hline\hline
 &\multicolumn{6}{c}{$\tau$} \\ 
 $\alpha = 0.05$ & $0$ & $1$ & $2$ & $3$ & $4$ & $5$ \\ 
 \hline$0.1$ & $ 93$ & $ 93$ & $ 95$ & $ 97$ & $ 99$ & $102$\\ 
$\pi^2\;\;\;$ $0.2$ & $ 96$ & $ 97$ & $ 99$ & $104$ & $106$ & $109$\\ 
$0.3$ & $ 97$ & $ 99$ & $102$ & $106$ & $108$ & $107$\\ 
$0.4$ & $ 98$ & $100$ & $105$ & $108$ & $106$ & $103$\\ 
 \hline 
 \end{tabular}
 
 \vspace{2em} 
 
\begin{tabular}{r|rrrrrr}
\hline\hline
 &\multicolumn{6}{c}{$\tau$} \\ 
 $\alpha = 0.1$ & $0$ & $1$ & $2$ & $3$ & $4$ & $5$ \\ 
 \hline$0.1$ & $ 90$ & $ 91$ & $ 92$ & $ 97$ & $ 99$ & $102$\\ 
$\pi^2\;\;\;$ $0.2$ & $ 94$ & $ 96$ & $100$ & $105$ & $108$ & $110$\\ 
$0.3$ & $ 96$ & $100$ & $104$ & $108$ & $109$ & $106$\\ 
$0.4$ & $ 97$ & $101$ & $106$ & $108$ & $106$ & $103$\\ 
 \hline 
 \end{tabular}
 
 \vspace{2em} 
 
\begin{tabular}{r|rrrrrr}
\hline\hline
 &\multicolumn{6}{c}{$\tau$} \\ 
 $\alpha = 0.2$ & $0$ & $1$ & $2$ & $3$ & $4$ & $5$ \\ 
 \hline$0.1$ & $ 83$ & $ 84$ & $ 87$ & $ 93$ & $ 99$ & $103$\\ 
$\pi^2\;\;\;$ $0.2$ & $ 91$ & $ 92$ & $ 96$ & $105$ & $109$ & $110$\\ 
$0.3$ & $ 93$ & $ 97$ & $104$ & $109$ & $108$ & $106$\\ 
$0.4$ & $ 95$ & $100$ & $107$ & $108$ & $105$ & $102$\\ 
 \hline 
 \end{tabular}
  \end{subtable}
  \caption{OLS vs TSLS Example: shortest 1-Step CI} 
  \label{tab:Limit1StepShortOLSvsIV}
\end{table}

\begin{table}[h]
  \footnotesize
  \centering
  \begin{subtable}{0.48\textwidth}
    \caption{Coverage Probability}
    \begin{tabular}{r|rrrrrr}
\hline\hline
 &\multicolumn{6}{c}{$\tau$} \\ 
 $\alpha = 0.05$ & $0$ & $1$ & $2$ & $3$ & $4$ & $5$ \\ 
 \hline$0.1$ & $96$ & $95$ & $94$ & $93$ & $94$ & $95$\\ 
$\gamma^2\;\;\;$ $0.2$ & $96$ & $96$ & $95$ & $93$ & $93$ & $94$\\ 
$0.3$ & $95$ & $95$ & $95$ & $93$ & $92$ & $92$\\ 
$0.4$ & $95$ & $95$ & $95$ & $94$ & $92$ & $91$\\ 
 \hline 
 \end{tabular}
 
 \vspace{2em} 
 
\begin{tabular}{r|rrrrrr}
\hline\hline
 &\multicolumn{6}{c}{$\tau$} \\ 
 $\alpha = 0.1$ & $0$ & $1$ & $2$ & $3$ & $4$ & $5$ \\ 
 \hline$0.1$ & $92$ & $90$ & $88$ & $88$ & $89$ & $91$\\ 
$\gamma^2\;\;\;$ $0.2$ & $92$ & $91$ & $88$ & $86$ & $86$ & $89$\\ 
$0.3$ & $92$ & $91$ & $89$ & $86$ & $85$ & $87$\\ 
$0.4$ & $91$ & $91$ & $89$ & $86$ & $84$ & $84$\\ 
 \hline 
 \end{tabular}
 
 \vspace{2em} 
 
\begin{tabular}{r|rrrrrr}
\hline\hline
 &\multicolumn{6}{c}{$\tau$} \\ 
 $\alpha = 0.2$ & $0$ & $1$ & $2$ & $3$ & $4$ & $5$ \\ 
 \hline$0.1$ & $83$ & $80$ & $76$ & $77$ & $80$ & $81$\\ 
$\gamma^2\;\;\;$ $0.2$ & $84$ & $80$ & $75$ & $73$ & $76$ & $80$\\ 
$0.3$ & $85$ & $81$ & $75$ & $71$ & $72$ & $77$\\ 
$0.4$ & $84$ & $80$ & $73$ & $72$ & $69$ & $74$\\ 
 \hline 
 \end{tabular}
  \end{subtable}
  ~
  \begin{subtable}{0.48\textwidth}
    \caption{Relative Width}
    \begin{tabular}{r|rrrrrr}
\hline\hline
 &\multicolumn{6}{c}{$\tau$} \\ 
 $\alpha = 0.05$ & $0$ & $1$ & $2$ & $3$ & $4$ & $5$ \\ 
 \hline$0.1$ & $ 98$ & $100$ & $104$ & $106$ & $106$ & $104$\\ 
$\gamma^2\;\;\;$ $0.2$ & $ 97$ & $ 99$ & $103$ & $106$ & $108$ & $108$\\ 
$0.3$ & $ 97$ & $ 98$ & $101$ & $104$ & $107$ & $109$\\ 
$0.4$ & $ 97$ & $ 97$ & $ 99$ & $103$ & $106$ & $108$\\ 
 \hline 
 \end{tabular}
 
 \vspace{2em} 
 
\begin{tabular}{r|rrrrrr}
\hline\hline
 &\multicolumn{6}{c}{$\tau$} \\ 
 $\alpha = 0.1$ & $0$ & $1$ & $2$ & $3$ & $4$ & $5$ \\ 
 \hline$0.1$ & $ 98$ & $100$ & $104$ & $107$ & $106$ & $104$\\ 
$\gamma^2\;\;\;$ $0.2$ & $ 97$ & $ 98$ & $103$ & $107$ & $109$ & $107$\\ 
$0.3$ & $ 96$ & $ 97$ & $101$ & $105$ & $108$ & $109$\\ 
$0.4$ & $ 95$ & $ 96$ & $100$ & $103$ & $107$ & $109$\\ 
 \hline 
 \end{tabular}
 
 \vspace{2em} 
 
\begin{tabular}{r|rrrrrr}
\hline\hline
 &\multicolumn{6}{c}{$\tau$} \\ 
 $\alpha = 0.2$ & $0$ & $1$ & $2$ & $3$ & $4$ & $5$ \\ 
 \hline$0.1$ & $ 98$ & $100$ & $105$ & $107$ & $106$ & $103$\\ 
$\gamma^2\;\;\;$ $0.2$ & $ 94$ & $ 97$ & $104$ & $108$ & $109$ & $107$\\ 
$0.3$ & $ 93$ & $ 96$ & $101$ & $106$ & $109$ & $109$\\ 
$0.4$ & $ 89$ & $ 93$ & $ 97$ & $105$ & $108$ & $110$\\ 
 \hline 
 \end{tabular}
  \end{subtable}
  \caption{Choosing IVs Example: shortest 1-Step CI.} 
  \label{tab:Limit1StepShortChooseIVs}
\end{table}

\section{Simulation Results}
\label{sec:simulations}
\subsection{OLS versus TSLS Example}
\label{sec:OLSvsIVsim}
I begin by examining the performance of the FMSC and averaging estimator in the OLS versus TSLS example.
All calculations in this section are based on the formulas from Sections \ref{sec:OLSvsIVExample} and \ref{sec:momentavgexample} with 10,000 simulation replications. 
The data generating process is given by 
\begin{eqnarray}
	y_i &=& 0.5 x_i + \epsilon_i\\
	\label{eq:OLSvsIVDGP1}
	x_i &=& \pi(z_{1i} + z_{2i} + z_{3i}) + v_i
	\label{eq:OLSvsIVDGP2}
\end{eqnarray}
with $(\epsilon_i, v_i, z_{1i}, z_{2i}, z_{3i}) \sim \mbox{ iid } N(0, \mathcal{S})$
\begin{equation}
	\mathcal{S} = \left[ \begin{array}
		{cc} \mathcal{S}_1 & 0 \\ 0 & \mathcal{S}_2
	\end{array} \right], \quad 
	\mathcal{S}_1 = \left[ \begin{array}
		{cc} 1 & \rho \\ \rho & 1 - \pi^2 
	\end{array} \right], \quad \mathcal{S}_2 = I_3 / 3
	\label{eq:OLSvsIVDGP3}
\end{equation}
for $i= 1, \hdots, N$ where $N$, $\rho$ and $\pi$ vary over a grid.
The goal is to estimate the effect of $x$ on $y$, in this case 0.5, with minimum MSE either by choosing between OLS and TSLS estimators or by averaging them.
To ensure that the finite-sample MSE of the TSLS estimator exists, this DGP includes three instruments leading to two overidentifying restrictions \citep{Phillips1980}.\footnote{Alternatively, one could use fewer instruments in the DGP and work with trimmed MSE, as described in Online Appendix \ref{append:trim}.}
This design satisfies regularity conditions that are sufficient for Theorem \ref{thm:OLSvsIV} -- in particular it satisfies Assumption \ref{assump:OLSvsIV} from Online Appendix \ref{sec:sufficient_conditions} -- and keeps the variance of $x$ fixed at one so that $\pi = Cor(x_i, z_{1i} + z_{2i} + z_{3i})$ and $\rho = Cor(x_i,\epsilon_i)$.
The first-stage R-squared is simply $1 - \sigma_v^2/\sigma_x^2 = \pi^2$ so that larger values of $|\pi|$ \emph{reduce} the variance of the TSLS estimator.
Since $\rho$ controls the endogeneity of $x$, larger values of $|\rho|$ \emph{increase} the bias of the OLS estimator.

Figure \ref{fig:OLSvsIV_RMSEbaseline} compares the root mean-squared error (RMSE) of the post-FMSC estimator to those of the OLS and TSLS estimators.\footnote{Note that, while the first two moments of the TSLS estimator exist in this simulation design, none of its higher moments do.
This can be seen from the simulation results: even with 10,000 replications, the RMSE of the TSLS estimator shows a noticeable degree of simulation error.}
For any values of $N$ and $\pi$ there is a value of $\rho$ below which OLS outperforms TSLS: as $N$ and $\pi$ increase this value approaches zero; as they decrease it approaches one.
In practice, of course, $\rho$ in unknown so we cannot tell which of OLS and TSLS is to be preferred \emph{a priori}.
If we make it our policy to always use TSLS we will protect ourselves against bias at the potential cost of very high variance.
If, on the other hand, we make it our policy to always use OLS then we protect ourselves against high variance at the potential cost of severe bias. 
FMSC represents a compromise between these two extremes that does not require advance knowledge of $\rho$. 
When the RMSE of TSLS is high, the FMSC behaves more like OLS; when the RMSE of OLS is high it behaves more like TSLS.
Because the FMSC is itself a random variable, however, it sometimes makes moment selection mistakes.\footnote{For more discussion of this point, see Section \ref{sec:avg}.} 
For this reason it does not attain an RMSE equal to the lower envelope of the OLS and TSLS estimators.
The larger the RMSE difference between OLS and TSLS, however, the closer the FMSC comes to this lower envelope: costly mistakes are rare.

\begin{figure}
\centering
	\input{./SimulationOLSvsIV/Results/RMSE_coarse_pi_baseline.tex}
	\caption{RMSE values for the two-stage least squares (TSLS) estimator, the ordinary least squares (OLS) estimator, and the post-Focused Moment Selection Criterion (FMSC) estimator based on 10,000 simulation draws from the DGP given in Equations \ref{eq:OLSvsIVDGP1}--\ref{eq:OLSvsIVDGP3} using the formulas described in Section \ref{sec:OLSvsIVExample}.}
	\label{fig:OLSvsIV_RMSEbaseline}
\end{figure}

As shown above, the FMSC takes a very special form in this example: it is equivalent to a DHW test with $\alpha \approx 0.16$.
Accordingly, Figure \ref{fig:OLSvsIV_AVG} compares the RMSE of the post-FMSC estimator to those of DHW pre-test estimators with significance levels $\alpha = 0.05$ and $\alpha = 0.1$, indicated in the legend by DHW95 and DHW90.
Since these three procedures differ only in their critical values, they show similar qualitative behavior.
When $\rho$ is sufficiently close to zero, we saw from Figure \ref{fig:OLSvsIV_RMSEbaseline} that OLS has a lower RMSE than TSLS.
Since DHW95 and DHW90 require a higher burden of proof to reject OLS in favor of TSLS, they outperform FMSC in this region of the parameter space.
When $\rho$ crosses the threshold beyond which TSLS has a lower RMSE than OLS, the tables are turned: FMSC outperforms DHW95 and DHW90.
As $\rho$ increases further, relative to sample size and $\pi$, the three procedures become indistinguishable in terms of RMSE.
In addition to comparing the FMSC to DHW pre-test estimators, Figure \ref{fig:OLSvsIV_AVG} also presents the finite-sample RMSE of the minimum-AMSE moment average estimator presented in Equations \ref{eq:OLSvsIV_AVG1} and \ref{eq:OLSvsIV_AVG2}.
The performance of the moment average estimator is very strong: it provides the lowest worst-case RMSE and improves uniformly on the FMSC for all but the largest values of $\rho$.

\begin{figure}
\centering
	\input{./SimulationOLSvsIV/Results/RMSE_coarse_pi_relative_all.tex}
	\caption{RMSE values for the post-Focused Moment Selection Criterion (FMSC) estimator, Durbin-Hausman-Wu pre-test estimators with $\alpha = 0.1$ (DWH90) and $\alpha = 0.05$ (DHW95), and the minmum-AMSE averaging estimator, based on 10,000 simulation draws from the DGP given in Equations \ref{eq:OLSvsIVDGP1}--\ref{eq:OLSvsIVDGP3} using the formulas described in Sections \ref{sec:OLSvsIVExample} and \ref{sec:momentavgexample}.}
	\label{fig:OLSvsIV_AVG}
\end{figure}

Because this example involves a scalar target parameter, no selection or averaging scheme can provide a \emph{uniform} improvement over the minimax estimator, namely TSLS. 
But the cost of protection against the worst case is extremely poor performance when $\pi$ and $N$ are small.
When this is the case, there is a strong argument for preferring the FMSC or minimum-AMSE estimator: we can reap the benefits of OLS when $\rho$ is small without risking the extremely large biases that could result if $\rho$ is in fact large.

Further simulation results for $\pi \in \left\{ 0.01, 0.05, 0.1 \right\}$ appear in Online Appendix \ref{sec:appendWeak}.
For these parameter values the TSLS estimator suffers from a weak instrument problem leading the FMSC to substantially outperform the TSLS estimator.
See Online Appendix \ref{sec:appendWeak} for a more detailed discussion.

\subsection{Choosing Instrumental Variables Example}
\label{sec:chooseIVsim}
I now evaluate the performance of FMSC in the instrument selection example described in Section \ref{sec:chooseIVexample} using the following simulation design:
\begin{eqnarray}
		y_i &=& 0.5 x_i + \epsilon_i\\ 
		\label{eq:chooseIVDGP1}
		x_i &=& (z_{1i} + z_{2i} + z_{3i}) /3 + \gamma w_i + v_i 
		\label{eq:chooseIVDGP2}
	\end{eqnarray}
for $i=1, 2, \hdots, N$ where $(\epsilon_i, v_i, w_i, z_{i1}, z_{2i}, z_{3i})' \sim \mbox{ iid  } N(0,\mathcal{V})$ with

\begin{equation}	
	\mathcal{V} = \left[  \begin{array}
		{cc} \mathcal{V}_1 & 0 \\ 0 & \mathcal{V}_2
	\end{array}\right], \quad
	\mathcal{V}_1 = \left[ \begin{array}
		{ccc} 
		1 & (0.5 - \gamma \rho) & \rho \\
		(0.5 - \gamma \rho) & (8/9 - \gamma^2) & 0\\ 
		\rho & 0 & 1 \\ 
	\end{array} \right], \quad \mathcal{V}_2 = I_3 / 3
	\label{eq:chooseIVDGP3}
\end{equation}
This setup keeps the variance of $x$ fixed at one and the endogeneity of $x$, $Cor(x, \epsilon)$, fixed at $0.5$ while allowing the relevance, $\gamma = Cor(x,w)$, and endogeneity, $\rho = Cor(w, \epsilon)$, of the instrument $w$ to vary.
The instruments $z_1, z_2, z_3$ are valid and exogenous: they have first-stage coefficients of $1/3$ and are uncorrelated with the second stage error $\epsilon$.
The additional instrument $w$ is only relevant if $\gamma \neq 0$ and is only exogenous if $\rho = 0$.
Since $x$ has unit variance, the first-stage R-squared for this simulation design is simply $1 - \sigma_v^2 = 1/9 + \gamma^2$.
Hence, when  $\gamma = 0$, so that $w$ is irrelevant, the first-stage R-squared is just over 0.11.
Increasing $\gamma$ increases the R-squared of the first-stage.
This design satisfies the sufficient conditions for Theorem \ref{thm:chooseIV} given in Assumption \ref{assump:chooseIV} from Online Appendix \ref{sec:sufficient_conditions}.
When $\gamma = 0$, it is a special case of the DGP from Section \ref{sec:OLSvsIVsim}.

As in Section \ref{sec:OLSvsIVsim}, the goal of moment selection in this exercise is to estimate the effect of $x$ on $y$, as before 0.5, with minimum MSE.
In this case, however, the choice is between two TSLS estimators rather than OLS and TSLS: the \emph{valid} estimator uses only $z_1, z_2,$ and $z_3$ as instruments, while the \emph{full} estimator uses $z_1, z_2, z_3,$ and $w$.
The inclusion of $z_1, z_2$ and $z_3$ in both moment sets means that the order of over-identification is two for the valid estimator and three for the full estimator. 
Because the moments of the TSLS estimator only exist up to the order of over-identification \citep{Phillips1980}, this ensures that the small-sample MSE is well-defined.\footnote{Alternatively, one could use fewer instruments for the valid estimator and compare the results using \emph{trimmed} MSE. For details, see Online Appendix \ref{append:trim}.}
All estimators in this section are calculated via TSLS without a constant term using the expressions from Section \ref{sec:chooseIVexample} and 20,000 simulation replications.
 
Figure \ref{fig:chooseIVsim_RMSEbaseline} presents RMSE values for the valid estimator, the full estimator, and the post-FMSC estimator for various combinations of $\gamma$, $\rho$, and $N$.
The results are broadly similar to those from the OLS versus TSLS example presented in Figure \ref{fig:OLSvsIV_RMSEbaseline}.
For any combination $(\gamma,N)$ there is a positive value of $\rho$ below which the full estimator yields a lower RMSE than the full estimator.
As the sample size increases, this cutoff becomes smaller; as $\gamma$ increases, it becomes larger.
As in the OLS versus TSLS example, the post-FMSC estimator represents a compromise between the two estimators over which the FMSC selects.
Unlike in the previous example, however, when $N$ is sufficiently small there is a range of values for $\rho$ within which the FMSC yields a lower RMSE than \emph{both} the valid and full estimators.
This comes from the fact that the valid estimator is quite erratic for small sample sizes.
Such behavior is unsurprising given that its first stage is not especially strong, $\mbox{R-squared}\approx 11\%$, and it has only two moments.
In contrast, the full estimator has three moments and a stronger first stage.
As in the OLS versus TSLS example, the post-FMSC estimator does not uniformly outperform the valid estimator for all parameter values, although it does for smaller sample sizes.
The FMSC never performs much worse than the valid estimator, however, and often performs substantially better, particularly for small sample sizes.
\begin{figure}
\centering
	\input{./SimulationChooseIVs/Results/RMSE_coarse_gamma_baseline.tex}
	\caption{RMSE values for the valid estimator, including only $(z_1, z_2, z_3)$, the full estimator, including $(z_1, z_2, z_3, w)$, and the post-Focused Moment Selection Criterion (FMSC) estimator based on 20,000 simulation draws from the DGP given in Equations \ref{eq:chooseIVDGP1}--\ref{eq:chooseIVDGP3} using the formulas described in Section \ref{sec:chooseIVexample}.}
	\label{fig:chooseIVsim_RMSEbaseline}
\end{figure}

I now compare the FMSC to the GMM moment selection criteria of \cite{Andrews1999}, which take the form $MSC(S) = J_n(S) - h(|S|)\kappa_n$, where $J_n(S)$ is the $J$-test statistic under moment set $S$ and $-h(|S|)\kappa_n$ is a ``bonus term'' that rewards the inclusion of more moment conditions.
For each member of this family we choose the moment set that \emph{minimizes} $MSC(S)$. 
If we take $h(|S|) = (p + |S| - r)$, then $\kappa_n = \log{n}$ gives a GMM analogue of Schwarz's Bayesian Information Criterion (GMM-BIC) while $\kappa_n = 2.01 \log{\log{n}}$ gives an analogue of the Hannan-Quinn Information Criterion (GMM-HQ), and $\kappa_n = 2$ gives an analogue of Akaike's Information Criterion (GMM-AIC). 
Like the maximum likelihood model selection criteria upon which they are based, the GMM-BIC and GMM-HQ are consistent provided that Assumption \ref{assump:Andrews} holds, while the GMM-AIC, like the FMSC, is conservative.
Figure \ref{fig:chooseIVsim_RMSErelMSC} gives the RMSE values for the post-FMSC estimator alongside those of the post-GMM-BIC, HQ and AIC estimators.
I calculate the $J$-test statistic using a centered covariance matrix estimator, following the recommendation of \cite{Andrews1999}.
For small sample sizes, the GMM-BIC, AIC and HQ are quite erratic: indded for $N = 50$ the FMSC has a uniformly smaller RMSE.
This problem comes from the fact that the $J$-test statistic can be very badly behaved in small samples.\footnote{For more details, see Online Appendix \ref{sec:downwardJ}.}
As the sample size becomes larger, the classic tradeoff between consistent and conservative selection emerges.
For the smallest values of $\rho$ the consistent criteria outperform the conservative criteria; for moderate values the situation is reversed.
The consistent criteria, however, have the highest worst-case RMSE.
For a discussion of a combined strategy based on the GMM information criteria of \cite{Andrews1999} and the canonical correlations information criteria of \cite{HallPeixe2003}, see Online Appendix \ref{sec:CCIC}.
For a comparison with the downward $J$-test, see Online Appendix \ref{sec:downwardJ}.
Online Appendix \ref{sec:appendWeak} presents results for a modified simulation experiment in which the valid estimator suffers from a weak instrument problem. 
The FMSC performs very well in this case.

\begin{figure}
\centering
	\input{./SimulationChooseIVs/Results/RMSE_coarse_gamma_rel_MSC.tex}
	\caption{RMSE values for the post-Focused Moment Selection Criterion (FMSC) estimator and the GMM-BIC, HQ, and AIC estimators based on 20,000 simulation draws from the DGP given in Equations \ref{eq:chooseIVDGP1}--\ref{eq:chooseIVDGP3} using the formulas described in Section \ref{sec:chooseIVexample}.}
	\label{fig:chooseIVsim_RMSErelMSC}
\end{figure}

\subsection{Confidence Interval Simulations}
\label{sec:CIsim}
I now revisit the simulation experiments introduced above in Sections \ref{sec:OLSvsIVsim} and \ref{sec:chooseIVsim} to evaluate the finite-sample performance of the confidence intervals whose asymptotic performance was studied in Section \ref{sec:limitexperiment} above.
All results in this section are based on 1000 simulation replications from the relevant DGP.
Coverage probabilities and relative widths are all given in percentage points, rounded to the nearest whole percent.
In the interest of brevity I present only results for $N=100$.
Likewise, for the two-step confidence intervals I present results only for $\alpha_1 = \alpha/4, \alpha_2 = 3\alpha/4$.
Simulation results for $N=50$ and $500$ and other configurations of $\alpha_1,\alpha_2$ are available upon request.
Taking $N=100$ has the advantage of making the tables in this section directly comparable to those of Section \ref{sec:limitexperiment}.
Because I set $\sigma_x^2 = \sigma_\epsilon^2 = 1$ in both simulation experiments, this implies that $\sqrt{N}\rho = \sqrt{N} Cor(x_i,\epsilon_i) = \tau$ in the OLS versus TSLS example and $\sqrt{N}\rho = \sqrt{N} Cor(w_i, \epsilon_i) = \tau$ in the choosing IVs example. 
Thus when $N = 100$, taking $\rho \in \{0, 0.1, \dots, 0.5\}$ is the finite-sample analogue of $\tau \in \{0, 1, \dots, 5\}$.

To begin, Tables \ref{tab:CISim100Naive_OLSvsIV} and \ref{tab:CISim100Naive_ChooseIVs} present the coverage probability and average relative width of a na\"{i}ve confidence interval that ignores the effects of moment selection on inference.
These are the finite-sample analogues of Tables \ref{tab:LimitNaiveCover} and \ref{tab:LimitNaiveWidth}.
For the OLS versus IV example, expected relative width is calculated relative to a textbook confidence interval for the TSLS estimator while for the choosing IVs example it is calculated relative to a textbook confidence interval for the valid estimator that excludes $w$ from the instrument set. 
As in the asymptotic calculations presented above, we find that the na\"{i}ve procedure suffers from severe size distortions but results in much shorter intervals.
Results for the 2-Step confidence interval appear in Online Appendix \ref{append:conf_sim}.
With a small allowance for sampling variability, we see that the 2-step intervals indeed provide uniform coverage no lower than their nominal level but result in wider intervals than simply using TSLS or the valid estimator, respectively.
Tables \ref{tab:CISim100_1stepShort_OLSvsIV} and \ref{tab:CISim100_1stepShort_OLSvsIV}, the finite-sample analogues of Tables \ref{tab:Limit1StepShortOLSvsIV} and \ref{tab:Limit1StepShortChooseIVs} present results for the one-step confidence interval that assumes $\widehat{\tau} = \tau$.
As expected from the asymptotic calculations from Section \ref{sec:limitexperiment}, this interval presents a good trade-off between the na\"{i}ve and 2-step CIs: it can yield shorter intervals with far smaller size distortions.
Because it is also simple to compute, the 1-Step interval could prove quite valuable in practice.
It would be interesting to explore this interval further both theoretically and in simulation studies.

\begin{table}[h]
  \footnotesize
  \centering
  \begin{subtable}{0.48\textwidth}
    \caption{Coverage Probability}
    \begin{tabular}{r|rrrrrr}
\hline\hline
 &\multicolumn{6}{c}{$\rho$} \\ 
 $\alpha = 0.05$ & $0$ & $0.1$ & $0.2$ & $0.3$ & $0.4$ & $0.5$ \\ 
 \hline$0.1$ & $92$ & $83$ & $56$ & $36$ & $38$ & $55$\\ 
$\pi^2\;\;\;$ $0.2$ & $91$ & $84$ & $63$ & $56$ & $73$ & $85$\\ 
$0.3$ & $92$ & $85$ & $70$ & $73$ & $88$ & $93$\\ 
$0.4$ & $93$ & $85$ & $76$ & $83$ & $93$ & $94$\\ 
 \hline 
 \end{tabular}
 
 \vspace{2em} 
 
\begin{tabular}{r|rrrrrr}
\hline\hline
 &\multicolumn{6}{c}{$\rho$} \\ 
 $\alpha = 0.1$ & $0$ & $0.1$ & $0.2$ & $0.3$ & $0.4$ & $0.5$ \\ 
 \hline$0.1$ & $87$ & $73$ & $42$ & $33$ & $41$ & $54$\\ 
$\pi^2\;\;\;$ $0.2$ & $85$ & $76$ & $50$ & $53$ & $73$ & $83$\\ 
$0.3$ & $86$ & $75$ & $63$ & $71$ & $86$ & $88$\\ 
$0.4$ & $86$ & $74$ & $68$ & $83$ & $89$ & $89$\\ 
 \hline 
 \end{tabular}
 
 \vspace{2em} 
 
\begin{tabular}{r|rrrrrr}
\hline\hline
 &\multicolumn{6}{c}{$\rho$} \\ 
 $\alpha = 0.2$ & $0$ & $0.1$ & $0.2$ & $0.3$ & $0.4$ & $0.5$ \\ 
 \hline$0.1$ & $72$ & $57$ & $29$ & $26$ & $38$ & $53$\\ 
$\pi^2\;\;\;$ $0.2$ & $74$ & $59$ & $40$ & $45$ & $65$ & $77$\\ 
$0.3$ & $74$ & $59$ & $46$ & $62$ & $78$ & $79$\\ 
$0.4$ & $75$ & $61$ & $59$ & $74$ & $81$ & $78$\\ 
 \hline 
 \end{tabular}
  \end{subtable}
  ~
  \begin{subtable}{0.48\textwidth}
    \caption{Average Relative Width}
    \begin{tabular}{r|rrrrrr}
\hline\hline
 &\multicolumn{6}{c}{$\rho$} \\ 
 $\alpha = 0.05$ & $0$ & $0.1$ & $0.2$ & $0.3$ & $0.4$ & $0.5$ \\ 
 \hline$0.1$ & $ 39$ & $ 37$ & $ 31$ & $ 50$ & $ 56$ & $ 69$\\ 
$\pi^2\;\;\;$ $0.2$ & $ 52$ & $ 54$ & $ 64$ & $ 74$ & $ 86$ & $ 95$\\ 
$0.3$ & $ 60$ & $ 65$ & $ 75$ & $ 87$ & $ 96$ & $100$\\ 
$0.4$ & $ 68$ & $ 73$ & $ 84$ & $ 95$ & $100$ & $100$\\ 
 \hline 
 \end{tabular}
 
 \vspace{2em} 
 
\begin{tabular}{r|rrrrrr}
\hline\hline
 &\multicolumn{6}{c}{$\rho$} \\ 
 $\alpha = 0.1$ & $0$ & $0.1$ & $0.2$ & $0.3$ & $0.4$ & $0.5$ \\ 
 \hline$0.1$ & $ 38$ & $ 39$ & $ 43$ & $ 49$ & $ 54$ & $ 69$\\ 
$\pi^2\;\;\;$ $0.2$ & $ 52$ & $ 55$ & $ 61$ & $ 74$ & $ 87$ & $ 95$\\ 
$0.3$ & $ 61$ & $ 65$ & $ 76$ & $ 88$ & $ 97$ & $100$\\ 
$0.4$ & $ 69$ & $ 74$ & $ 85$ & $ 95$ & $100$ & $100$\\ 
 \hline 
 \end{tabular}
 
 \vspace{2em} 
 
\begin{tabular}{r|rrrrrr}
\hline\hline
 &\multicolumn{6}{c}{$\rho$} \\ 
 $\alpha = 0.2$ & $0$ & $0.1$ & $0.2$ & $0.3$ & $0.4$ & $0.5$ \\ 
 \hline$0.1$ & $ 40$ & $ 41$ & $ 40$ & $ 49$ & $ 57$ & $ 68$\\ 
$\pi^2\;\;\;$ $0.2$ & $ 52$ & $ 55$ & $ 63$ & $ 73$ & $ 86$ & $ 95$\\ 
$0.3$ & $ 61$ & $ 65$ & $ 74$ & $ 88$ & $ 96$ & $ 99$\\ 
$0.4$ & $ 69$ & $ 73$ & $ 85$ & $ 95$ & $100$ & $100$\\ 
 \hline 
 \end{tabular}
  \end{subtable}
  \caption{Na\"{i}ve CI, OLS vs IV Example, $N=100$}
  \label{tab:CISim100Naive_OLSvsIV}
\end{table}

\begin{table}[h]
  \footnotesize
  \centering
  \begin{subtable}{0.48\textwidth}
    \caption{Coverage Probability}
    \begin{tabular}{r|rrrrrr}
\hline\hline
 &\multicolumn{6}{c}{$\rho$} \\ 
 $\alpha = 0.05$ & $0$ & $0.1$ & $0.2$ & $0.3$ & $0.4$ & $0.5$ \\ 
 \hline$0.1$ & $90$ & $80$ & $77$ & $83$ & $88$ & $93$\\ 
$\gamma^2\;\;\;$ $0.2$ & $89$ & $78$ & $66$ & $69$ & $82$ & $92$\\ 
$0.3$ & $87$ & $78$ & $59$ & $60$ & $77$ & $86$\\ 
$0.4$ & $89$ & $77$ & $55$ & $50$ & $67$ & $81$\\ 
 \hline 
 \end{tabular}
 
 \vspace{2em} 
 
\begin{tabular}{r|rrrrrr}
\hline\hline
 &\multicolumn{6}{c}{$\rho$} \\ 
 $\alpha = 0.1$ & $0$ & $0.1$ & $0.2$ & $0.3$ & $0.4$ & $0.5$ \\ 
 \hline$0.1$ & $84$ & $72$ & $73$ & $77$ & $86$ & $89$\\ 
$\gamma^2\;\;\;$ $0.2$ & $83$ & $72$ & $58$ & $66$ & $78$ & $86$\\ 
$0.3$ & $84$ & $70$ & $51$ & $55$ & $73$ & $82$\\ 
$0.4$ & $81$ & $65$ & $46$ & $48$ & $63$ & $78$\\ 
 \hline 
 \end{tabular}
 
 \vspace{2em} 
 
\begin{tabular}{r|rrrrrr}
\hline\hline
 &\multicolumn{6}{c}{$\rho$} \\ 
 $\alpha = 0.2$ & $0$ & $0.1$ & $0.2$ & $0.3$ & $0.4$ & $0.5$ \\ 
 \hline$0.1$ & $76$ & $63$ & $60$ & $70$ & $77$ & $79$\\ 
$\gamma^2\;\;\;$ $0.2$ & $75$ & $59$ & $49$ & $58$ & $73$ & $77$\\ 
$0.3$ & $69$ & $54$ & $39$ & $48$ & $67$ & $76$\\ 
$0.4$ & $71$ & $53$ & $34$ & $40$ & $57$ & $74$\\ 
 \hline 
 \end{tabular}
  \end{subtable}
  ~
  \begin{subtable}{0.48\textwidth}
    \caption{Average Relative Width}
    \begin{tabular}{r|rrrrrr}
\hline\hline
 &\multicolumn{6}{c}{$\rho$} \\ 
 $\alpha = 0.05$ & $0$ & $0.1$ & $0.2$ & $0.3$ & $0.4$ & $0.5$ \\ 
 \hline$0.1$ & $72$ & $72$ & $80$ & $91$ & $97$ & $99$\\ 
$\gamma^2\;\;\;$ $0.2$ & $62$ & $61$ & $72$ & $83$ & $91$ & $98$\\ 
$0.3$ & $56$ & $54$ & $58$ & $71$ & $87$ & $94$\\ 
$0.4$ & $49$ & $47$ & $53$ & $67$ & $80$ & $87$\\ 
 \hline 
 \end{tabular}
 
 \vspace{2em} 
 
\begin{tabular}{r|rrrrrr}
\hline\hline
 &\multicolumn{6}{c}{$\rho$} \\ 
 $\alpha = 0.1$ & $0$ & $0.1$ & $0.2$ & $0.3$ & $0.4$ & $0.5$ \\ 
 \hline$0.1$ & $71$ & $67$ & $81$ & $91$ & $99$ & $98$\\ 
$\gamma^2\;\;\;$ $0.2$ & $60$ & $61$ & $67$ & $78$ & $89$ & $92$\\ 
$0.3$ & $55$ & $54$ & $59$ & $72$ & $82$ & $94$\\ 
$0.4$ & $51$ & $47$ & $55$ & $68$ & $78$ & $89$\\ 
 \hline 
 \end{tabular}
 
 \vspace{2em} 
 
\begin{tabular}{r|rrrrrr}
\hline\hline
 &\multicolumn{6}{c}{$\rho$} \\ 
 $\alpha = 0.2$ & $0$ & $0.1$ & $0.2$ & $0.3$ & $0.4$ & $0.5$ \\ 
 \hline$0.1$ & $71$ & $71$ & $79$ & $92$ & $96$ & $98$\\ 
$\gamma^2\;\;\;$ $0.2$ & $60$ & $60$ & $68$ & $81$ & $91$ & $97$\\ 
$0.3$ & $55$ & $56$ & $61$ & $72$ & $84$ & $95$\\ 
$0.4$ & $51$ & $48$ & $53$ & $64$ & $78$ & $89$\\ 
 \hline 
 \end{tabular}
  \end{subtable}
  \caption{Na\"{i}ve CI, Choosing IVs Example, $N=100$}
  \label{tab:CISim100Naive_ChooseIVs}
\end{table}

\begin{table}[h]
  \footnotesize
  \centering
  \begin{subtable}{0.48\textwidth}
    \caption{Coverage Probability}
    \begin{tabular}{r|rrrrrr}
\hline\hline
 &\multicolumn{6}{c}{$\rho$} \\ 
 $\alpha = 0.05$ & $0$ & $0.1$ & $0.2$ & $0.3$ & $0.4$ & $0.5$ \\ 
 \hline$0.1$ & $95$ & $97$ & $97$ & $93$ & $80$ & $72$\\ 
$\pi^2\;\;\;$ $0.2$ & $95$ & $98$ & $94$ & $88$ & $84$ & $89$\\ 
$0.3$ & $96$ & $96$ & $95$ & $88$ & $91$ & $96$\\ 
$0.4$ & $96$ & $96$ & $92$ & $90$ & $96$ & $95$\\ 
 \hline 
 \end{tabular}
 
 \vspace{2em} 
 
\begin{tabular}{r|rrrrrr}
\hline\hline
 &\multicolumn{6}{c}{$\rho$} \\ 
 $\alpha = 0.1$ & $0$ & $0.1$ & $0.2$ & $0.3$ & $0.4$ & $0.5$ \\ 
 \hline$0.1$ & $94$ & $93$ & $91$ & $82$ & $73$ & $68$\\ 
$\pi^2\;\;\;$ $0.2$ & $91$ & $92$ & $87$ & $79$ & $81$ & $88$\\ 
$0.3$ & $93$ & $91$ & $86$ & $82$ & $90$ & $93$\\ 
$0.4$ & $91$ & $90$ & $85$ & $89$ & $92$ & $90$\\ 
 \hline 
 \end{tabular}
 
 \vspace{2em} 
 
\begin{tabular}{r|rrrrrr}
\hline\hline
 &\multicolumn{6}{c}{$\rho$} \\ 
 $\alpha = 0.2$ & $0$ & $0.1$ & $0.2$ & $0.3$ & $0.4$ & $0.5$ \\ 
 \hline$0.1$ & $83$ & $79$ & $72$ & $64$ & $61$ & $63$\\ 
$\pi^2\;\;\;$ $0.2$ & $86$ & $81$ & $70$ & $67$ & $75$ & $83$\\ 
$0.3$ & $84$ & $80$ & $73$ & $75$ & $84$ & $82$\\ 
$0.4$ & $83$ & $78$ & $75$ & $81$ & $84$ & $78$\\ 
 \hline 
 \end{tabular}
  \end{subtable}
  ~
  \begin{subtable}{0.48\textwidth}
    \caption{Average Relative Width}
    \begin{tabular}{r|rrrrrr}
\hline\hline
 &\multicolumn{6}{c}{$\rho$} \\ 
 $\alpha = 0.05$ & $0$ & $0.1$ & $0.2$ & $0.3$ & $0.4$ & $0.5$ \\ 
 \hline$0.1$ & $ 92$ & $ 91$ & $ 94$ & $ 94$ & $ 95$ & $ 99$\\ 
$\pi^2\;\;\;$ $0.2$ & $ 96$ & $ 96$ & $ 99$ & $102$ & $107$ & $112$\\ 
$0.3$ & $ 97$ & $ 98$ & $103$ & $107$ & $111$ & $110$\\ 
$0.4$ & $ 98$ & $100$ & $105$ & $109$ & $108$ & $102$\\ 
 \hline 
 \end{tabular}
 
 \vspace{2em} 
 
\begin{tabular}{r|rrrrrr}
\hline\hline
 &\multicolumn{6}{c}{$\rho$} \\ 
 $\alpha = 0.1$ & $0$ & $0.1$ & $0.2$ & $0.3$ & $0.4$ & $0.5$ \\ 
 \hline$0.1$ & $ 88$ & $ 89$ & $ 91$ & $ 93$ & $ 95$ & $102$\\ 
$\pi^2\;\;\;$ $0.2$ & $ 94$ & $ 96$ & $ 99$ & $105$ & $110$ & $114$\\ 
$0.3$ & $ 95$ & $ 98$ & $104$ & $110$ & $111$ & $107$\\ 
$0.4$ & $ 97$ & $101$ & $106$ & $110$ & $106$ & $101$\\ 
 \hline 
 \end{tabular}
 
 \vspace{2em} 
 
\begin{tabular}{r|rrrrrr}
\hline\hline
 &\multicolumn{6}{c}{$\rho$} \\ 
 $\alpha = 0.2$ & $0$ & $0.1$ & $0.2$ & $0.3$ & $0.4$ & $0.5$ \\ 
 \hline$0.1$ & $ 81$ & $ 83$ & $ 81$ & $ 91$ & $ 97$ & $105$\\ 
$\pi^2\;\;\;$ $0.2$ & $ 89$ & $ 91$ & $ 98$ & $106$ & $113$ & $114$\\ 
$0.3$ & $ 93$ & $ 96$ & $104$ & $111$ & $111$ & $105$\\ 
$0.4$ & $ 96$ & $ 99$ & $108$ & $110$ & $105$ & $101$\\ 
 \hline 
 \end{tabular}
  \end{subtable}
  \caption{1-Step Shortest CI, OLS vs IV Example, $N=100$}
  \label{tab:CISim100_1stepShort_OLSvsIV}
\end{table}

\begin{table}[h]
  \footnotesize
  \centering
  \begin{subtable}{0.48\textwidth}
    \caption{Coverage Probability}
    \begin{tabular}{r|rrrrrr}
\hline\hline
 &\multicolumn{6}{c}{$\rho$} \\ 
 $\alpha = 0.05$ & $0$ & $0.1$ & $0.2$ & $0.3$ & $0.4$ & $0.5$ \\ 
 \hline$0.1$ & $93$ & $90$ & $89$ & $89$ & $92$ & $94$\\ 
$\gamma^2\;\;\;$ $0.2$ & $93$ & $92$ & $89$ & $85$ & $90$ & $95$\\ 
$0.3$ & $92$ & $93$ & $88$ & $86$ & $88$ & $92$\\ 
$0.4$ & $94$ & $94$ & $90$ & $83$ & $83$ & $87$\\ 
 \hline 
 \end{tabular}
 
 \vspace{2em} 
 
\begin{tabular}{r|rrrrrr}
\hline\hline
 &\multicolumn{6}{c}{$\rho$} \\ 
 $\alpha = 0.1$ & $0$ & $0.1$ & $0.2$ & $0.3$ & $0.4$ & $0.5$ \\ 
 \hline$0.1$ & $89$ & $84$ & $84$ & $85$ & $90$ & $91$\\ 
$\gamma^2\;\;\;$ $0.2$ & $90$ & $89$ & $82$ & $81$ & $86$ & $90$\\ 
$0.3$ & $90$ & $89$ & $83$ & $78$ & $82$ & $88$\\ 
$0.4$ & $87$ & $89$ & $82$ & $78$ & $77$ & $85$\\ 
 \hline 
 \end{tabular}
 
 \vspace{2em} 
 
\begin{tabular}{r|rrrrrr}
\hline\hline
 &\multicolumn{6}{c}{$\rho$} \\ 
 $\alpha = 0.2$ & $0$ & $0.1$ & $0.2$ & $0.3$ & $0.4$ & $0.5$ \\ 
 \hline$0.1$ & $81$ & $76$ & $73$ & $76$ & $80$ & $80$\\ 
$\gamma^2\;\;\;$ $0.2$ & $83$ & $77$ & $74$ & $75$ & $80$ & $80$\\ 
$0.3$ & $80$ & $74$ & $69$ & $69$ & $78$ & $82$\\ 
$0.4$ & $81$ & $76$ & $68$ & $67$ & $70$ & $80$\\ 
 \hline 
 \end{tabular}
  \end{subtable}
  ~
  \begin{subtable}{0.48\textwidth}
    \caption{Average Relative Width}
    \begin{tabular}{r|rrrrrr}
\hline\hline
 &\multicolumn{6}{c}{$\rho$} \\ 
 $\alpha = 0.05$ & $0$ & $0.1$ & $0.2$ & $0.3$ & $0.4$ & $0.5$ \\ 
 \hline$0.1$ & $ 97$ & $ 99$ & $102$ & $106$ & $107$ & $106$\\ 
$\gamma^2\;\;\;$ $0.2$ & $ 97$ & $ 97$ & $101$ & $104$ & $108$ & $109$\\ 
$0.3$ & $ 96$ & $ 95$ & $ 97$ & $101$ & $106$ & $110$\\ 
$0.4$ & $ 94$ & $ 94$ & $ 96$ & $ 99$ & $104$ & $108$\\ 
 \hline 
 \end{tabular}
 
 \vspace{2em} 
 
\begin{tabular}{r|rrrrrr}
\hline\hline
 &\multicolumn{6}{c}{$\rho$} \\ 
 $\alpha = 0.1$ & $0$ & $0.1$ & $0.2$ & $0.3$ & $0.4$ & $0.5$ \\ 
 \hline$0.1$ & $ 97$ & $ 96$ & $104$ & $108$ & $108$ & $105$\\ 
$\gamma^2\;\;\;$ $0.2$ & $ 95$ & $ 96$ & $100$ & $105$ & $109$ & $107$\\ 
$0.3$ & $ 94$ & $ 94$ & $ 98$ & $103$ & $107$ & $111$\\ 
$0.4$ & $ 93$ & $ 92$ & $ 96$ & $102$ & $106$ & $110$\\ 
 \hline 
 \end{tabular}
 
 \vspace{2em} 
 
\begin{tabular}{r|rrrrrr}
\hline\hline
 &\multicolumn{6}{c}{$\rho$} \\ 
 $\alpha = 0.2$ & $0$ & $0.1$ & $0.2$ & $0.3$ & $0.4$ & $0.5$ \\ 
 \hline$0.1$ & $ 95$ & $ 98$ & $104$ & $110$ & $109$ & $104$\\ 
$\gamma^2\;\;\;$ $0.2$ & $ 90$ & $ 95$ & $102$ & $110$ & $113$ & $109$\\ 
$0.3$ & $ 90$ & $ 93$ & $ 99$ & $105$ & $111$ & $113$\\ 
$0.4$ & $ 87$ & $ 89$ & $ 95$ & $103$ & $109$ & $113$\\ 
 \hline 
 \end{tabular}
  \end{subtable}
  \caption{1-Step Shortest CI, Choosing IVs Example, $N=100$}
  \label{tab:CISim100_1StepShort_ChooseIVs}
\end{table}

\section{Empirical Example: Geography or Institutions?}
\label{sec:application}
\cite{Carstensen2006} address a controversial question from the development literature: what is the causal effect of geography on income per capita after controlling for the quality of institutions?
A number of well-known studies find little or no direct effect of geographic endowments \citep{Acemoglu,Rodrik,Easterly}. \cite{Sachs}, on the other hand, shows that malaria transmission, a variable largely driven by ecological conditions, directly influences the level of per capita income, even after controlling for institutions.
Because malaria transmission is very likely endogenous, Sachs uses a measure of ``malaria ecology,'' constructed to be exogenous both to present economic conditions and public health interventions, as an instrument. 
\cite{Carstensen2006} extend Sachs's work using the following baseline regression for a sample of 44 countries:
\begin{equation}
	\mbox{ln\emph{gdpc}}_i = \beta_1 + \beta_2 \cdot \mbox{\emph{institutions}}_i + \beta_3 \cdot \mbox{\emph{malaria}}_i + \epsilon_i
\end{equation}
This model augments the baseline specification of \cite{Acemoglu} to include a direct effect of malaria transmission which, like institutions, is treated as endogenous.\footnote{Due to a lack of data for certain instruments, \cite{Carstensen2006} work with a smaller sample of countries than \cite{Acemoglu}.} 
Considering a variety of measures of both institutions and malaria transmission, and a number of instrument sets, \cite{Carstensen2006} find large negative effects of malaria transmission, lending support to Sach's conclusion.

In this section, I revisit and expand upon the instrument selection exercise given in Table 2 of \cite{Carstensen2006} using the FMSC and corrected confidence intervals described above. 
All results in this section are calculated by TSLS using the formulas from Section \ref{sec:chooseIVexample} and the variables described in Table \ref{tab:desc}, with ln\emph{gdpc} as the outcome  variable and \emph{rule} and \emph{malfal} as measures of institutions and malaria transmission.
In this exercise I imagine two hypothetical econometricians.
The first, like \cite{Sachs} and \cite{Carstensen2006}, seeks the best possible estimate of the causal effect of malaria transmission, $\beta_3$, after controlling for institutions by selecting over a number of possible instruments.
The second, in contrast, seeks the best possible estimate of the causal effect of \emph{institutions}, $\beta_2$, after controlling for malaria transmission by selecting over the same collection of instruments.
After estimating their desired target parameters, both econometricians also wish to report valid confidence intervals that account for the additional uncertainty introduced by instrument selection.
For the purposes of this example, to illustrate the results relevant to each of my hypothetical researchers, I take each of $\beta_2$ and $\beta_3$ \emph{in turn} as the target parameter.\footnote{A researcher interested in \emph{both} $\beta_2$ and $\beta_3$, however, should not proceed in this fashion, as it could lead to contradictory inferences. 
Instead, she should define a target parameter that includes both $\beta_2$ and $\beta_3$.}

\begin{table}[!tbp]
\small
\centering
\begin{tabular}{lll}
\hline \hline
Name& Description &\\
\hline
ln\emph{gdpc}&Real GDP/capita at PPP, 1995 International Dollars &Outcome\\
\emph{rule}&Institutional quality (Average Governance Indicator)&Regressor\\
\emph{malfal}&Fraction of population at risk of malaria transmission, 1994&Regressor\\
ln\emph{mort}&Log settler mortality (per 1000 settlers), early 19th century&Baseline\\
\emph{maleco}&Index of stability of malaria transmission&Baseline\\
\emph{frost}&Prop.\ of land receiving at least 5 days of frost in winter&Climate\\
\emph{humid}&Highest temp. in month with highest avg.\ afternoon humidity&Climate\\
\emph{latitude}&Distance from equator (absolute value of latitude in degrees)&Climate \\
\emph{eurfrac}&Fraction of pop.\ that speaks major West.\ European Language&Europe \\
\emph{engfrac}&Fraction of pop.\ that speaks English&Europe\\
\emph{coast}&Proportion of land area within 100km of sea coast&Openness\\
\emph{trade}&Log Frankel-Romer predicted trade share&Openness\\
\hline
\end{tabular}
\caption{Description of variables for Empirical Example.}
\label{tab:desc}
\end{table}

To apply the FMSC to the present example, we require a minimum of two valid instruments besides the constant term. 
Based on the arguments given by \cite{Acemoglu} and \cite{Sachs}, I proceed under the assumption that ln\emph{mort} and \emph{maleco}, measures of early settler mortality and malaria ecology, are exogenous.
Rather than selecting over all 128 possible instrument sets, I consider eight specifications formed from the four instrument blocks defined by \cite{Carstensen2006}.
The baseline block contains ln\emph{mort}, \emph{maleco} and a constant; the climate block contains \emph{frost}, \emph{humid}, and \emph{latitude}; the Europe block contains \emph{eurfrac} and \emph{engfrac}; and the openness block contains \emph{coast} and \emph{trade}. 
Full descriptions of these variables appear in Table \ref{tab:desc}.
Table \ref{tab:fullresults} lists the eight instrument sets considered here, along with TSLS estimates and traditional 95\% confidence intervals for each.\footnote{The results for the baseline instrument presented in panel 1 of Table \ref{tab:fullresults} are slightly different from those in \cite{Carstensen2006} as I exclude Vietnam to keep the sample fixed across instrument sets.}

\begin{table}[h]
  \small
\centering
\begin{tabular}{lrrrrrrrr}
\hline \hline 
& \multicolumn{2}{c}{1} & \multicolumn{2}{c}{2} & \multicolumn{2}{c}{3} & \multicolumn{2}{c}{4}\\ 
& \multicolumn{1}{c}{\emph{rule}} & \multicolumn{1}{c}{\emph{malfal}} & \multicolumn{1}{c}{\emph{rule}} & \multicolumn{1}{c}{\emph{malfal}} & \multicolumn{1}{c}{\emph{rule}} & \multicolumn{1}{c}{\emph{malfal}} & \multicolumn{1}{c}{\emph{rule}} & \multicolumn{1}{c}{\emph{malfal}}\\ 
 \hline 
 
coeff. & $0.89$ & $-1.04$ & $0.97$ & $-0.90$ & $0.81$ & $-1.09$ & $0.86$ & $-1.14$\\ 
SE & $0.18$ & $0.31$ & $0.16$ & $0.29$ & $0.16$ & $0.29$ & $0.16$ & $0.27$\\ 
lower & $0.53$ & $-1.66$ & $0.65$ & $-1.48$ & $0.49$ & $-1.67$ & $0.55$ & $-1.69$\\ 
upper & $1.25$ & $-0.42$ & $1.30$ & $-0.32$ & $1.13$ & $-0.51$ & $1.18$ & $-0.59$\\ 
& \multicolumn{2}{c}{Baseline} & \multicolumn{2}{c}{Baseline} & \multicolumn{2}{c}{Baseline} & \multicolumn{2}{c}{Baseline}\\ 
& \multicolumn{2}{c}{} & \multicolumn{2}{c}{Climate} & \multicolumn{2}{c}{} & \multicolumn{2}{c}{}\\ 
& \multicolumn{2}{c}{} & \multicolumn{2}{c}{} & \multicolumn{2}{c}{Openness} & \multicolumn{2}{c}{}\\ 
& \multicolumn{2}{c}{} & \multicolumn{2}{c}{} & \multicolumn{2}{c}{} & \multicolumn{2}{c}{Europe}\\ 
 \hline
\end{tabular} 
 
 \vspace{2em} 
 
 \begin{tabular}{lrrrrrrrr}
\hline \hline 
& \multicolumn{2}{c}{5} & \multicolumn{2}{c}{6} & \multicolumn{2}{c}{7} & \multicolumn{2}{c}{8}\\ 
& \multicolumn{1}{c}{\emph{rule}} & \multicolumn{1}{c}{\emph{malfal}} & \multicolumn{1}{c}{\emph{rule}} & \multicolumn{1}{c}{\emph{malfal}} & \multicolumn{1}{c}{\emph{rule}} & \multicolumn{1}{c}{\emph{malfal}} & \multicolumn{1}{c}{\emph{rule}} & \multicolumn{1}{c}{\emph{malfal}}\\ 
 \hline 
 
coeff. & $0.93$ & $-1.02$ & $0.86$ & $-0.98$ & $0.81$ & $-1.16$ & $0.84$ & $-1.08$\\ 
SE & $0.15$ & $0.26$ & $0.14$ & $0.27$ & $0.15$ & $0.27$ & $0.13$ & $0.25$\\ 
lower & $0.63$ & $-1.54$ & $0.59$ & $-1.53$ & $0.51$ & $-1.70$ & $0.57$ & $-1.58$\\ 
upper & $1.22$ & $-0.49$ & $1.14$ & $-0.43$ & $1.11$ & $-0.62$ & $1.10$ & $-0.58$\\ 
& \multicolumn{2}{c}{Baseline} & \multicolumn{2}{c}{Baseline} & \multicolumn{2}{c}{Baseline} & \multicolumn{2}{c}{Baseline}\\ 
& \multicolumn{2}{c}{Climate} & \multicolumn{2}{c}{Climate} & \multicolumn{2}{c}{} & \multicolumn{2}{c}{Climate}\\ 
& \multicolumn{2}{c}{} & \multicolumn{2}{c}{Openness} & \multicolumn{2}{c}{Openness} & \multicolumn{2}{c}{Openness}\\ 
& \multicolumn{2}{c}{Europe} & \multicolumn{2}{c}{} & \multicolumn{2}{c}{Europe} & \multicolumn{2}{c}{Europe}\\ 
 \hline
\end{tabular}
\caption{Two-stage least squares estimation results for all instrument sets.}
\label{tab:fullresults}

\end{table}

Table \ref{tab:FMSC_values} presents FMSC and ``positive-part'' FMSC results for instrument sets 1--8.
The positive-part FMSC sets a negative squared bias estimate to zero when estimating AMSE.
If the squared bias estimate is positive, FMSC and positive-part FMSC coincide; if the squared bias estimate is negative, positive-part FMSC is strictly greater than FMSC. 
Additional simulation results for the choosing instrumental variables experiment from Section \ref{sec:chooseIVsim}, available upon request, reveal that the positive-part FMSC never performs worse than the ordinary FMSC and sometimes performs slightly better, suggesting that it may be preferable in real-world applications.
For each criterion the table presents two cases: the first takes the effect of \emph{malfal}, a measure of malaria transmission, as the target parameter while the second uses the effect of \emph{rule}, a measure of institutions. 
In each case the two best instrument sets are numbers 5 (baseline, climate and Europe) and 8 (all instruments).
When the target parameter is the coefficient on \emph{malfal}, 8 is the clear winner under both the plain-vanilla and positive-part FMSC, leading to an estimate of $-1.08$ for the effect of malaria transmission on per-capita income.
When the target parameter is the coefficient on \emph{rule}, however, instrument sets 5 and 8 are virtually indistinguishable.
Indeed, while the plain-vanilla FMSC selects instrument set 8, leading to an estimate of $0.84$ for the effect of instutitions on per-capita income, the positive-part FMSC selects instrument set 5, leading to an estimate of $0.93$. 
Thus the FMSC methodology shows that, while helpful for estimating the effect of malaria transmission, the openness instruments \emph{coast} and \emph{trade} provide essentially no additional information for studying the effect of institutions.

\begin{table}[htbp]
  \small
	\centering
	\begin{tabular}{lcccccc}
\hline\hline
 & \multicolumn{3}{c}{$\mu = malfal$}& \multicolumn{3}{c}{$\mu = rule$}\\ 
 & FMSC & posFMSC & $\widehat{\mu}$ & FMSC & posFMSC & $\widehat{\mu}$ \\ 
 \hline
 (1) Valid & $ 3.03$ & $ 3.03$ & $-1.04$ & $1.27$ & $1.27$ & $0.89$\\ 
(2) Climate & $ 3.07$ & $ 3.07$ & $-0.90$ & $1.00$ & $1.00$ & $0.97$\\ 
(3) Openness & $ 2.30$ & $ 2.42$ & $-1.09$ & $1.21$ & $1.21$ & $0.81$\\ 
(4) Europe & $ 1.82$ & $ 2.15$ & $-1.14$ & $0.52$ & $0.73$ & $0.86$\\ 
(5) Climate, Europe & $ 0.85$ & $ 2.03$ & $-1.02$ & $0.25$ & $0.59$ & $0.93$\\ 
(6) Climate, Openness & $ 1.85$ & $ 2.30$ & $-0.98$ & $0.45$ & $0.84$ & $0.86$\\ 
(7) Openness, Europe & $ 1.63$ & $ 1.80$ & $-1.16$ & $0.75$ & $0.75$ & $0.81$\\ 
(8) Full & $ 0.53$ & $ 1.69$ & $-1.08$ & $0.23$ & $0.62$ & $0.84$ \\ 
 \hline
 \end{tabular}
		\caption{FMSC and and positive-part FMSC values corresponding to the instrument sets from Table \ref{tab:fullresults}}
		\label{tab:FMSC_values}
\end{table}

Table \ref{tab:postFMSC_CIs} presents three alternative post-selection confidence intervals for each of the instrument selection exercises from Table \ref{tab:FMSC_values}: Na\"{i}ve, 1-Step, and 2-Step.
These intervals are constructed as described in the simulation experiments from Section \ref{sec:CIsim} above. 
The simulation-based intervals are based on 10,000 random draws.
For the two-step interval I take $\alpha = \delta = 0.025$ which guarantees asymptotic coverage of at least 95\%.
From the resulting intervals, we see that both of our two hypothetical econometricians would report a statistically significant result even after accounting for the effects of instrument selection on inference and in spite of the conservatism of the 2-Step interval.

\begin{table}[htbp]
  \small
	\centering
	\begin{tabular}{lcccc} 
 \hline \hline 
 & \multicolumn{2}{c}{\emph{$\mu=$malfal}} & \multicolumn{2}{c}{$\mu=$\emph{rule}}\\ 
 & FMSC & posFMSC & FMSC & posFMSC\\ 
 \hline 
Na\"{i}ve & $(-1.66, -0.50)$ & $(-1.66, -0.50)$ & $(0.53, 1.14)$ & $(0.59, 1.27)$ \\ 
 1-Step & $(-1.58, -0.61)$ & $(-1.57, -0.62)$ & $(0.53, 1.12)$ & $(0.64, 1.21)$ \\ 
 2-Step & $(-1.69, -0.48)$ & $(-1.69, -0.48)$ & $(0.45, 1.22)$ & $(0.54, 1.31)$\\ 
 \hline 
\end{tabular}
	\caption{Post-selection CIs for the instrument selection exercise from Table \ref{tab:FMSC_values}.}
	\label{tab:postFMSC_CIs}
\end{table}


\section{Conclusion}
\label{sec:conclude}
This paper has introduced the FMSC, a proposal to choose moment conditions using AMSE. 
The criterion performs well in simulations, and the framework under which it is derived can be used produce confidence intervals that adjust for the effects of moment selection on subsequent inference. 
Moment selection is not a panacea, but the FMSC and related confidence interval procedures can yield sizeable benefits in empirically relevant settings.
While the discussion here concentrates on two cross-section examples, the FMSC could prove useful in any context in which moment conditions arise from more than one source. 
In a panel model, for example, the assumption of contemporaneously exogenous instruments may be plausible while that of predetermined instruments is more dubious.
Using the FMSC, we could assess whether the extra information contained in the lagged instruments outweighs their potential invalidity. 
Work in progress explores this idea in both static and dynamic panel settings by extending the FMSC to allow for simultaneous moment and model selection.
Other potentially fruitful extensions include the consideration of risk functions other than MSE, and an explicit treatment of weak identification and many moment conditions.

\newpage
\section*{Proofs}

\begin{proof}[Proof of Theorems \ref{thm:consist}, \ref{thm:normality}]
Essentially identical to the proofs of \cite{NeweyMcFadden1994} Theorems 2.6 and 3.1.
\end{proof}

\begin{proof}[Proof of Theorems \ref{thm:OLSvsIV}, \ref{thm:chooseIV}]
The proofs of both results are similar and standard, so I provide only a sketch of the argument for Theorem \ref{thm:chooseIV}. 
First substitute the DGP into the expression for $\widehat{\beta}_S$ and rearrange so that the left-hand side becomes $\sqrt{n}(\beta_S - \beta)$. 
The right-hand side has two factors: the first converges in probability to $-K_S$ by an $L_2$ argument and the second converges in distribution to $M + (0', \tau')'$ by the Lindeberg-Feller CLT. 
\end{proof}

\begin{proof}[Proof of Theorem \ref{thm:tau}]
By a mean-value expansion:
	\begin{eqnarray*}
	\widehat{\tau} &=& \sqrt{n} h_n\left(\widehat{\theta}_{v}\right) = \sqrt{n}h_n(\theta_0) + H \sqrt{n}\left(\widehat{\theta}_{v} - \theta_0\right) + o_p(1)\\
		&=&-HK_{v} \sqrt{n}g_n(\theta_0) + \mathbf{I}_q\sqrt{n}h_n(\theta_0) +o_p(1) = \Psi \sqrt{n}f_n(\theta_0) + o_p(1) 
\end{eqnarray*}
The result follows since $\sqrt{n}f_n(\theta_0) \rightarrow_d M + (0', \tau')'$ under Assumption \ref{assump:highlevel} (h).
\end{proof}

\begin{proof}[Proof of Corollary \ref{cor:tautau}]
By Theorem \ref{thm:tau} and the Continuous Mapping Theorem, we have $\widehat{\tau}\widehat{\tau}' \rightarrow_d UU'$  where $U =\Psi M + \tau$. Since $E[M]=0$, $E[UU'] = \Psi \Omega \Psi' + \tau\tau'$. 
\end{proof}

\begin{proof}[Proof of Theorem \ref{thm:DHW}] 
By Theorem \ref{thm:tauOLSvsIV}, $\sqrt{n}(\widehat{\beta}_{OLS} - \widetilde{\beta}_{TSLS}) \rightarrow_d N\left(\tau/\sigma_x^2, \Sigma \right)$
where $\Sigma =  \sigma_\epsilon^2 \left(1/\gamma^2 - 1/\sigma_x^2 \right)$.
Thus, under $H_0 \colon \tau = 0$, the DHW test statistic 
$$\widehat{T}_{DHW} = n\, \widehat{\Sigma}^{-1}(\widehat{\beta}_{OLS} - \widetilde{\beta}_{TSLS})^2 = \frac{n(\widehat{\beta}_{OLS} - \widetilde{\beta}_{TSLS})^2}{ \widehat{\sigma}_\epsilon^2 \left(1/\widehat{\gamma}^2 - 1/\widehat{\sigma}_x^2 \right)}$$
converges in distribution to a $\chi^2(1)$ random variable. 
Now, rewriting $\widehat{V}$, we find that
$$\widehat{V} = \widehat{\sigma}_\epsilon^2 \widehat{\sigma}_x^2\left(\frac{\widehat{\sigma}_v^2 }{\widehat{\gamma}^2}\right) = \widehat{\sigma}_\epsilon^2 \widehat{\sigma}_x^2\left(\frac{\widehat{\sigma}_x^2 - \widehat{\gamma}^2 }{\widehat{\gamma}^2}\right) = \widehat{\sigma}_\epsilon^2 \widehat{\sigma}_x^4\left(\frac{1}{\widehat{\gamma}^2} - \frac{1}{\widehat{\sigma}_x^2}\right) = \widehat{\sigma}_x^4 \,\widehat{\Sigma}$$
using the fact that $\widehat{\sigma}_v = \widehat{\sigma}_x^2 - \widehat{\gamma}^2$. 
Thus, to show that $\widehat{T}_{FMSC} = \widehat{T}_{DHW}$, all that remains is to establish that $\widehat{\tau}^2 = n\widehat{\sigma}_x^4 (\widehat{\beta}_{OLS} - \widetilde{\beta}_{TSLS})^2$, which we obtain as follows:
    $$\widehat{\tau}^2  =  \left[n^{-1/2} \textbf{x}'(\mathbf{y} - \mathbf{x}\widetilde{\beta})\right]^2 = n^{-1}\left[\mathbf{x}'\mathbf{x} \left( \widehat{\beta} - \widetilde{\beta}\right) \right]^2 = n^{-1}\left[n \widehat{\sigma}_x^2 \left( \widehat{\beta} - \widetilde{\beta}\right) \right]^2.$$
\end{proof}

\begin{proof}[Proof of Corollary \ref{cor:momentavg}]
We have
		$\sqrt{n}\left(\widehat{\mu} - \mu_0\right) = \sum_{S \in \mathscr{S}}\left[ \widehat{\omega}_S \sqrt{n}\left(\widehat{\mu}_S - \mu_0\right)\right]$
because the weights sum to one.
By Corollary \ref{cor:target}, 
$$\sqrt{n}\left(\widehat{\mu}_S - \mu_0\right)\rightarrow_d-\nabla_\theta\mu(\theta_0)'K_S \Xi_S \left(M +  \left[\begin{array}
	{c} 0 \\ \tau
\end{array} \right]\right)$$
and by the assumptions of this Corollary we find that $\widehat{\omega}_S \rightarrow_d\varphi_S(\tau,M)$ for each $S\in \mathscr{S}$, where $\varphi_S(\tau,M)$ is a function of $M$ and constants only. 
Hence $\widehat{\omega}_S$ and $\sqrt{n}\left(\widehat{\mu}_S - \mu_0\right)$ converge jointly in distribution to their respective functions of $M$, for all $S \in \mathscr{S}$. 
The result follows by application of the Continuous Mapping Theorem.
\end{proof}

\begin{proof}[Proof of Theorem \ref{thm:OLSvsIVavg}]
By Theorem \ref{thm:OLSvsIV}
$\sqrt{n}\left(\widehat{\beta}(\omega) - \beta \right)  \overset{d}{\rightarrow}  N\left(B(\omega), V(\omega)\right)$ 
since the weights sum to one, where
\begin{equation*}
	B(\omega) = \omega \left( \frac{\tau}{\sigma_x^2} \right), \quad 
	 V(\omega) =  \frac{\sigma_\epsilon^2}{\sigma_x^2} \left[(2\omega^2 - \omega)\left( \frac{\sigma_x^2}{\gamma^2} - 1\right)+\frac{\sigma_x^2}{\gamma^2} \right]
\end{equation*}
and accordingly
	$$\mbox{AMSE}\left[\widehat{\beta}(\omega) \right] =  \omega^2 \left(\frac{\tau^2}{\sigma_x^4} \right) + (\omega^2 - 2 \omega)\left(\frac{\sigma_\epsilon^2}{\sigma_x^2}\right)\left( \frac{\sigma_x^2}{\gamma^2} - 1\right) + \frac{\sigma_\epsilon^2}{\gamma^2}.$$
The preceding expression is a globally convex function of $\omega$. 
Taking the first order condition and rearranging gives the desired result.
\end{proof}

\begin{proof}[Proof of Theorem \ref{pro:jstat}]
By a mean-value expansion,
	$$\sqrt{n}\left[\Xi_S f_n\left(\widehat{\theta}_S\right)\right]  = \sqrt{n}\left[\Xi_S f_n(\theta_0)\right] + F_S  \sqrt{n}\left(\widehat{\theta}_S - \theta_0\right) + o_p(1).$$
Since $\sqrt{n}\left(\widehat{\theta}_S - \theta_0\right) \rightarrow_p -\left(F_S' W_S F_S  \right)^{-1}F_S'W_S\sqrt{n}\left[\Xi_S f_n(\theta_0)\right]$, we have
	$$\sqrt{n}\left[\Xi_S f_n(\widehat{\theta}_S)\right] = \left[I - F_S\left(F_S' W_S F_S  \right)^{-1}F_S'W_S\right] \sqrt{n}\left[\Xi_S f_n(\theta_0)\right] + o_p(1).$$
Thus, for estimation using the efficient weighting matrix 
$$\widehat{\Omega}^{-1/2}_S \sqrt{n}\left[\Xi_S f_n\left(\widehat{\theta}_S\right)\right] \rightarrow_d\left[I - P_S\right] \Omega_S^{-1/2}\Xi_S\left(M + \left[\begin{array}{c}0\\ \tau \end{array} \right] \right)$$
by Assumption \ref{assump:highlevel} (h), where $\widehat{\Omega}^{-1/2}_S$ is a consistent estimator of $\Omega_S^{-1/2}$ and $P_S$ is the projection matrix based on $\Omega^{-1/2}_S F_S$, the identifying restrictions.\footnote{See \cite{Hallbook}, Chapter 3.} The result follows by combining and rearranging these expressions.
\end{proof}

\begin{proof}[Proof of Theorem \ref{pro:andrews}]
Let $S_1$ and $S_2$ be arbitrary moment sets in $\mathscr{S}$ and let $|S|$ denote the cardinality of $S$. 
Further, define $\Delta_n(S_1, S_2) = MSC(S_1) - MSC(S_2)$
By Theorem \ref{pro:jstat}, $J_n(S) = O_p(1)$, $S \in \mathscr{S}$, thus
	\begin{eqnarray*}
			\Delta_n(S_1, S_2)	&=&   \left[J_{n}(S_1) - J_{n}(S_2)\right] - \left[h\left(p+|S_1|\right) - h\left(p+|S_2|\right)\right]\kappa_n\\
				&=& O_p(1) - C\kappa_n
	\end{eqnarray*}
where $C = \left[h\left(p+|S_1|\right) - h\left(p+|S_2|\right)\right]$. 
Since $h$ is strictly increasing, $C$ is positive for $|S_1|>|S_2|$, negative for $|S_1|<|S_2|$, and zero for $|S_1|=|S_2|$. 
Hence:
	\begin{eqnarray*}
		|S_1|>|S_2|&\implies& \Delta_n(S_1, S_2)  \rightarrow -\infty\\
		|S_1|=|S_2|&\implies&\Delta_n(S_1, S_2)  = O_p(1)\\
		|S_1|<|S_2|&\implies& \Delta_n(S_1, S_2)  \rightarrow \infty
\end{eqnarray*}
The result follows because the full moment set contains more moment conditions than any other moment set $S$.
\end{proof}

\begin{proof}[Proof of Theorem \ref{thm:sim}]
By Theorem \ref{thm:tau} and Corollary \ref{cor:momentavg},
\begin{eqnarray*}
	P\left\{\mu_0 \in \mbox{CI}_{sim} \right\} 
	&\rightarrow& P \left\{ a_{min} \leq \Lambda(\tau) \leq b_{max}\right\}
\end{eqnarray*}
where $a(\tau^*), b(\tau^*)$ define a collection of $(1-\alpha)\times 100\%$ intervals indexed by $\tau^*$, each of which is constructed under the assumption that $\tau = \tau^*$
$$P\left\{a(\tau^*) \leq \Lambda(\tau^*) \leq b(\tau^*) \right\} = 1-\alpha $$
and we define the shorthand $a_{min}, b_{max}$ as follows
	\begin{eqnarray*}
	a_{min}(\Psi M + \tau)&=&\min \left\{a(\tau^*)\colon \tau^* \in \mathscr{T}(\Psi M + \tau,\delta) \right\}\\
	b_{max}(\Psi M + \tau)&=&\max \left\{b(\tau^*)\colon \tau^* \in \mathscr{T}(\Psi M + \tau,\delta) \right\}\\
	\mathscr{T}(\Psi M + \tau,\delta) &=& \left\{\tau^* \colon  \Delta(\tau, \tau^*) \leq \chi^2_q(\delta) \right\}\\
	\Delta(\tau,\tau^*) &=&  (\Psi M + \tau - \tau^*)' (\Psi\Omega\Psi')^{-1} \left(\Psi M + \tau - \tau^*\right)
	\end{eqnarray*}
Now, let $A = \left\{ \Delta(\tau, \tau) \leq \chi^2_q(\delta) \right\}$ where $\chi^2_q(\delta)$ is the $1- \delta$ quantile of a $\chi^2_q$ random variable. 
This is the event that the \emph{limiting version} of the confidence region for $\tau$ contains the true bias parameter. 
Since $\Delta(\tau, \tau)\sim\chi^2_q$, $P(A) = 1 - \delta$. For every $\tau^*\in \mathscr{T}(\Psi M + \tau,\delta)$ we have
$$P\left[\left\{a(\tau^*) \leq \Lambda(\tau^*) \leq b(\tau^*)  \right\}\cap A \right] + P\left[\left\{a(\tau^*) \leq \Lambda(\tau) \leq b(\tau^*)  \right\}\cap A^c \right] = 1-\alpha$$
by decomposing $P\left\{a(\tau^*) \leq \Lambda(\tau^*) \leq b(\tau^*) \right\} $ into the sum of mutually exclusive events. 
But since
$$P\left[\left\{a(\tau^*) \leq \Lambda(\tau^*) \leq b(\tau^*)  \right\}\cap A^c \right] \leq P(A^c) = \delta$$
we see that
$$P\left[\left\{a(\tau^*) \leq \Lambda(\tau^*) \leq b(\tau^*)  \right\}\cap A \right]  
\geq 1-\alpha-\delta$$
for every $\tau^* \in \mathscr{T}(\Psi M + \tau,\delta)$. 
Now, by definition, if $A$ occurs then the true bias parameter $\tau$ is contained in $\mathscr{T}(\Psi M + \tau,\delta)$ and hence 
$$P\left[\left\{a(\tau) \leq \Lambda(\tau) \leq b(\tau)  \right\}\cap A \right]  
\geq 1-\alpha-\delta.$$
But when $\tau \in \mathscr{T}(\Psi M + \tau,\delta)$, $a_{min} \leq a(\tau)$ and $b(\tau) \leq b_{max}$. 
It follows that
	$$\left\{a(\tau) \leq \Lambda(\tau) \leq b(\tau)  \right\}\cap A \subseteq \{a_{min} \leq \Lambda(\tau) \leq b_{max}\}$$
and therefore
	$$1 -\alpha - \delta \leq P\left[\left\{a(\tau^*) \leq \Lambda(\tau^*) \leq b(\tau^*)  \right\}\cap A \right] \leq P\left[ \{a_{min} \leq \Lambda(\tau) \leq b_{max}\}\right]$$
as asserted.
\end{proof}

\bibliographystyle{elsarticle-harv}
\small
\bibliography{fmsc_refs}

\begin{thebibliography}{67}
\expandafter\ifx\csname natexlab\endcsname\relax\def\natexlab#1{#1}\fi
\expandafter\ifx\csname url\endcsname\relax
  \def\url#1{\texttt{#1}}\fi
\expandafter\ifx\csname urlprefix\endcsname\relax\def\urlprefix{URL }\fi

\bibitem[{Abramson et~al.(2013)Abramson, Audet, Couture, {Dennis, Jr.},
  {Le~Digabel}, and Tribes}]{NOMADcode}
Abramson, M., Audet, C., Couture, G., {Dennis, Jr.}, J., {Le~Digabel}, S.,
  Tribes, C., 2013. The {NOMAD} project. Software available at
  \url{http://www.gerad.ca/nomad}.

\bibitem[{Acemoglu et~al.(2001)Acemoglu, Johnson, and Robinson}]{Acemoglu}
Acemoglu, D., Johnson, S., Robinson, J.~A., 2001. The colonial origins of
  comparative development: An empirical investigation. American Economic Review
  91~(5), 1369--1401.

\bibitem[{Andrews(1988)}]{Andrews1988}
Andrews, D. W.~K., December 1988. Laws of large numbers for dependent
  non-identically distributed random variables. Econometric Theory 4~(3),
  458--467.

\bibitem[{Andrews(1992)}]{Andrews1992}
Andrews, D. W.~K., June 1992. Generic uniform convergence. Econometric Theory
  8~(2), 241--257.

\bibitem[{Andrews(1999)}]{Andrews1999}
Andrews, D. W.~K., May 1999. Consistent moment selection procedures for
  generalized methods of moments estimation. Econometrica 67~(3), 543--564.

\bibitem[{Andrews and Lu(2001)}]{AndrewsLu}
Andrews, D. W.~K., Lu, B., 2001. Consistent model and moment selection
  procedures for {GMM} estimation with application to dynamic panel data
  models. Journal of Econometrics 101, 123--164.

\bibitem[{Audet et~al.(2009)Audet, {Le~Digabel}, and Tribes}]{NOMADuserguide}
Audet, C., {Le~Digabel}, S., Tribes, C., 2009. {NOMAD} user guide. Tech. Rep.
  G-2009-37, Les cahiers du GERAD.
\newline\urlprefix\url{http://www.gerad.ca/NOMAD/Downloads/user_guide.pdf}

\bibitem[{Berger and Boos(1994)}]{Berger1994}
Berger, R.~L., Boos, D.~D., September 1994. P values maximized over a
  confidence set for the nuisance parameter. Journal of the American
  Statistical Association 89~(427), 1012--1016.

\bibitem[{Berkowitz et~al.(2008)Berkowitz, Caner, and Fang}]{Berkowitz2008}
Berkowitz, D., Caner, M., Fang, Y., 2008. Are {``Nearly Exogenous''}
  instruments reliable? Economics Letters 108, 20--23.

\bibitem[{Berkowitz et~al.(2012)Berkowitz, Caner, and Fang}]{Berkowitz2012}
Berkowitz, D., Caner, M., Fang, Y., 2012. The validity of instruments
  revisited. Journal of Econometrics 166, 255--266.

\bibitem[{Caner(2009)}]{Caner2009}
Caner, M., 2009. Lasso-type {GMM} estimator. Econometric Theory 25, 270--290.

\bibitem[{Caner(2014)}]{Caner2014}
Caner, M., 2014. Near exogeneity and weak identification in generalized
  empirical likliehood estimators: Many moment asymptotics. Journal of
  Econometrics 182, 247--268.

\bibitem[{Caner et~al.(Forthcoming)Caner, Han, and Lee}]{CanerHanLee}
Caner, M., Han, X., Lee, Y., Forthcoming. Adaptive elastic net {GMM} estimation
  with many invalid moment conditions: Simultaneous model and moment selection.
  Journal of Business and Economic Statistics.

\bibitem[{Carstensen and Gundlach(2006)}]{Carstensen2006}
Carstensen, K., Gundlach, E., 2006. The primacy of institutions reconsidered:
  Direct income effects of malaria prevelance. World Bank Economic Review
  20~(3), 309--339.

\bibitem[{Chen et~al.(2009)Chen, Jacho-Ch\'avez, and Linton}]{ChenChavezLinton}
Chen, X., Jacho-Ch\'avez, D.~T., Linton, O., June 2009. An alternative way of
  computing efficient instrumental variables estimators, lSE STICERD Research
  Paper EM/2009/536.
\newline\urlprefix\url{http://sticerd.lse.ac.uk/dps/em/em536.pdf}

\bibitem[{Cheng and Liao(2013)}]{ChengLiao}
Cheng, X., Liao, Z., October 2013. Select the valid and relevant moments: An
  information-based {LASSO} for {GMM} with many moments, {PIER Working Paper}
  13-062.
\newline\urlprefix\url{http://economics.sas.upenn.edu/system/files/13-062.pdf}

\bibitem[{Cheng et~al.(2014)Cheng, Liao, and Shi}]{ChengLiaoShi}
Cheng, X., Liao, Z., Shi, R., October 2014. Uniform asymptotic risk of
  averaging gmm estimator robust to misspecification, working Paper.

\bibitem[{Claeskens et~al.(2006)Claeskens, Croux, and Jo}]{Claeskens2006}
Claeskens, G., Croux, C., Jo, 2006. Variable selection for logistic regression
  using a prediction-focused information criterion. Biometrics 62, 972--979.

\bibitem[{Claeskens and Hjort(2003)}]{ClaeskensHjort2003}
Claeskens, G., Hjort, N.~L., 2003. The focused information criterion. Journal
  of the American Statistical Association 98~(464), 900--945.

\bibitem[{Claeskens and Hjort(2008{\natexlab{a}})}]{ClaeskensHjort2008}
Claeskens, G., Hjort, N.~L., 2008{\natexlab{a}}. Minimizing average risk in
  regression models. Econometric Theory 24, 493--527.

\bibitem[{Claeskens and Hjort(2008{\natexlab{b}})}]{ClaeskensHjortbook}
Claeskens, G., Hjort, N.~L., 2008{\natexlab{b}}. Model Selection and Model
  Averaging. Cambridge Series in Statistical and Probabilistic Mathematics.
  Cambridge.

\bibitem[{Conley et~al.(2012)Conley, Hansen, and Rossi}]{Conleyetal}
Conley, T.~G., Hansen, C.~B., Rossi, P.~E., 2012. Plausibly exogenous. Review
  of Economics and Statistics 94~(1), 260--272.

\bibitem[{Demetrescu et~al.(2011)Demetrescu, Hassler, and Kuzin}]{Demetrescu}
Demetrescu, M., Hassler, U., Kuzin, V., 2011. Pitfalls of post-model-selection
  testing: Experimental quantification. Empirical Economics 40, 359--372.

\bibitem[{Donald et~al.(2009)Donald, Imbens, and Newey}]{DonaldImbensNewey2009}
Donald, S.~G., Imbens, G.~W., Newey, W.~K., 2009. Choosing instrumental
  variables in conditional moment restriction models. Journal of Econometrics
  152, 28--36.

\bibitem[{Donald and Newey(2001)}]{DonaldNewey2001}
Donald, S.~G., Newey, W.~K., September 2001. Choosing the number of
  instruments. Econometrica 69~(5), 1161--1191.

\bibitem[{Easterly and Levine(2003)}]{Easterly}
Easterly, W., Levine, R., 2003. Tropics, germs, and crops: how endowments
  influence economic development. Journal of Monetary Economics 50, 3--39.

\bibitem[{Eddelbuettel(2013)}]{RcppBook}
Eddelbuettel, D., 2013. Seamless {R} and {C++} Integration with {Rcpp}.
  Springer, New York, iSBN 978-1-4614-6867-7.

\bibitem[{Eddelbuettel and Fran\c{c}ois(2011)}]{RcppArticle}
Eddelbuettel, D., Fran\c{c}ois, R., 2011. {Rcpp}: Seamless {R} and {C++}
  integration. Journal of Statistical Software 40~(8), 1--18.
\newline\urlprefix\url{http://www.jstatsoft.org/v40/i08/}

\bibitem[{Eddelbuettel and Sanderson(2014)}]{RcppArmadillo}
Eddelbuettel, D., Sanderson, C., March 2014. {RcppArmadillo}: Accelerating {R}
  with high-performance {C++} linear algebra. Computational Statistics and Data
  Analysis 71, 1054--1063.
\newline\urlprefix\url{http://dx.doi.org/10.1016/j.csda.2013.02.005}

\bibitem[{Fox et~al.(2014)Fox, Nie, and Byrnes}]{sem}
Fox, J., Nie, Z., Byrnes, J., 2014. sem: Structural Equation Models. R package
  version 3.1-4.
\newline\urlprefix\url{http://CRAN.R-project.org/package=sem}

\bibitem[{Guggenberger(2010)}]{Guggenberger2010}
Guggenberger, P., 2010. The impact of a {H}ausman pretest on the asymptotic
  size of a hypothesis test. Econometric Theory 26, 369--382.

\bibitem[{Guggenberger(2012)}]{Guggenberger2012}
Guggenberger, P., 2012. On the asymptotic size distortion of tests when
  instruments locally violate the exogeneity assumption. Econometric Theory 28,
  387--421.

\bibitem[{Guggenberger and Kumar(2012)}]{GuggenbergerKumar}
Guggenberger, P., Kumar, G., 2012. On the size distortion of tests after an
  overidentifying restrictions pretest. Journal of Applied Econometrics 27,
  1138--1160.

\bibitem[{Hall(2005)}]{Hallbook}
Hall, A.~R., 2005. Generalized Method of Moments. Advanced Texts in
  Econometrics. Oxford.

\bibitem[{Hall and Peixe(2003)}]{HallPeixe2003}
Hall, A.~R., Peixe, F.~P., 2003. A consistent method for the selection of
  relevant instruments in linear models. Econometric Reviews 22, 269--288.

\bibitem[{Hansen(2015{\natexlab{a}})}]{HansenShrink}
Hansen, B.~E., 2015{\natexlab{a}}. Efficient shrinkage in parametric models,
  {U}niversity of {W}isconsin.

\bibitem[{Hansen(2015{\natexlab{b}})}]{HansenStein}
Hansen, B.~E., 2015{\natexlab{b}}. A {S}tein-like {2SLS} estimator, forthcoming
  in Econometric Reviews.

\bibitem[{Hjort and Claeskens(2003)}]{HjortClaeskens}
Hjort, N.~L., Claeskens, G., 2003. Frequentist model average estimators.
  Journal of the American Statistical Association 98~(464), 879--899.

\bibitem[{Hong et~al.(2003)Hong, Preston, and Shum}]{HongPrestonShum}
Hong, H., Preston, B., Shum, M., 2003. Generalized empirical likelihood-based
  model selection for moment condition models. Econometric Theory 19, 923--943.

\bibitem[{Jana(2005)}]{Jana2005}
Jana, K., 2005. Canonical correlations and instrument selection in
  econometrics. Ph.D. thesis, North Carolina State University.
\newline\urlprefix\url{http://www.lib.ncsu.edu/resolver/1840.16/4315}

\bibitem[{Judge and Mittelhammer(2007)}]{Judge2007}
Judge, G.~G., Mittelhammer, R.~C., 2007. Estimation and inference in the case
  of competing sets of estimating equations. Journal of Econometrics 138,
  513--531.

\bibitem[{Kabaila(1998)}]{Kabaila1998}
Kabaila, P., 1998. Valid confidence intervals in regressions after variable
  selection. Econometric Theory 14, 463--482.

\bibitem[{Kabaila and Leeb(2006)}]{KabailaLeeb2006}
Kabaila, P., Leeb, H., 2006. On the large-sample minimal coverage probability
  of confidence intervals after model selection. Journal of the American
  Statistical Association 101~(474), 819--829.

\bibitem[{Kinal(1980)}]{Kinal}
Kinal, T.~W., 1980. The existence of moments of k-class estimators.
  Econometrica 48~(1), 241--249.

\bibitem[{Kraay(2012)}]{Kraay}
Kraay, A., 2012. Instrumental variables regressions with uncertain exclusion
  restrictions: A {Bayesian} approach. Journal of Applied Econometrics 27~(1),
  108--128.

\bibitem[{Kuersteiner and Okui(2010)}]{KuersteinerOkui2010}
Kuersteiner, G., Okui, R., March 2010. Constructing optimal instruments by
  first-stage prediction averaging. Econometrica 78~(2), 679--718.

\bibitem[{{Le~Digabel}(2011)}]{NOMADalgorithm}
{Le~Digabel}, S., 2011. Algorithm 909: {NOMAD}: Nonlinear optimization with the
  {MADS} algorithm. {ACM} Transactions on Mathematical Software 37~(4), 1--15.

\bibitem[{Leeb and P\"{o}tscher(2005)}]{LeebPoetscher2005}
Leeb, H., P\"{o}tscher, B.~M., 2005. Model selection and inference: Facts and
  fiction. Econometric Theory 21~(1), 21--59.

\bibitem[{Leeb and P\"{o}tscher(2008)}]{LeebPoetscher2008}
Leeb, H., P\"{o}tscher, B.~M., 2008. Sparse estimators and the oracle property,
  or the return of {Hodges'} estimator. Journal of Econometrics 142, 201--211.

\bibitem[{Leeb and P\"{o}tscher(2009)}]{LeebPoetscher2009}
Leeb, H., P\"{o}tscher, B.~M., 2009. Model selection. In: Handbook of Financial
  Time Series. Springer.

\bibitem[{Leeb and P\"{o}tscher(2014)}]{Leeb}
Leeb, H., P\"{o}tscher, B.~M., May 2014. Testing in the presence of nuisance
  parameters: Some comments on tests post-model-selection and random critical
  values, {U}niversity of {V}ienna.
\newline\urlprefix\url{http://arxiv.org/pdf/1209.4543.pdf}

\bibitem[{Liao(2013)}]{Liao}
Liao, Z., November 2013. Adaptive {GMM} shrinkage estimation with consistent
  moment selection. Econometric Theory 29, 857--904.

\bibitem[{Loh(1985)}]{Loh1985}
Loh, W.-Y., 1985. A new method for testing separate families of hypotheses.
  Journal of the American Statistical Association 80~(390), 362--368.

\bibitem[{McCloskey(2012)}]{McCloskey}
McCloskey, A., October 2012. Bonferroni-based size-correction for nonstandard
  testing problems, {B}rown {U}niversity.
\newline\urlprefix\url{http://www.econ.brown.edu/fac/adam_mccloskey/Research_files/McCloskey_BBCV.pdf}

\bibitem[{Newey(1985)}]{Newey1985}
Newey, W.~K., 1985. Genearlized method of moments specification testing.
  Journal of Econometrics 29, 229--256.

\bibitem[{Newey and McFadden(1994)}]{NeweyMcFadden1994}
Newey, W.~K., McFadden, D., 1994. Large Sample Estimation and Hypothesis
  Testing. Vol.~IV. Elsevier Science, Ch.~36, pp. 2111--2245.

\bibitem[{Phillips(1980)}]{Phillips1980}
Phillips, P. C.~B., 1980. The exact distribution of instrumental variables
  estimators in an equation containing $n+1$ endogenous variables. Econometrica
  48~(4), 861--878.

\bibitem[{{R Core Team}(2014)}]{R}
{R Core Team}, 2014. R: A Language and Environment for Statistical Computing. R
  Foundation for Statistical Computing, Vienna, Austria.
\newline\urlprefix\url{http://www.R-project.org/}

\bibitem[{Racine and Nie(2014)}]{crs}
Racine, J.~S., Nie, Z., 2014. crs: Categorical Regression Splines. R package
  version 0.15-22.
\newline\urlprefix\url{http://CRAN.R-project.org/package=crs}

\bibitem[{Rodrik et~al.(2004)Rodrik, Subramanian, and Trebbi}]{Rodrik}
Rodrik, D., Subramanian, A., Trebbi, F., 2004. Institutions rule: The primacy
  of institutions over geography and integration in economic development.
  Journal of Economic Growth 9, 131--165.

\bibitem[{Sachs(2003)}]{Sachs}
Sachs, J.~D., February 2003. Institutions don't rule: Direct effects of
  geography on per capita income, {NBER} Working Paper No. 9490.
\newline\urlprefix\url{http://www.nber.org/papers/w9490}

\bibitem[{Sanderson(2010)}]{Armadillo}
Sanderson, C., 2010. {Armadillo}: An open source {C++} linear algebra library
  for fast prototyping and computationally intensive experiments. Tech. rep.,
  NICTA.
\newline\urlprefix\url{http://arma.sourceforge.net/armadillo_nicta_2010.pdf}

\bibitem[{Schorfheide(2005)}]{Schorfheide2005}
Schorfheide, F., 2005. {VAR} forecasting under misspecification. Journal of
  Econometrics 128, 99--136.

\bibitem[{Sharpsteen and Bracken(2013)}]{tikzDevice}
Sharpsteen, C., Bracken, C., 2013. tikzDevice: R Graphics Output in LaTeX
  Format. R package version 0.7.0.
\newline\urlprefix\url{http://CRAN.R-project.org/package=tikzDevice}

\bibitem[{Silvapulle(1996)}]{Silvapulle1996}
Silvapulle, M.~J., December 1996. A test in the presence of nuisance
  parameters. Journal of the American Statistical Association 91~(436),
  1690--1693.

\bibitem[{Xiao(2010)}]{Xiao}
Xiao, Z., 2010. The weighted method of moments approach for moment condition
  models. Economics Letters 107, 183--186.

\bibitem[{Yang(2005)}]{Yang2005}
Yang, Y., 2005. Can the strengths of {AIC} and {BIC} be shared? a conflict
  between model identification and regression estimation. Biometrika 92~(4),
  937--950.

\end{thebibliography}

\appendix
\normalsize
\clearpage
\pagenumbering{arabic}
\renewcommand*{\thepage}{A-\arabic{page}}
\numberwithin{equation}{section}
\numberwithin{table}{section}
\numberwithin{figure}{section}
\begin{center}
  {\Huge Online Appendix\\}
  \vspace{1em}
  {\Large Using Invalid Instruments on Purpose: Focused Moment \\
  Selection and Averaging for GMM \\}
  \vspace{2em}
  {\large Francis J.\ DiTraglia} 
\end{center}
\section{Computational Details}
\label{append:comp}
This paper is fully replicable using freely available, open-source software.
For full source code and replication details, see \url{https://github.com/fditraglia/fmsc}.
Results for the simulation studies and empirical example were generated using \texttt{R} \citep{R} and C\raisebox{0.5ex}{\tiny\textbf{++}} via the \texttt{Rcpp} \citep{RcppArticle,RcppBook} and \texttt{RcppArmadillo} \citep{RcppArmadillo} packages.
\texttt{RcppArmadillo} provides an interface to the Armadillo C\raisebox{0.5ex}{\tiny\textbf{++}} linear algebra library \citep{Armadillo}.
All figures in the paper were converted to tikz using the \texttt{tikzDevice} package \citep{tikzDevice}.
Confidence interval calculations for Sections \ref{sec:limitexperiment} and \ref{sec:CIsim} rely routines from my R package \texttt{fmscr}, available from \url{https://github.com/fditraglia/fmscr}.
The simulation-based intervals for the empirical example from Section \ref{sec:application} were constructed following Algorithm \ref{alg:conf} with  $J = 10,000$ using a mesh-adaptive search algorithm provided by the NOMAD C\raisebox{0.5ex}{\tiny\textbf{++}} optimization package \citep{NOMADalgorithm,NOMADcode,NOMADuserguide}, called from R using the \texttt{crs} package \citep{crs}.
TSLS results for Table \ref{tab:fullresults} were generated using version 3.1-4 of the \texttt{sem} package \citep{sem}.

\section{Failure of the Identification Condition}
\label{sec:digress}
When there are fewer moment conditions in the $g$-block than elements of the parameter vector $\theta$, i.e.\ when $r > p$, Assumption \ref{assump:Identification} fails: $\theta_0$ is not estimable by $\widehat{\theta}_v$ so $\widehat{\tau}$ is an infeasible estimator of $\tau$. 
A na\"{i}ve approach to this problem would be to substitute another consistent estimator of $\theta_0$ and proceed analogously. 
Unfortunately, this approach fails. To understand why, consider the case in which all moment conditions are potentially invalid so that the $g$--block is empty. 
Letting $\widehat{\theta}_f$ denote the estimator based on the full set of moment conditions in $h$,  $\sqrt{n}h_n(\widehat{\theta}_f) \rightarrow_d\Gamma  \mathcal{N}_q(\tau, \Omega)$ where $\Gamma = \mathbf{I}_q - H \left(H'WH\right)^{-1}H'W$, using an argument similar to that in the proof of Theorem \ref{thm:tau}. 
The mean, $\Gamma \tau$, of the resulting limit distribution does not equal $\tau$, and because $\Gamma$ has rank $q-r$ we cannot pre-multiply by its inverse to extract an estimate of $\tau$.
Intuitively, $q-r$ over-identifying restrictions are insufficient to estimate a $q$-vector: $\tau$ cannot be estimated without a minimum of $r$ valid moment conditions. 
However, the limiting distribution of $\sqrt{n}h_n(\widehat{\theta}_f)$ partially identifies $\tau$ even when we have no valid moment conditions at our disposal. 
A combination of this information with prior restrictions on the magnitude of the components of $\tau$ allows the use of the FMSC framework to carry out a sensitivity analysis when $r>p$. 
For example, the worst-case estimate of AMSE over values of $\tau$ in the identified region could still allow certain moment sets to be ruled out.
This idea shares similarities with \citet{Kraay} and \citet{Conleyetal}, two recent papers that suggest methods for evaluating the robustness of conclusions drawn from IV regressions when the instruments used may be invalid.
\section{Trimmed MSE}
\label{append:trim}
    Even in situations where finite sample MSE does not exist, it is still meaningful to consider comparisons of asymptotic MSE.
    To make the connection between the finite-sample and limit experiment a bit tidier in this case we can work in terms of \emph{trimmed} MSE,  following \cite{HansenShrink}. 
    To this end, define
    \begin{align*}
      MSE_n(\widehat{\mu}_S, \zeta) &= E\left[ \min\left\{ n(\widehat{\mu} - \mu_0)^2, \zeta  \right\} \right]\\
      AMSE(\widehat{\mu}_S) &= \lim_{\zeta \rightarrow \infty} \liminf_{n\rightarrow \infty} MSE_n(\widehat{\mu}_S, \zeta)
    \end{align*}
    where $\zeta$ is a positive constant that bounds the expectation for finite $n$.
    By Corollary \ref{cor:target} $\sqrt{n}(\widehat{\mu}_S-\mu_0) \rightarrow_d \Lambda$ where $\Lambda$ is a normally distributed random variable.
    Thus, by Lemma 1 of \cite{HansenShrink}, we have $AMSE(\widehat{\mu}_S) = E[\Lambda^2]$.
    In other words, working with a sequence of trimmed MSE functions leaves AMSE unchanged while ensuring that finite-sample risk is bounded.
    This justifies the approximation $MSE_n(\widehat{\mu}_S, \zeta) \approx E[\Lambda^2]$ for large $n$ and $\zeta$.
    In a simulation exercise in which ordinary MSE does not exist, for example instrumental variables with a single instrument, one could remove the largest 1\% of simulation draws in absolute value and evaluate the performance of the FMSC against the empirical MSE calculated for the remaining draws.

\section{The Case of Multiple Target Parameters}
\label{append:mult}
  The fundamental idea behind the FMSC is to approximate finite-sample risk with asymptotic risk under local mis-specification.
  Although the discussion given above is restricted to a scalar target parameter, the same basic idea is easily extended to accomodate a vector of target parameters.
  All that is required is to specify an appropriate risk function.
  Consider a generic weighted quadratic risk function of the form
  \[
    R(\widehat{\theta}_S,W) = E\left[ \left( \widehat{\theta}_S - \theta_0 \right)' W \left( \widehat{\theta} - \theta_0 \right)\right]
  \]
  where $W$ is a positive semi-definite matrix.
  The finite-sample distribution of $\widehat{\theta}$ is, in general, unknown, but by Theorem \ref{thm:normality} $\sqrt{n} \left( \widehat{\theta}_S - \theta_0\right) \rightarrow_d U_S$ where
  \[
    U_S = -K_S \Xi_S \left( M + \left[
    \begin{array}{c}
      0 \\ \tau
    \end{array}
  \right]\right)
  \]
  and $M\sim N(0,\Omega)$ so we instead consider the \emph{asymptotic risk}
  \begin{equation*}
    AR(\widehat{\theta}_S, W) = E\left[ U_S' W U_S \right] =
     \mbox{trace}\left\{ W^{1/2}K_S \Xi_S 
      \left( \left[
      \begin{array}{cc}
        0 & 0 \\
        0 & \tau \tau'
      \end{array}
    \right] + \Omega\right)
    \Xi_S' K_S'W^{1/2} \right\}
  \end{equation*}
  where $W^{1/2}$ is the symmetric, positive semi-definite square root of $W$.
  To construct an asymptotically unbiased estimator of $AR(\widehat{\theta}_S, W)$ we substitute consistent estimators of $\Omega$ and $K_S$ and the asymptotically unbiased estimator of $\widehat{\tau}\widehat{\tau}'$ from Corollary \ref{cor:tautau} yielding
  \[
    \widehat{AR}\left( \widehat{\theta}_S, W \right) = 
    \mbox{trace}\left\{ W^{1/2}\widehat{K}_S \Xi_S 
      \left( \left[
      \begin{array}{cc}
        0 & 0 \\
        0 & \widehat{\tau}\widehat{\tau}' - \widehat{\Psi}\widehat{\Omega}\widehat{\Psi}
      \end{array}
    \right] + \Omega\right)
    \Xi_S' \widehat{K}_S'W^{1/2} \right\}
  \]
  which is nearly identical to the expression for the FMSC given in Equation \ref{eq:fmsc}.
  The only difference is the presence of the weighting matrix $W$ and the trace operator in place of the vector of derivatives $\nabla_\theta\mu(\widehat{\theta})$.
  When $W$ is a diagonal matrix this difference disappears completely as this effectively amounts to defining a target parameter that is a weighted average of some subset of the elements of $\theta$.
  In this case the FMSC can be used without modification simply by defining the function $\mu$ appropriately.

\section{Low-Level Sufficient Conditions}
\label{sec:sufficient_conditions}

\begin{assump}[Sufficient Conditions for Theorem \ref{thm:OLSvsIV}]
\label{assump:OLSvsIV}
	Let $\{(\mathbf{z}_{ni}, v_{ni}, \epsilon_{ni})\colon 1\leq i \leq n, n = 1, 2, \hdots\}$ be a triangular array of random variables such that
	\begin{enumerate}[(a)]
		\item $(\mathbf{z}_{ni}, v_{ni}, \epsilon_{ni}) \sim$ iid and mean zero within each row of the array (i.e.\ for fixed $n$)
		\item $E[\mathbf{z}_{ni} \epsilon_{ni}]=\mathbf{0}$, $E[\mathbf{z}_{ni} v_{ni}]=\mathbf{0}$, and $E[\epsilon_{ni}v_{ni}] = \tau/\sqrt{n}$ for all $n$
		\item $E[\left|\mathbf{z}_{ni}\right|^{4+\eta}] <C$, $E[\left|\epsilon_{ni}\right|^{4+\eta}] <C$, and $E[\left|v_{ni}\right|^{4+\eta}] <C$ for some $\eta >0$, $C <\infty$
		\item $E[\mathbf{z}_{ni} \mathbf{z}_{ni}'] \rightarrow Q>0$, $E[v_{ni}^2]\rightarrow \sigma_v^2 >0$, and $E[\epsilon_{ni}^2] \rightarrow \sigma_\epsilon^2 >0$ as $n\rightarrow \infty$
		\item As $n\rightarrow \infty$, $E[\epsilon_{ni}^2 \mathbf{z}_{ni} \mathbf{z}_{ni}']- E[\epsilon_{ni}^2]E[ \mathbf{z}_{ni} \mathbf{z}_{ni}'] \rightarrow 0$, $E[\epsilon_i^2 v_{ni} \mathbf{z}_{ni}'] - E[\epsilon_{ni}^2]E[v_{ni} \mathbf{z}_{ni}'] \rightarrow 0$, and $E[\epsilon_{ni}^2 v_{ni}^2] - E[\epsilon_{ni}^2]E[v_{ni}^2] \rightarrow 0$
		\item $x_{ni} = \mathbf{z}_{ni}'\boldsymbol{\pi} + v_i$ where $\boldsymbol{\pi} \neq \mathbf{0}$, and $y_{ni} = \beta x_{ni} + \epsilon_{ni}$
	\end{enumerate}
\end{assump}
Parts (a), (b) and (d) correspond to the local mis-specification assumption, part (c) is a set of moment restrictions, and (f) is simply the DGP.
Part (e) is the homoskedasticity assumption: an \emph{asymptotic} restriction on the joint distribution of $v_{ni}$, $\epsilon_{ni}$, and $\mathbf{z}_{ni}$. 
This condition holds automatically, given the other asssumptions, if $(\mathbf{z}_{ni}, v_{ni}, \epsilon_{ni})$ are jointly normal, as in the simulation experiment described in the paper.

\begin{assump}[Sufficient Conditions for Theorem \ref{thm:chooseIV}.]
\label{assump:chooseIV} 
	Let $\{(\mathbf{z}_{ni}, \mathbf{v}_{ni}, \epsilon_{ni})\colon 1\leq i \leq n, n = 1, 2, \hdots\}$ be a triangular array of random variables with $\mathbf{z}_{ni} = (\mathbf{z}_{ni}^{(1)}$, $\mathbf{z}_{ni}^{(1)})$ such that
	\begin{enumerate}[(a)]
		\item $(\mathbf{z}_{ni}, \mathbf{v}_{ni}, \epsilon_{ni}) \sim$ iid within each row of the array (i.e.\ for fixed $n$)
		\item $E[\mathbf{v}_{ni}\mathbf{z}_{ni}']=\mathbf{0}$, $E[\mathbf{z}^{(1)}_{ni} \epsilon_{ni}]=\mathbf{0}$, and $E[\mathbf{z}^{(2)}_{ni} \epsilon_{ni}] = \boldsymbol{\tau}/\sqrt{n}$ for all $n$
		\item $E[\left|\mathbf{z}_{ni}\right|^{4+\eta}] <C$, $E[\left|\epsilon_{ni}\right|^{4+\eta}] <C$, and $E[\left|\mathbf{v}_{ni}\right|^{4+\eta}] <C$ for some $\eta >0$, $C <\infty$
		\item $E[\mathbf{z}_{ni} \mathbf{z}_{ni}'] \rightarrow Q>0$ and $E[\epsilon_{ni}^2 \mathbf{z}_{ni} \mathbf{z}_{ni}'] \rightarrow \Omega >0$ as $n\rightarrow \infty$
		\item $\mathbf{x}_{ni} =  \Pi_1' \mathbf{z}_{ni}^{(1)} + \Pi_2'\mathbf{z}_{ni}^{(2)} + \mathbf{v}_{ni}$ where $\Pi_1 \neq \mathbf{0}$, $\Pi_2 \neq \mathbf{0}$, and $y_i = \mathbf{x}_{ni}' \beta +  \epsilon_{ni}$
	\end{enumerate}
\end{assump}
These conditions are similar to although more general than those contained in Assumption \ref{assump:OLSvsIV} as they do not impose homoskedasticity.

\section{A Special Case of Post-FMSC Inference}
This appendix presents calculations and numerical results to supplement Section \ref{sec:limitexperiment}.

\subsection{The Limit Experiment}
\label{append:limitexperiment}
The joint limit distribution for the OLS versus TSLS example from Section \ref{sec:OLSvsIVExample} is as follows
\begin{equation*}
  \left[ 
  \begin{array}{c}
    \sqrt{n} \left( \widehat{\beta}_{OLS} - \beta \right)\\
    \sqrt{n} \left( \widetilde{\beta}_{TSLS} - \beta \right)\\
    \widehat{\tau}
\end{array}
\right] \overset{d}{\rightarrow} N\left( \left[
\begin{array}{c}
  \tau/\sigma_x^2 \\ 0 \\ \tau
\end{array}
\right], \sigma_{\epsilon}^2 
\left[
\begin{array}{ccc}
  1/\sigma_x^2 & 1/\sigma_x^2 & 0\\
  1/\sigma_x^2 & 1/\gamma^2 & -\sigma_v^2/\gamma^2\\
  0 & -\sigma_v^2/\gamma^2 & \sigma_x^2 \sigma_v^2/\gamma^2
\end{array}
\right]\right).
\end{equation*}
Now consider a slightly simplified version of the choosing instrumental variables example from Section \ref{sec:chooseIVexample}, namely
\begin{eqnarray*}
  y_{ni} &=& \beta x_{ni} + \epsilon_{ni}\\
  x_{ni} &=& \gamma w_{ni} + \mathbf{z}_{ni}' \boldsymbol{\pi} + v_{ni}
\end{eqnarray*}
where $x$ is the endogenous regressor of interest, $\mathbf{z}$ is a vector of exogenous instruments, and $w$ is a single potentially endogenous instrument.
Without loss of generality assume that $w$ and $\mathbf{z}$ are uncorrelated and that all random variables are mean zero.
For simplicity, further assume that the errors satisfy the same kind of asymptotic homoskedasticity condition used in the OLS versus TSLS example so that TSLS is the efficient GMM estimator.
Let the ``Full'' estimator denote the TSLS estimator using $w$ and $\mathbf{z}$ and the ``Valid'' estimator denote the TSLS estimator using only $\mathbf{z}$.
Then we have,
\begin{equation*}
  \left[ 
  \begin{array}{c}
    \sqrt{n} \left( \widehat{\beta}_{Full} - \beta \right)\\
    \sqrt{n} \left( \widetilde{\beta}_{Valid} - \beta \right)\\
    \widehat{\tau}
\end{array}
\right] \overset{d}{\rightarrow} N\left( \left[
\begin{array}{c}
  \tau\gamma/q^2_{F} \\ 0 \\ \tau
\end{array}
\right], \sigma_{\epsilon}^2 
\left[
\begin{array}{ccc}
  1/q_{F}^2 & 1/q_{F}^2 & 0\\
  1/q_{F}^2 & 1/q_{V}^2 & -\gamma\sigma_w^2/q^2_{V}\\ 
  0 & -\gamma\sigma_w^2/q^2_{V} & \sigma_w^2 q^2_{F}/q^2_{V}
\end{array}
\right]\right)
\end{equation*}
where $q^2_{F} = \gamma^2 \sigma_w^2 + q^2_{V}$, $q^2_{V} = \boldsymbol{\pi}'\Sigma_{zz}\boldsymbol{\pi}$, $\Sigma_{zz}$ is the covariance matrix of the valid instruments $\mathbf{z}$, and $\sigma_w^2$ is the variance of the ``suspect'' instrument $w$.
After some algebraic manipulations we see that both of these examples share the same structure, namely
\begin{equation}
  \label{eq:limitExperiment}
  \left[ 
  \begin{array}{c}
    \sqrt{n} \left( \widehat{\beta} - \beta \right)\\
    \sqrt{n} \left( \widetilde{\beta} - \beta \right)\\
    \widehat{\tau}
\end{array}
\right] \overset{d}{\rightarrow} 
\left[
\begin{array}{c}
  U \\ V \\ T
\end{array}
\right] \sim
N\left( \left[
\begin{array}{c}
  c\tau\\ 0 \\ \tau
\end{array}
\right], 
\left[
\begin{array}{ccc}
  \eta^2 & \eta^2 & 0\\
  \eta^2 & \eta^2 + c^2 \sigma^2 & -c\sigma^2\\ 
  0 & -c\sigma^2 & \sigma^2 
\end{array}
\right]\right)
\end{equation}
where $\widehat{\beta}$ denotes the low variance but possibly biased estimator, and $\widetilde{\beta}$ denotes the higher variance but unbiased estimator.
For any example with a limit distribution that takes this form, simple algebra shows that FMSC selection amounts to choosing $\widehat{\beta}$ whenever $|\widehat{\tau}|<\sqrt{2}\sigma$, and choosing $\widetilde{\beta}$ otherwise, in other words
\begin{eqnarray*}
  \sqrt{n}(\widehat{\beta}_{FMSC} - \beta) = \mathbf{1}\left\{ |\widehat{\tau}|<\sigma \sqrt{2} \right\} \sqrt{n}(\widehat{\beta} - \beta) +  \mathbf{1}\left\{ |\widehat{\tau}|\geq\sigma \sqrt{2} \right\}\sqrt{n}(\widetilde{\beta} - \beta)
\end{eqnarray*}
and so by the Continuous Mapping Theorem,
\begin{equation*}
  \sqrt{n}(\widehat{\beta}_{FMSC} - \beta) \overset{d}{\rightarrow}  \mathbf{1}\left\{ |T|<\sigma \sqrt{2} \right\} U +  \mathbf{1}\left\{ |T|\geq\sigma \sqrt{2} \right\} V.
\end{equation*}
Re-expressing Equation \ref{eq:limitExperiment} in terms of the marginal distribution of $T$ and the conditional distribution of $U$ and $V$ given $T$, we find that $T \sim N(\tau, \sigma^2)$ and 
\begin{equation*}
  \left.\left[
  \begin{array}{c}
   U \\ V 
  \end{array}
\right]\right| (T = t) \sim N\left(
\left[
\begin{array}{c}
  c \tau \\ c\tau - ct
\end{array}
\right], \eta^2
\left[
\begin{array}{cc}
  1 & 1 \\ 1 & 1
\end{array}
\right]
\right)
\end{equation*}
which is a \emph{singular distribution}.
While $U$ is independent of $T$, but conditional on $T$ the random variables $U$ and $V$ are perfectly correlated with the same variance.
Given $T$, the only difference between $U$ and $V$ is that the mean of $V$ is shifted by a distance that depends on the realization $t$ of $T$.
Thus, the conditional distribution of $V$ shows a \emph{random bias}: on average $V$ has mean zero because the mean of $T$ is $\tau$ but any particular realization $t$ of $T$ will not in general equal $\tau$.
Using the form of the conditional distributions we can express the distribution of $(U,V,T)'$ in a more transparent form as
\begin{eqnarray*}
  T &=& \sigma Z_1 + \tau\\
  U &=& \eta Z_2 + c\tau\\
  V &=& \eta Z_2 - c\sigma Z_1
\end{eqnarray*}
where $Z_1, Z_2$ are independent standard normal random variables.

\subsection{Numerical Results for the 2-Step Interval}
\label{append:limitexperiment_2step}
For the two-step procedure I take lower and upper bounds over a collection of equal-tailed intervals. 
It does not necessarily follow that the bounds over these intervals would be tighter if each interval in the collection were constructed to be a short as possible.
As we are free when using the 2-Step interval to choose any pair $(\alpha_1, \alpha_2)$ such that $\alpha_1 + \alpha_2 = \alpha$ I experimented with three possibilities: $\alpha_1 = \alpha_2 = \alpha/2$, followed by $\alpha_1 = \alpha/4, \alpha_2 = 3\alpha/4$ and $\alpha_1 = 3\alpha/4, \alpha_2 = \alpha/4$.
Setting $\alpha_1 = \alpha/4$ produced the shortest intervals so I discuss only results for the middle configuration here.
Additional results are available on request.
As we see from Table \ref{tab:Limit2StepWideTauOLSvsIV} for the OLS versus TSLS example and Table \ref{tab:Limit2StepWideTauChooseIVs} for the choosing IVs example, this procedure delivers on its promise that asymptotic coverage will never fall below $1-\alpha$ but this comes at the cost of extreme conservatism and correspondingly wider intervals.

\begin{table}[h]
  \centering
  \begin{subtable}{0.48\textwidth}
    \caption{Coverage Probability}
    \begin{tabular}{r|rrrrrr}
\hline\hline
 &\multicolumn{6}{c}{$\tau$} \\ 
 $\alpha = 0.05$ & $0$ & $1$ & $2$ & $3$ & $4$ & $5$ \\ 
 \hline$0.1$ & $97$ & $97$ & $97$ & $98$ & $98$ & $98$\\ 
$\pi^2\;\;\;$ $0.2$ & $97$ & $97$ & $98$ & $97$ & $97$ & $97$\\ 
$0.3$ & $98$ & $98$ & $98$ & $97$ & $96$ & $97$\\ 
$0.4$ & $98$ & $98$ & $97$ & $96$ & $97$ & $98$\\ 
 \hline 
 \end{tabular}
 
 \vspace{2em} 
 
\begin{tabular}{r|rrrrrr}
\hline\hline
 &\multicolumn{6}{c}{$\tau$} \\ 
 $\alpha = 0.1$ & $0$ & $1$ & $2$ & $3$ & $4$ & $5$ \\ 
 \hline$0.1$ & $94$ & $95$ & $96$ & $96$ & $95$ & $94$\\ 
$\pi^2\;\;\;$ $0.2$ & $95$ & $96$ & $96$ & $95$ & $94$ & $93$\\ 
$0.3$ & $96$ & $96$ & $95$ & $94$ & $92$ & $94$\\ 
$0.4$ & $96$ & $95$ & $94$ & $92$ & $94$ & $95$\\ 
 \hline 
 \end{tabular}
 
 \vspace{2em} 
 
\begin{tabular}{r|rrrrrr}
\hline\hline
 &\multicolumn{6}{c}{$\tau$} \\ 
 $\alpha = 0.2$ & $0$ & $1$ & $2$ & $3$ & $4$ & $5$ \\ 
 \hline$0.1$ & $91$ & $92$ & $92$ & $91$ & $90$ & $90$\\ 
$\pi^2\;\;\;$ $0.2$ & $93$ & $92$ & $91$ & $89$ & $87$ & $85$\\ 
$0.3$ & $93$ & $92$ & $89$ & $86$ & $85$ & $89$\\ 
$0.4$ & $93$ & $91$ & $86$ & $85$ & $88$ & $89$\\ 
 \hline 
 \end{tabular}
  \end{subtable}
  ~
  \begin{subtable}{0.48\textwidth}
    \caption{Relative Width}
    \begin{tabular}{r|rrrrrr}
\hline\hline
 &\multicolumn{6}{c}{$\tau$} \\ 
 $\alpha = 0.05$ & $0$ & $1$ & $2$ & $3$ & $4$ & $5$ \\ 
 \hline$0.1$ & $114$ & $115$ & $117$ & $119$ & $123$ & $126$\\ 
$\pi^2\;\;\;$ $0.2$ & $116$ & $117$ & $120$ & $121$ & $125$ & $126$\\ 
$0.3$ & $117$ & $117$ & $120$ & $122$ & $123$ & $123$\\ 
$0.4$ & $116$ & $118$ & $120$ & $121$ & $121$ & $120$\\ 
 \hline 
 \end{tabular}
 
 \vspace{2em} 
 
\begin{tabular}{r|rrrrrr}
\hline\hline
 &\multicolumn{6}{c}{$\tau$} \\ 
 $\alpha = 0.1$ & $0$ & $1$ & $2$ & $3$ & $4$ & $5$ \\ 
 \hline$0.1$ & $121$ & $123$ & $125$ & $128$ & $129$ & $131$\\ 
$\pi^2\;\;\;$ $0.2$ & $122$ & $124$ & $126$ & $129$ & $130$ & $131$\\ 
$0.3$ & $123$ & $125$ & $126$ & $127$ & $128$ & $128$\\ 
$0.4$ & $123$ & $123$ & $124$ & $125$ & $125$ & $123$\\ 
 \hline 
 \end{tabular}
 
 \vspace{2em} 
 
\begin{tabular}{r|rrrrrr}
\hline\hline
 &\multicolumn{6}{c}{$\tau$} \\ 
 $\alpha = 0.2$ & $0$ & $1$ & $2$ & $3$ & $4$ & $5$ \\ 
 \hline$0.1$ & $135$ & $139$ & $140$ & $140$ & $144$ & $145$\\ 
$\pi^2\;\;\;$ $0.2$ & $136$ & $136$ & $137$ & $139$ & $141$ & $141$\\ 
$0.3$ & $135$ & $135$ & $136$ & $137$ & $136$ & $135$\\ 
$0.4$ & $133$ & $133$ & $133$ & $133$ & $131$ & $128$\\ 
 \hline 
 \end{tabular}
  \end{subtable}
  \caption{OLS versus TSLS Example: Asymptotic coverage and expected relative width of two-step confidence interval with $\alpha_1 = \alpha/4,  \alpha_2 = 3\alpha/4$.}
  \label{tab:Limit2StepWideTauOLSvsIV}
\end{table}

\begin{table}[h]
  \centering
  \begin{subtable}{0.48\textwidth}
    \caption{Coverage Probability}
    \begin{tabular}{r|rrrrrr}
\hline\hline
 &\multicolumn{6}{c}{$\tau$} \\ 
 $\alpha = 0.05$ & $0$ & $1$ & $2$ & $3$ & $4$ & $5$ \\ 
 \hline$0.1$ & $98$ & $98$ & $97$ & $96$ & $96$ & $97$\\ 
$\gamma^2\;\;\;$ $0.2$ & $98$ & $98$ & $98$ & $97$ & $96$ & $96$\\ 
$0.3$ & $98$ & $98$ & $98$ & $97$ & $97$ & $96$\\ 
$0.4$ & $97$ & $97$ & $98$ & $98$ & $97$ & $97$\\ 
 \hline 
 \end{tabular}
 
 \vspace{2em} 
 
\begin{tabular}{r|rrrrrr}
\hline\hline
 &\multicolumn{6}{c}{$\tau$} \\ 
 $\alpha = 0.1$ & $0$ & $1$ & $2$ & $3$ & $4$ & $5$ \\ 
 \hline$0.1$ & $96$ & $96$ & $94$ & $93$ & $93$ & $94$\\ 
$\gamma^2\;\;\;$ $0.2$ & $96$ & $96$ & $95$ & $94$ & $93$ & $93$\\ 
$0.3$ & $96$ & $96$ & $95$ & $95$ & $93$ & $92$\\ 
$0.4$ & $95$ & $96$ & $96$ & $95$ & $94$ & $93$\\ 
 \hline 
 \end{tabular}
 
 \vspace{2em} 
 
\begin{tabular}{r|rrrrrr}
\hline\hline
 &\multicolumn{6}{c}{$\tau$} \\ 
 $\alpha = 0.2$ & $0$ & $1$ & $2$ & $3$ & $4$ & $5$ \\ 
 \hline$0.1$ & $93$ & $91$ & $87$ & $85$ & $87$ & $88$\\ 
$\gamma^2\;\;\;$ $0.2$ & $93$ & $92$ & $89$ & $86$ & $85$ & $87$\\ 
$0.3$ & $93$ & $92$ & $90$ & $88$ & $85$ & $85$\\ 
$0.4$ & $93$ & $92$ & $91$ & $89$ & $87$ & $85$\\ 
 \hline 
 \end{tabular}
  \end{subtable}
  ~
  \begin{subtable}{0.48\textwidth}
    \caption{Relative Width}
    \begin{tabular}{r|rrrrrr}
\hline\hline
 &\multicolumn{6}{c}{$\tau$} \\ 
 $\alpha = 0.05$ & $0$ & $1$ & $2$ & $3$ & $4$ & $5$ \\ 
 \hline$0.1$ & $117$ & $117$ & $118$ & $118$ & $118$ & $118$\\ 
$\gamma^2\;\;\;$ $0.2$ & $117$ & $117$ & $119$ & $121$ & $121$ & $122$\\ 
$0.3$ & $117$ & $117$ & $119$ & $122$ & $123$ & $124$\\ 
$0.4$ & $116$ & $116$ & $119$ & $122$ & $124$ & $125$\\ 
 \hline 
 \end{tabular}
 
 \vspace{2em} 
 
\begin{tabular}{r|rrrrrr}
\hline\hline
 &\multicolumn{6}{c}{$\tau$} \\ 
 $\alpha = 0.1$ & $0$ & $1$ & $2$ & $3$ & $4$ & $5$ \\ 
 \hline$0.1$ & $122$ & $122$ & $122$ & $122$ & $121$ & $121$\\ 
$\gamma^2\;\;\;$ $0.2$ & $123$ & $124$ & $125$ & $126$ & $126$ & $126$\\ 
$0.3$ & $123$ & $123$ & $125$ & $128$ & $128$ & $129$\\ 
$0.4$ & $122$ & $123$ & $126$ & $128$ & $130$ & $131$\\ 
 \hline 
 \end{tabular}
 
 \vspace{2em} 
 
\begin{tabular}{r|rrrrrr}
\hline\hline
 &\multicolumn{6}{c}{$\tau$} \\ 
 $\alpha = 0.2$ & $0$ & $1$ & $2$ & $3$ & $4$ & $5$ \\ 
 \hline$0.1$ & $131$ & $130$ & $129$ & $129$ & $128$ & $127$\\ 
$\gamma^2\;\;\;$ $0.2$ & $134$ & $134$ & $134$ & $134$ & $134$ & $134$\\ 
$0.3$ & $135$ & $135$ & $136$ & $137$ & $138$ & $138$\\ 
$0.4$ & $136$ & $136$ & $138$ & $138$ & $140$ & $140$\\ 
 \hline 
 \end{tabular}
  \end{subtable}
  \caption{Choosing IVs Example: Asymptotic coverage and expected relative width of two-step confidence interval with $\alpha_1 = \alpha/4,  \alpha_2 = 3\alpha/4$.}
  \label{tab:Limit2StepWideTauChooseIVs}
\end{table}

\section{Supplementary Simulation Results}
\label{sec:simsupplement}
This section discusses additional simulation results for the OLS versus IV example and the choosing instrumental variables example, as a supplement to those given in Sections \ref{sec:OLSvsIVsim}--\ref{sec:CIsim} of the paper.

\subsection{Downward J-Test}
\label{sec:downwardJ}
This appendix presents simulation results for the downward $J$-test -- an informal moment selection method that is fairly common in applied work -- for the choosing instrumental variables example from Section \ref{sec:chooseIVsim}.
In this simulation design the downward $J$-test amounts to simply using the full estimator unless it is rejected by a $J$-test.
Table \ref{fig:chooseIVsim_RMSErelJ} compares the RMSE of the post-FMSC estimator to that of the downward $J$-test with $\alpha = 0.1$ (J90), and $\alpha = 0.05$ (J95).
For robustness, I calculate the $J$-test statistic using a centered covariance matrix estimator, as in the FMSC formulas from section \ref{sec:chooseIVexample}.
Unlike the FMSC, the downward $J$-test is very badly behaved for small sample sizes, particularly for the smaller values of $\gamma$.
For larger sample sizes, the relative performance of the FMSC and the $J$-test is quite similar to what we saw in Figure \ref{fig:OLSvsIV_RMSEbaseline} for the OLS versus TSLS example: the $J$-test performs best for the smallest values of $\rho$, the FMSC performs best for moderate values, and the two procedures perform similarly for large values.
\begin{figure}
\centering
	\input{./SimulationChooseIVs/Results/RMSE_coarse_gamma_rel_J.tex}
	\caption{RMSE values for the post-Focused Moment Selection Criterion (FMSC) estimator and the downward $J$-test estimator with $\alpha = 0.1$ (J90) and $\alpha = 0.05$ (J95) based on 20,000 simulation draws from the DGP given in Equations \ref{eq:chooseIVDGP1}--\ref{eq:chooseIVDGP3} using the formulas from Sections \ref{sec:chooseIVexample}.}
	\label{fig:chooseIVsim_RMSErelJ}
\end{figure}
These results are broadly similar to those for the GMM moment selection criteria of \cite{Andrews1999} considered in Section \ref{sec:chooseIVsim}, which should not come as a surprise since the J-test statistic is an ingredient in the construction of the GMM-AIC, BIC and HQ. 

\subsection{Canonical Correlations Information Criterion}
\label{sec:CCIC}
Because the GMM moment selection criteria suggested by \cite{Andrews1999} consider only instrument exogeneity, not relevance, \cite{HallPeixe2003} suggest combining them with their canonical correlations information criterion (CCIC), which aims to detect and eliminate ``redundant instruments.''
Including such instruments, which add no information beyond that already contained in the other instruments, can lead to poor finite-sample performance in spite of the fact that the first-order limit distribution is unchanged.
For the choosing instrumental variables simulation example, presented in Section \ref{sec:chooseIVsim}, the CCIC takes the following simple form
	\begin{equation}
	\mbox{CCIC}(S) = n \log\left[1 - R_n^2(S) \right] + h(p + |S|)\kappa_n
	\end{equation}
where $R_n^2(S)$ is the first-stage $R^2$ based on instrument set $S$ and $h(p + |S|)\kappa_n$ is a penalty term \citep{Jana2005}. 
Instruments are chosen to \emph{minimize} this criterion.
If we define $h(p + |S|) = (p + |S| - r)$, setting $\kappa_n = \log{n}$ gives the CCIC-BIC, while $\kappa_n = 2.01 \log{\log{n}}$ gives the CCIC-HQ and $\kappa_n = 2$ gives the CCIC-AIC.
By combining the CCIC with an Andrews-type criterion, \cite{HallPeixe2003} propose to first eliminate invalid instruments and then redundant ones.
A combined GMM-BIC/CCIC-BIC criterion for the simulation example from section \ref{sec:chooseIVsim} uses the valid estimator unless both the GMM-BIC \emph{and} CCIC-BIC select the full estimator.
Combined HQ and AIC-type procedures can be defined analogously.
In the simulation design from this paper, however, \emph{each} of these combined criteria gives results that are practically identical to those of the valid estimator.
This hold true across all parameter values and sample sizes.
Full details are available upon request.

\subsection{Simulation Results for the 2-Step Confidence Interval}
\label{append:conf_sim}

This appendix presents results for the 2-Step confidence interval in the simulation experiment from Section \ref{sec:CIsim}.
Tables \ref{tab:CISim100_2stepWideTau_OLSvsIV} and \ref{tab:CISim100_2stepWideTau_ChooseIVs} 
present coverage probabilities and average relative width of the two-step confidence interval procedure with $\alpha_1 = \alpha/4$ and $\alpha_2 = 3\alpha/4$, the finite sample analogues to Tables \ref{tab:Limit2StepWideTauOLSvsIV} and \ref{tab:Limit2StepWideTauChooseIVs}. 
Results for other configurations of $\alpha_1, \alpha_2$, available upon request, result in even wider intervals.

\begin{table}[h]
  \centering
  \begin{subtable}{0.48\textwidth}
    \caption{Coverage Probability}
    \begin{tabular}{r|rrrrrr}
\hline\hline
 &\multicolumn{6}{c}{$\rho$} \\ 
 $\alpha = 0.05$ & $0$ & $0.1$ & $0.2$ & $0.3$ & $0.4$ & $0.5$ \\ 
 \hline$0.1$ & $98$ & $99$ & $99$ & $98$ & $95$ & $90$\\ 
$\pi^2\;\;\;$ $0.2$ & $97$ & $99$ & $99$ & $98$ & $94$ & $94$\\ 
$0.3$ & $98$ & $98$ & $98$ & $96$ & $95$ & $98$\\ 
$0.4$ & $97$ & $98$ & $97$ & $94$ & $96$ & $98$\\ 
 \hline 
 \end{tabular}
 
 \vspace{2em} 
 
\begin{tabular}{r|rrrrrr}
\hline\hline
 &\multicolumn{6}{c}{$\rho$} \\ 
 $\alpha = 0.1$ & $0$ & $0.1$ & $0.2$ & $0.3$ & $0.4$ & $0.5$ \\ 
 \hline$0.1$ & $97$ & $97$ & $98$ & $97$ & $92$ & $88$\\ 
$\pi^2\;\;\;$ $0.2$ & $95$ & $96$ & $95$ & $92$ & $90$ & $92$\\ 
$0.3$ & $95$ & $96$ & $95$ & $91$ & $94$ & $96$\\ 
$0.4$ & $95$ & $95$ & $92$ & $93$ & $95$ & $95$\\ 
 \hline 
 \end{tabular}
 
 \vspace{2em} 
 
\begin{tabular}{r|rrrrrr}
\hline\hline
 &\multicolumn{6}{c}{$\rho$} \\ 
 $\alpha = 0.2$ & $0$ & $0.1$ & $0.2$ & $0.3$ & $0.4$ & $0.5$ \\ 
 \hline$0.1$ & $92$ & $93$ & $93$ & $92$ & $86$ & $83$\\ 
$\pi^2\;\;\;$ $0.2$ & $93$ & $92$ & $89$ & $85$ & $85$ & $89$\\ 
$0.3$ & $91$ & $92$ & $87$ & $85$ & $88$ & $91$\\ 
$0.4$ & $92$ & $89$ & $84$ & $87$ & $90$ & $90$\\ 
 \hline 
 \end{tabular}
  \end{subtable}
  ~
  \begin{subtable}{0.48\textwidth}
    \caption{Average Relative Width}
    \begin{tabular}{r|rrrrrr}
\hline\hline
 &\multicolumn{6}{c}{$\rho$} \\ 
 $\alpha = 0.05$ & $0$ & $0.1$ & $0.2$ & $0.3$ & $0.4$ & $0.5$ \\ 
 \hline$0.1$ & $113$ & $114$ & $113$ & $119$ & $121$ & $124$\\ 
$\pi^2\;\;\;$ $0.2$ & $115$ & $117$ & $120$ & $123$ & $125$ & $126$\\ 
$0.3$ & $117$ & $117$ & $121$ & $122$ & $123$ & $124$\\ 
$0.4$ & $117$ & $118$ & $120$ & $121$ & $121$ & $121$\\ 
 \hline 
 \end{tabular}
 
 \vspace{2em} 
 
\begin{tabular}{r|rrrrrr}
\hline\hline
 &\multicolumn{6}{c}{$\rho$} \\ 
 $\alpha = 0.1$ & $0$ & $0.1$ & $0.2$ & $0.3$ & $0.4$ & $0.5$ \\ 
 \hline$0.1$ & $121$ & $123$ & $124$ & $127$ & $128$ & $133$\\ 
$\pi^2\;\;\;$ $0.2$ & $123$ & $125$ & $126$ & $129$ & $131$ & $132$\\ 
$0.3$ & $122$ & $123$ & $126$ & $128$ & $128$ & $128$\\ 
$0.4$ & $122$ & $124$ & $124$ & $125$ & $125$ & $125$\\ 
 \hline 
 \end{tabular}
 
 \vspace{2em} 
 
\begin{tabular}{r|rrrrrr}
\hline\hline
 &\multicolumn{6}{c}{$\rho$} \\ 
 $\alpha = 0.2$ & $0$ & $0.1$ & $0.2$ & $0.3$ & $0.4$ & $0.5$ \\ 
 \hline$0.1$ & $138$ & $139$ & $137$ & $142$ & $144$ & $146$\\ 
$\pi^2\;\;\;$ $0.2$ & $136$ & $137$ & $138$ & $140$ & $142$ & $142$\\ 
$0.3$ & $135$ & $135$ & $136$ & $137$ & $137$ & $137$\\ 
$0.4$ & $133$ & $133$ & $133$ & $133$ & $133$ & $132$\\ 
 \hline 
 \end{tabular}
  \end{subtable}
  \caption{2-step CI, $\alpha_1 = \alpha/4,\alpha_2 = 3\alpha/4$, OLS vs IV Example, $N=100$}
  \label{tab:CISim100_2stepWideTau_OLSvsIV}
\end{table}

\begin{table}[h]
  \centering
  \begin{subtable}{0.48\textwidth}
    \caption{Coverage Probability}
    \begin{tabular}{r|rrrrrr}
\hline\hline
 &\multicolumn{6}{c}{$\rho$} \\ 
 $\alpha = 0.05$ & $0$ & $0.1$ & $0.2$ & $0.3$ & $0.4$ & $0.5$ \\ 
 \hline$0.1$ & $96$ & $95$ & $94$ & $94$ & $95$ & $96$\\ 
$\gamma^2\;\;\;$ $0.2$ & $95$ & $95$ & $94$ & $93$ & $93$ & $97$\\ 
$0.3$ & $94$ & $97$ & $94$ & $94$ & $94$ & $96$\\ 
$0.4$ & $95$ & $96$ & $95$ & $93$ & $94$ & $94$\\ 
 \hline 
 \end{tabular}
 
 \vspace{2em} 
 
\begin{tabular}{r|rrrrrr}
\hline\hline
 &\multicolumn{6}{c}{$\rho$} \\ 
 $\alpha = 0.1$ & $0$ & $0.1$ & $0.2$ & $0.3$ & $0.4$ & $0.5$ \\ 
 \hline$0.1$ & $92$ & $90$ & $90$ & $89$ & $93$ & $95$\\ 
$\gamma^2\;\;\;$ $0.2$ & $92$ & $94$ & $91$ & $90$ & $92$ & $93$\\ 
$0.3$ & $93$ & $93$ & $93$ & $90$ & $90$ & $93$\\ 
$0.4$ & $90$ & $94$ & $93$ & $91$ & $87$ & $90$\\ 
 \hline 
 \end{tabular}
 
 \vspace{2em} 
 
\begin{tabular}{r|rrrrrr}
\hline\hline
 &\multicolumn{6}{c}{$\rho$} \\ 
 $\alpha = 0.2$ & $0$ & $0.1$ & $0.2$ & $0.3$ & $0.4$ & $0.5$ \\ 
 \hline$0.1$ & $88$ & $87$ & $83$ & $82$ & $88$ & $90$\\ 
$\gamma^2\;\;\;$ $0.2$ & $91$ & $88$ & $86$ & $85$ & $87$ & $89$\\ 
$0.3$ & $87$ & $88$ & $87$ & $84$ & $86$ & $89$\\ 
$0.4$ & $88$ & $91$ & $88$ & $84$ & $82$ & $88$\\ 
 \hline 
 \end{tabular}
  \end{subtable}
  ~
  \begin{subtable}{0.48\textwidth}
    \caption{Average Relative Width}
    \begin{tabular}{r|rrrrrr}
\hline\hline
 &\multicolumn{6}{c}{$\rho$} \\ 
 $\alpha = 0.05$ & $0$ & $0.1$ & $0.2$ & $0.3$ & $0.4$ & $0.5$ \\ 
 \hline$0.1$ & $116$ & $117$ & $118$ & $118$ & $118$ & $118$\\ 
$\gamma^2\;\;\;$ $0.2$ & $116$ & $117$ & $120$ & $121$ & $121$ & $122$\\ 
$0.3$ & $115$ & $116$ & $119$ & $121$ & $123$ & $124$\\ 
$0.4$ & $114$ & $115$ & $119$ & $121$ & $124$ & $125$\\ 
 \hline 
 \end{tabular}
 
 \vspace{2em} 
 
\begin{tabular}{r|rrrrrr}
\hline\hline
 &\multicolumn{6}{c}{$\rho$} \\ 
 $\alpha = 0.1$ & $0$ & $0.1$ & $0.2$ & $0.3$ & $0.4$ & $0.5$ \\ 
 \hline$0.1$ & $121$ & $121$ & $122$ & $122$ & $122$ & $122$\\ 
$\gamma^2\;\;\;$ $0.2$ & $122$ & $123$ & $125$ & $126$ & $127$ & $125$\\ 
$0.3$ & $122$ & $123$ & $126$ & $127$ & $128$ & $129$\\ 
$0.4$ & $122$ & $123$ & $126$ & $128$ & $130$ & $131$\\ 
 \hline 
 \end{tabular}
 
 \vspace{2em} 
 
\begin{tabular}{r|rrrrrr}
\hline\hline
 &\multicolumn{6}{c}{$\rho$} \\ 
 $\alpha = 0.2$ & $0$ & $0.1$ & $0.2$ & $0.3$ & $0.4$ & $0.5$ \\ 
 \hline$0.1$ & $131$ & $131$ & $130$ & $130$ & $131$ & $129$\\ 
$\gamma^2\;\;\;$ $0.2$ & $134$ & $134$ & $135$ & $136$ & $136$ & $135$\\ 
$0.3$ & $135$ & $136$ & $137$ & $138$ & $139$ & $139$\\ 
$0.4$ & $135$ & $137$ & $139$ & $140$ & $140$ & $141$\\ 
 \hline 
 \end{tabular}
  \end{subtable}
  \caption{2-step CI, $\alpha_1 = \alpha/4,  \alpha_2 = 3\alpha/4$, Choosing IVs Example, $N=100$}
  \label{tab:CISim100_2stepWideTau_ChooseIVs}
\end{table}

\newpage

\subsection{Weak Instruments} 
\label{sec:appendWeak}
The FMSC is derived under an asymptotic sequence that assumes strong identification.
But what if this assumption fails? 
The following simulation results provide a partial answer to this question by extending the RMSE comparisons from Sections \ref{sec:OLSvsIVsim} and \ref{sec:chooseIVsim} to the case in which the ``valid'' estimator suffers from a weak instruments problem.

Figures \ref{fig:OLSvsIV_RMSEbaseline_weak} and \ref{fig:OLSvsIV_AVG_weak} present further results for the OLS versus IV example from Section \ref{sec:OLSvsIVsim} with $\pi \in \left\{0.1, 0.05, 0.01\right\}$.
When $\pi = 0.01$ the TSLS estimator suffers from a severe weak instrument problem.
All other parameters values are identical to those in the corresponding figures from the body of the paper (Figures \ref{fig:OLSvsIV_RMSEbaseline} and \ref{fig:OLSvsIV_AVG}).
We see from Figure \ref{fig:OLSvsIV_RMSEbaseline_weak} that the post-FMSC estimator dramatically outperforms the TSLS estimator in the presence of a weak instrument.
Indeed, the RMSE curves for the these two estimators only cross in the bottom right panel where $\pi = 0.1$ and $N = 500$. 
Turning our attention to Figure \ref{fig:OLSvsIV_AVG_weak}, the minimum-AMSE averaging estimator provides a uniform improvement over the post-FMSC estimator although the advantage is relatively small unless $\pi = 0.1$ and $N=500$. 
Moreover, the DHW test with $\alpha = 0.05$ performs extremely well unless $\rho$ is large.
This is because, by construction, it is more likely to choose OLS than the other methods -- the correct decision if the instrument is sufficiently weak.

\begin{figure}[h]
\centering
	\input{./WeakOLSvsIV/Results/RMSE_coarse_pi_baseline.tex}
	\caption{RMSE values for the two-stage least squares (TSLS) estimator, the ordinary least squares (OLS) estimator, and the post-Focused Moment Selection Criterion (FMSC) estimator based on 10,000 simulation draws from the DGP given in Equations \ref{eq:OLSvsIVDGP1}--\ref{eq:OLSvsIVDGP3} using the formulas from Section \ref{sec:OLSvsIVExample}.}
	\label{fig:OLSvsIV_RMSEbaseline_weak}
\end{figure}

\begin{figure}[h]
\centering
	\input{./WeakOLSvsIV/Results/RMSE_coarse_pi_relative_all.tex}
	\caption{RMSE values for the post-Focused Moment Selection Criterion (FMSC) estimator, Durbin-Hausman-Wu pre-test estimators with $\alpha = 0.1$ (DWH90) and $\alpha = 0.05$ (DHW95), and the minmum-AMSE averaging estimator, based on 10,000 simulation draws from the DGP given in Equations \ref{eq:OLSvsIVDGP1}--\ref{eq:OLSvsIVDGP3} using the formulas from Sections \ref{sec:OLSvsIVExample} and \ref{sec:momentavgexample}.}
	\label{fig:OLSvsIV_AVG_weak}
\end{figure}

Figures \ref{fig:chooseIVsim_RMSErelMSC_weak} and \ref{fig:chooseIVsim_RMSEbaseline_weak} present RMSE comparisons for a slightly more general version of the simulation experiment from Section \ref{sec:chooseIVsim} in which the strength of the valid instruments can vary according to a scalar parameter $\pi$, specifically
\begin{eqnarray}
		y_i &=& 0.5 x_i + \epsilon_i\\ 
		\label{eq:chooseIVDGP1_weak}
		x_i &=& \pi (z_{1i} + z_{2i} + z_{3i}) + \gamma w_i + v_i 
		\label{eq:chooseIVDGP2_weak}
	\end{eqnarray}
for $i=1, 2, \hdots, N$ where $(\epsilon_i, v_i, w_i, z_{i1}, z_{2i}, z_{3i})' \sim \mbox{ iid  } N(0,\mathcal{V})$ with
\begin{equation}	
	\mathcal{V} = \left[  \begin{array}
		{cc} \mathcal{V}_1 & 0 \\ 0 & \mathcal{V}_2
	\end{array}\right], \quad
	\mathcal{V}_1 = \left[ \begin{array}
		{ccc} 
		1 & (0.5 - \gamma \rho) & \rho \\
		(0.5 - \gamma \rho) & (1 - \pi^2 - \gamma^2) & 0\\ 
		\rho & 0 & 1 \\ 
	\end{array} \right], \quad \mathcal{V}_2 = I_3 / 3
	\label{eq:chooseIVDGP3_weak}
\end{equation}
As in Section \ref{sec:chooseIVsim}, this setup keeps the variance of $x$ fixed at one and the endogeneity of $x$, $Cor(x, \epsilon)$, fixed at $0.5$ while allowing the relevance, $\gamma = Cor(x,w)$, and endogeneity, $\rho = Cor(w, \epsilon)$, of the instrument $w$ to vary.
The instruments $z_1, z_2, z_3$ remain valid and exogenous and the meaning of the parameters $\rho$ and $\gamma$ is unchanged.
By varying $\pi$, however, the present design allows the strength of the first-stage to vary: the first-stage R-squared is $1 - \sigma_v^2 = \pi^2 + \gamma^2$.
Setting $\pi$ sufficiently small creates a weak instrument problem for the ``valid'' estimator that uses only $z_1, z_2$ and $z_3$ as instruments.
Figures \ref{fig:chooseIVsim_RMSEbaseline_weak} and \ref{fig:chooseIVsim_RMSErelMSC_weak} present results for $\pi = 0.01$.
The results are qualitatively similar to those of Figures \ref{fig:chooseIVsim_RMSEbaseline} and \ref{fig:chooseIVsim_RMSErelMSC} although somewhat starker.
When the valid estimator suffers from a weak instruments problem, the post-FMSC estimator in general dramatically outperforms both the valid estimator and the GMM moment selection criteria of \cite{Andrews1999}.
There are only two exceptions. 
First when $N = 500$ and $\gamma = 0.2$, the valid estimator outperforms FMSC for $\rho$ greater than $0.25$.
Second, when $N = 500$, GMM-BIC outperforms FMSC for the smallest values of $\rho$.

\begin{figure}[h]
\centering
	\input{./WeakChooseIVs/Results/RMSE_coarse_gamma_baseline.tex}
	\caption{RMSE values for the valid estimator, including only $(z_1, z_2, z_3)$, the full estimator, including $(z_1, z_2, z_3, w)$, and the post-Focused Moment Selection Criterion (FMSC) estimator based on 20,000 simulation draws from the DGP given in Equations \ref{eq:chooseIVDGP1_weak}--\ref{eq:chooseIVDGP3_weak} with $\pi = 0.01$ using the formulas from Section \ref{sec:chooseIVexample}.}
	\label{fig:chooseIVsim_RMSEbaseline_weak}
\end{figure}

\begin{figure}[h]
\centering
	\input{./WeakChooseIVs/Results/RMSE_coarse_gamma_rel_MSC.tex}
	\caption{RMSE values for the post-Focused Moment Selection Criterion (FMSC) estimator and the GMM-BIC, HQ, and AIC estimators based on 20,000 simulation draws from the DGP given in Equations \ref{eq:chooseIVDGP1_weak}--\ref{eq:chooseIVDGP3_weak} with $\pi = 0.01$ using the formulas from Section \ref{sec:chooseIVexample}.}
	\label{fig:chooseIVsim_RMSErelMSC_weak}
\end{figure}

\end{document}